\documentclass[12pt]{article}

\usepackage{amsthm,amsmath,amsfonts,amssymb}
\allowdisplaybreaks
\usepackage[authoryear]{natbib}
\usepackage[colorlinks,citecolor=blue,urlcolor=blue]{hyperref}
\usepackage{graphicx}

\usepackage[utf8]{inputenc} 
\usepackage[T1]{fontenc}    
\usepackage[subscriptcorrection]{newtxmath}
\usepackage[plain,noend]{algorithm2e}
\usepackage{setspace} 

\usepackage[dvipsnames]{xcolor} 
\usepackage{url, booktabs, nicefrac, caption, subcaption, cprotect, enumitem, bbm, tikz, adjustbox, wrapfig, xr, changepage, datetime, comment, physics}
\let\dv\undefined

\let\op\undefined
\let\vert\undefined

\let\pmat\undefined
\usetikzlibrary{arrows, automata, calc, positioning}
\usepackage[multiple]{footmisc}

\usepackage{float, mathtools, bm, thmtools, multirow, etoolbox, xpatch, anyfontsize, pdflscape}

\appto{\normalsize}{%
  \setlength{\abovedisplayskip}{5pt plus 2pt minus 4pt}%
  \setlength{\belowdisplayskip}{5pt plus 2pt minus 4pt}%
  \setlength{\abovedisplayshortskip}{0pt plus 2pt}%
  \setlength{\belowdisplayshortskip}{3pt plus 2pt minus 2pt}%
  \setlength{\jot}{1pt}
}

\input{_commands}

\newcommand{\blind}{1}
\def\spacingset#1{\renewcommand{\baselinestretch}%
{#1}\small\normalsize} \spacingset{1}

\begin{document}

\if1\blind
{
  \title{\bf Placebo Discontinuity Design}
  \author{Rahul Singh\thanks{
    The authors thank Joshua Angrist for helpful discussions.}\hspace{.2cm}\\
    Harvard Society of Fellows and Department of Economics\\
   \and 
    Moses Stewart \\
    Harvard Department of Economics}

     \date{Original draft: July 2025. This draft: July 2026.}
  \maketitle
} \fi

\if0\blind
{
  \bigskip
  \bigskip
  \bigskip
  \begin{center}
    {\LARGE\bf Placebo Discontinuity Design}
\end{center}
  \medskip
} \fi

\bigskip
\begin{abstract}
\noindent
Standard regression discontinuity design (RDD) models rely on the continuity of expected potential outcomes at the cutoff. 
The standard continuity assumption can be violated by strategic manipulation of the running variable, which is realistic when the cutoff is widely known and when the treatment of interest is a social program or government benefit. 
In this work, we identify the treatment effect despite  such  a violation, by leveraging a placebo treatment and a placebo outcome. 
We introduce a local instrumental variable estimator. 
Our estimator decomposes into two terms: the standard RDD estimator of the target outcome's discontinuity, and a new adjustment term based on the placebo outcome's discontinuity. 
We show that our estimator is consistent, and we justify a robust bias-corrected inference procedure. 
Our method expands the applicability of RDD to settings with strategic behavior around the cutoff, which commonly arise in social science.

\end{abstract}

\noindent%
\vfill

\newpage
\spacingset{1.8} 

\section{Introduction}

Regression discontinuity design (RDD) strategies are widely used to estimate the effect of a treatment on an outcome of interest, when the probability of treatment assignment exhibits a discontinuous jump at a known threshold \citep{thistlethwaite1960regression}. Applications of RDD are broad and include studies in electoral politics \citep{hall_what_2015,spenkuch_political_2018}, education \citep{cook_birthdays_2016,goodman_can_2019}, housing markets \citep{kumar_restrictions_2018}, and market design \citep{abdulkadiroglu_regression_2017,kong_sequential_2021}, among others. 

The main identifying assumption in standard RDD is the continuity of expected potential outcomes at the cutoff \citep{hahn2001identification}. This assumption ensures that, in the absence of treatment, units just above and below the cutoff are comparable. As discussed in Remark~\ref{rmk:local_random}, this condition is closely related to the requirement that the conditional distribution of unobserved confounding given the running variable 
is continuous at the cutoff 
\citep{lee2010regression}. To empirically assess the plausibility of this assumption, researchers often conduct falsification tests using placebo variables.\footnote{The literature also considers predetermined covariates, of which the placebo treatment may be viewed as a special case.} For example, \citet{imbens_regression_2008} and \citet{cattaneo_practical_2024} recommend checking for discontinuities in placebo outcomes at the cutoff as a diagnostic for invalid RDD. However, when these falsification tests indicate a failure of the RDD continuity assumption, current methods offer no path forward for recovering the treatment effect using these falsification test results.

In this paper, we propose a new identification strategy that leverages a placebo treatment and a placebo outcome to recover the treatment effect even when the continuity of expected potential outcomes may fail. We introduce a local linear instrumental variable estimator that calculates the treatment effect as the difference between the right-limit of the observed outcome and an adjusted left-limit that accounts for unobserved confounding via placebo variables. The estimator decomposes into a standard RDD estimate for the target outcome, $\hat{\tau}_{\text{rdd}}^{y}$, minus an adjustment term $(\hat{\tau}_{\text{rdd}}^{w})^{\top}\hat{\gamma}_{-}$, where $\hat{\tau}_{\text{rdd}}^{w}$ is the discontinuity in the placebo outcome and $\hat{\gamma}_{-}$ is an appropriate weight.

Conceptually, we build on the insights and techniques of the placebo variable literature (also called the negative control literature) \citep{miao2018identifying,deaner2018nonparametric,tchetgen2020introduction}, to answer a salient question about RDD, which is an extremely popular method in social science. Whereas that literature uses placebo variables to adjust for unobserved confounders of random treatments, we use placebo variables to adjust for sorting and strategic manipulation with possibly nonrandom treatment rules. To the best of our knowledge, our identification result appears to be novel and motivates a new estimation problem. For this new estimation problem, we propose and analyze what appears to be a novel variation of the widely used local linear RDD estimator \citep{fan_variable_1992}, and we derive bias-corrected inference \citep{calonico_robust_2014,calonico_regression_2019}. Future work may develop bias-aware inference  \citep{armstrong2018optimal,imbens2019optimized,noack2024bias}. 

We contribute to a literature that studies estimation of treatment effects at a cutoff when the standard RDD continuity assumption fails. For a French regulation that binds on firms exceeding 50 employees, \cite{garicano_firm_2016} recover the effect of crossing the threshold by imposing a parametric model of firm production; we instead recover the RDD treatment effect nonparametrically. \cite{eckles2020noise} provide identification and inference results for the RDD treatment effect without the standard continuity assumption, in settings where the researcher \textit{knows} the conditional distribution of the running variable given the unobserved variables, $f(d \mid u, \eta_{y})$. We instead study settings in which this distribution is unknown, where using standard RDD strategies would require justifying that there is no unobserved confounding around the cutoff. Most closely related, \cite{gerard_2020_bounds} bound the treatment effect when the running variable is manipulated around the cutoff. We complement their partial-identification approach by providing conditions under which the effect is point identified.

A related literature augments standard RDD inference using auxiliary plots and variables while maintaining the continuity of potential outcomes, and so rules out the unobserved confounding we study. A ``difference-in-discontinuities'' strategy combines RDD estimates across populations to isolate a treatment effect when several factors besides the treatment change at the threshold \citep{grembi_fiscal_rule_aej, galindo2018fuzzy, butts2023geographic, picchetti2024difference, leventer2024correcting}; this decoupling of co-occurring changes is sometimes called confounding, but it differs from the unobserved confounding around the cutoff that we address. When distinct cutoffs apply to different groups, \cite{bertanha2020regression} and \cite{cattaneo2021extrapolating} extend RDD identification to trace a dose-response curve in a neighborhood of the cutoffs. A third strand combines additional variables with auxiliary assumptions to extrapolate the RDD effect away from the cutoff \citep{angrist2015wanna, dong2015identifying, deaner2025extrapolation}. Like these papers, we use auxiliary structure (here placebo variables) together with additional assumptions to augment RDD inference; unlike them, we attain point identification and nonparametric inference precisely when the continuity assumption fails at the cutoff.

The remainder of the paper is organized as follows. Section~\ref{sec:example} illustrates the main idea through a real-world application. Section~\ref{sec:double_proxy} presents our identification assumptions and main result. Section~\ref{sec:estimation} derives our estimator. Section~\ref{sec:inference} establishes consistency and asymptotic normality with bias-corrected inference. 
Section~\ref{sec:conclusion} concludes. Formal proofs are deferred to the appendix.

\section{Example: Maimonides rule}\label{sec:example}

To motivate our approach, consider an example adapted from \citet{angrist_maimonides_2019}. Suppose we are interested in estimating the effect of transitioning from a large class (39-40 pupils) to a small class (20-21 pupils), denoted by the binary treatment $A$. The outcome of interest is student achievement $Y$. A Maimonides-style rule determines class size: schools must split classes when total enrollment in a grade level $D$ exceeds a statutory threshold $d^{*} = 40$. For instance, a school with $D=42$ students in a grade level receives funding for two classes of 21 students ($A=1$), while a school with exactly $D=40$ students in a grade level keeps one class of 40 students $(A=0)$.

\begin{figure}[htbp]
\caption{\label{fig:potential_outcome} Valid and invalid regression discontinuity design}
\begin{subfigure}{.4\textwidth}
\begin{centering}
\begin{tikzpicture}[domain=0:6]
    \draw[->] (-0.2,0) -- (6.2,0) node[right] {$d$};
    \draw[->] (0,-0.2) -- (0,6.2) node[above] {$Y$};
    \draw[color=BlueViolet, domain = 0:3, dashed, thick] plot (\x,{sin(\x r)/2 + 4});
    \draw[color=BlueViolet, domain = 3:6, thick] plot (\x,{sin(\x r)/2 + 4}) node[above, xshift=-2em] {\scalebox{0.95}{$\E{Y(1, d, U, \eta_{y}) \mid D = d}$}};
    \draw[color=Black, domain = 0:3, thick] plot (\x,{sin(\x r)/2 + 2});
    \draw[color=Black, domain = 3:6, dashed, thick] plot (\x,{sin(\x r)/2 + 2}) node[above, xshift=-2em] {\scalebox{0.95}{$\E{Y(0, d, U, \eta_{y}) \mid D = d}$}};
    \draw[dashed] (3, 0) node[below] {$d^{*}$} (3,-0.2) -- (3,6.2);
\end{tikzpicture}
\end{centering}
\end{subfigure}
\hspace{4em}
\begin{subfigure}{.4\textwidth}
\begin{centering}
\begin{tikzpicture}[domain=0:6]
    \draw[->] (-0.2,0) -- (6.2,0) node[right] {$d$};
    \draw[->] (0,-0.2) -- (0,6.2) node[above] {$Y$};
    \draw[color=BlueViolet, domain = 0:3, dashed, thick] plot (\x,{sin(\x r)/2 + 3.5});
    \draw[color=BlueViolet, domain = 3:6, thick] plot (\x,{sin(\x r)/2 + 4}) node[above, xshift=-2em] {\scalebox{0.95}{$\E{Y(1, d, U, \eta_{y}) \mid D = d}$}};
    \draw[color=Black, domain = 0:3, thick] plot (\x,{sin(\x r)/2 + 2});
    \draw[color=Black, domain = 3:6, dashed, thick] plot (\x,{sin(\x r)/2 + 2.5}) node[above, xshift=-2em] {\scalebox{0.95}{$\E{Y(0, d, U, \eta_{y}) \mid D = d}$}};
    \draw[dashed] (3, 0) node[below] {$d^{*}$} (3,-0.2) -- (3,6.2);
\end{tikzpicture}
\end{centering}
\end{subfigure}
\par
\raggedright
{\footnotesize Notes: Figure~\ref{fig:potential_outcome} illustrates the standard continuity assumption on the expected potential outcomes for a valid RDD (on the left) and an invalid RDD (on the right). The $Y$-axis gives the expected test score $\E{Y(a, d, U, \eta_{y}) \mid D = d}$ for a student in a class of 20-21 students $(A = 1)$ or 39-40 students $(A = 0)$ given $D=d$ enrolled students in the school. The $x$-axis gives the number of students enrolled in the school. The dotted regions of the line indicate the unobserved, counterfactual expected outcomes.}
\end{figure}
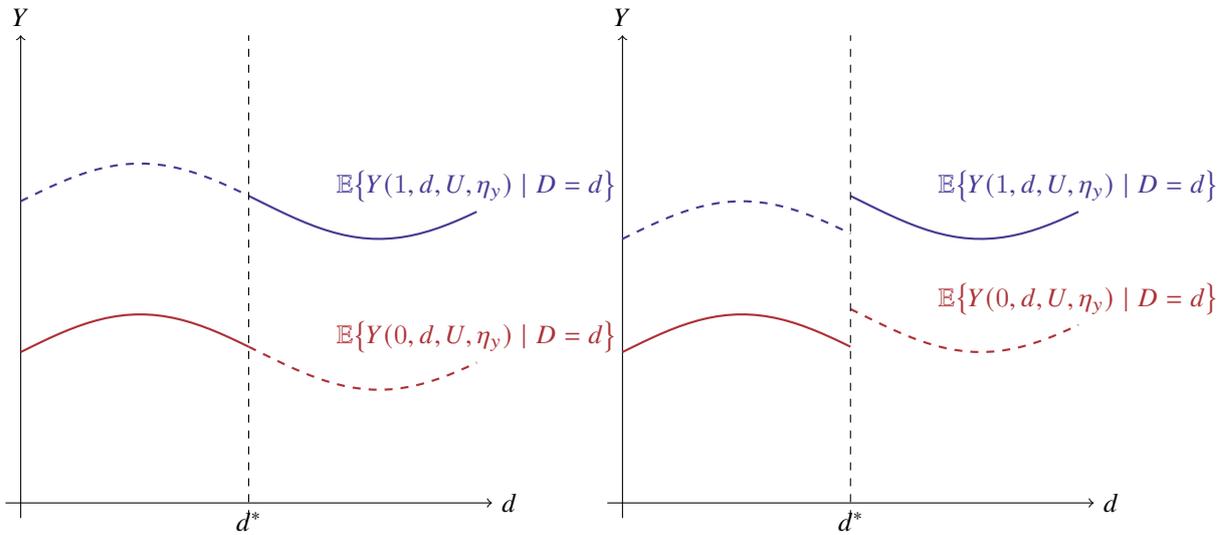

To apply a conventional RDD estimator in this setting, one must assume that the expected potential outcomes $\E{Y(1, d, U, \eta_{y}) \mid D = d}$ and $\E{Y(0, d, U, \eta_{y}) \mid D = d}$ are continuous at the cutoff $d^{*}$. Here, $Y(A, D, U, \eta_y)$ denotes the test score of a student in a class size of 20-21 students $(A = 1)$ or 39-40 students $(A = 0)$, given the total enrollment in the grade level  $(D)$. We model unobserved heterogeneity through two components---school administrator sophistication $U$, which may influence total enrollment near the cutoff, and student-level heterogeneity $\eta_y$, which is random. The left panel of Figure~\ref{fig:potential_outcome} depicts a valid RDD scenario under this assumption.

\begin{remark}[Continuity and unobserved confounding]\label{rmk:local_random}
    As noted by \citet{lee2010regression}, the standard RDD assumption, i.e. continuity of the expected potential outcomes at the cutoff, is implied by the following two conditions: for all $(u, \eta_y)$,
    (i) the conditional distributions $d\mapsto \pr{d \mid u, \eta_y}$ and $d \mapsto \pr{d}$ are continuous at $d^{*}$ ;
    and (ii) the potential outcome functions $d\mapsto Y(1, d, u, \eta_y)$ and $d \mapsto Y(0, d, u, \eta_y)$ are bounded and continuous at $d^*$.
    
    To see this, note that continuity of $\pr{d \mid u, \eta_y}$ and $\pr{d}$ implies continuity of $\pr{u, \eta_y \mid d}$ by Bayes’ rule.\footnote{This holds only at points where $\pr{d} > 0$.} Then, applying the dominated convergence theorem, we can exchange the limit and expectation, so that
    \begin{align*}
    \lim_{\epsilon \downarrow 0} \E{Y(1,d^{*} + \epsilon,U,\eta_{y}) \mid D = d^{*} + \epsilon} = \lim_{\epsilon \downarrow 0} \int Y(1, d^{*} + \epsilon, u, \eta_{y}) \dv \pr{u, \eta_{y} \mid d^{*} + \epsilon}\\
    = \int Y(1, d^{*}, U, \eta_{y}) \dv \pr{u, \eta_{y} \mid d^{*}}
    = \E{Y(1,d^{*},U,\eta_{y}) \mid D = d^{*}}.
    \end{align*}
    The same argument applies from the left. \hfill $\triangleleft$
\end{remark}

However, this assumption may fail if school administrators behave strategically. For instance, more sophisticated school leaders may recognize the budgetary and pedagogical advantages of exceeding the class size threshold, perhaps by enrolling just above the cutoff. These administrators may also tend to run more effective schools with higher-achieving students, thereby inducing a spurious upward jump in test scores around $d^*$ that is not due to class size per se \citep{angrist_maimonides_2019}. In such a case, the unobserved aptitude of school leaders $(U)$ confounds the relationship between enrollment $(D)$ and student achievement $(Y)$ near the cutoff. As a result, schools just above $d^*$ may perform better because they are led by more capable administrators, invalidating the continuity assumption that would allow researchers to evaluate the effect of class size. This confounded scenario is depicted in the right panel of Figure~\ref{fig:potential_outcome}.

Our identification strategy recovers the treatment effect in such confounded settings by positing that continuity holds after further conditioning on the unobserved school-level aptitude $U$. As shown in Figure~\ref{fig:confounded_outcomes}, we require that, conditional on both $D=d$ and $U=u$, the expected potential outcomes become continuous at the cutoff. We then exploit the statistical relationship between test scores in the prior year $W$, which serve as placebo outcomes, and the availability of transfer-eligible students based on birth-dates $Z$, which act as placebo treatments, to appropriately adjust for the unobserved confounder $U$.

\begin{figure}[htbp]
\begin{centering}
\caption{\label{fig:confounded_outcomes} Valid placebo discontinuity design}
\begin{subfigure}{.4\textwidth}
\begin{centering}
\begin{tikzpicture}[domain=0:6]
    \draw[->] (-0.2,0) -- (6.2,0) node[right] {$d$};
    \draw[->] (0,-0.2) -- (0,6.2) node[above] {$Y$};
    \draw[color=BlueViolet, domain = 0:3, dashed, thick] plot (\x,{sin(\x r)/2 + 4.5});
    \draw[color=BlueViolet, domain = 3:6, thick] plot (\x,{sin(\x r)/2 + 4.5}) node[above, xshift=-1.5em] {\scalebox{0.9}{$\E{Y(1, d, u_{1}, \eta_{y}) \mid D = d, U = u_{1}}$}};
    \draw[color=Black, domain = 0:3, thick] plot (\x,{sin(\x r)/2 + 3});
    \draw[color=Black, domain = 3:6, dashed, thick] plot (\x,{sin(\x r)/2 + 3}) node[above, xshift=-1.5em] {\scalebox{0.9}{$\E{Y(0, d, u_{1}, \eta_{y}) \mid D = d, U = u_{1}}$}};
    \draw[dashed] (3, 0) node[below] {$d^{*}$} (3,-0.2) -- (3,6.2);
\end{tikzpicture}
\end{centering}
\end{subfigure}
\hspace{4em}
\begin{subfigure}{.4\textwidth}
\begin{centering}
\begin{tikzpicture}[domain=0:6]
    \draw[->] (-0.2,0) -- (6.2,0) node[right] {$d$};
    \draw[->] (0,-0.2) -- (0,6.2) node[above] {$Y$};
    \draw[color=BlueViolet, domain = 0:3, dashed, thick] plot (\x,{sin(\x r)/2 + 3.5});
    \draw[color=BlueViolet, domain = 3:6, thick] plot (\x,{sin(\x r)/2 + 3.5}) node[above, xshift=-1.5em] {\scalebox{0.9}{$\E{Y(1, d, u_{2}, \eta_{y}) \mid D = d, U = u_{2}}$}};
    \draw[color=Black, domain = 0:3, thick] plot (\x,{sin(\x r)/2 + 2});
    \draw[color=Black, domain = 3:6, dashed, thick] plot (\x,{sin(\x r)/2 + 2}) node[above, xshift=-1.5em] {\scalebox{0.9}{$\E{Y(0, d, u_{2}, \eta_{y}) \mid D = d, U = u_{2}}$}};
    \draw[dashed] (3, 0) node[below] {$d^{*}$} (3,-0.2) -- (3,6.2);
\end{tikzpicture}
\end{centering}
\end{subfigure}
\par
\end{centering}
\raggedright
{\footnotesize Notes: Figure~\ref{fig:potential_outcome} illustrates the relaxed assumption we place on the expected potential outcomes. The $Y$-axis gives the test score $\E{Y(a, d, u, \eta_{y}) \mid D = d, U = u}$ for a student in a class of 20-21 students $(A = 1)$ or 39-40 students $(A = 0)$ enrolled in a school with $D=d$ pupils and administrator aptitude level $U=u$. The $x$-axis gives the number of students enrolled in the school. The dotted regions of the line indicate the unobserved, counterfactual outcome.}
\end{figure}

Intuitively, test scores in a previous year are unaffected by enrollment $(D)$ in the current year, making it a valid placebo outcome ($W)$. In this sense, previous test scores $(W)$ serve as a proxy for the school administrator aptitude ($U$).

Meanwhile, the availability of students able to transfer across grades should only affect student test scores $(Y)$ through its influence on total enrollment in the grade $(D)$, making it a valid placebo treatment $(Z)$. The availability of transfer students is like a local instrument, or a special type of pre-determined covariate, that can help to distill the variation due to school administrator aptitude.

\section{Identification: Placebo discontinuity}\label{sec:double_proxy}

In this section, we formalize the model, interpret the key assumptions, and present our main identification result: how to recover the treatment effect in an RDD setting with strategic behavior around the cutoff. Standard practice uses placebo variables to detect this issue; as our first  contribution, we use placebo variables to correct this issue.

\subsection{Notation}

Consider the nonseparable model
$$ Y=Y(A,D,U,\eta_y),$$
where $A$ is the treatment, $D$ is the running variable, $U$ is unobserved confounding, and $\eta_y$ is unobserved heterogeneity. Notice that we decompose the unobservable into the component $U$ for which we will have a model, and the component $\eta_y$ which will be exogenous conditional on $U$.

We will study the usual causal parameter in RDD, 
$$\tau_0= \E{Y(1,D,U,\eta_y)-Y(0,D,U,\eta_y)\mid  D = d^*},$$
which reflects the average causal effect of moving from untreated to treated at the cutoff $d^*$, conditional on the running variable.

In contrast with standard approaches, we allow for discontinuities in the conditional distribution of unobservables $f(u,\eta_y \mid d)$ at $d^*$, acknowledging that strategic behavior or sorting can lead to non-smooth densities near the cutoff. However, we require the density to remain strictly positive in a neighborhood of $d^*$ to maintain identifiability.

We employ placebo variables $(Z, W)$ to handle unobserved confounding. The full structural system is written as,
$$ Y = Y(A, D, Z, U, \eta_y), \quad W = W(A, D, Z, U, \eta_w), \quad A = A(D, Z, U, \eta_a). $$
We summarize the invariances of the structural system in the next section. A graphical summary of the causal structure is previewed in Figure~\ref{dag:nc}.
\medskip

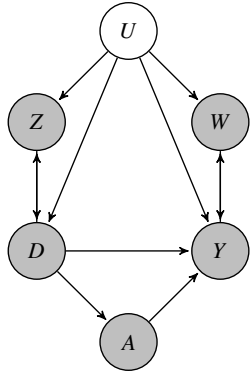
\begin{figure}[htbp]
\vspace{-10pt}
\begin{center}
\begin{adjustbox}{width=.2\textwidth}
\begin{tikzpicture}[->,>=stealth',shorten >=1pt,auto,node distance=2cm,
                    semithick]
  \tikzstyle{every state}=[draw=black,text=black]

  \node[state]         (u) [fill = white]                        {$U$};
  \node[state]         (w) [below right of=u, fill=lightgray]    {$W$};
  \node[state]         (z) [below left of=u, fill=lightgray]     {$Z$};
  \node[state]         (d) [below of=z, fill=lightgray]             {$D$};
  \node[state]         (y) [below of=w, fill=lightgray]             {$Y$};
  \node[state]         (a) [below right of=d, fill=lightgray]    {$A$};

  \path (u) edge           node {$ $} (z)
            edge           node {$ $} (w)
            edge           node {$ $} (d)
            edge           node {$ $} (y)
        (z) edge           node {$ $} (d)
        (w) edge           node {$ $} (y)
        (d) edge           node {$ $} (y)
            edge           node {$ $} (a)
            edge           node {$ $} (z)
        (a) edge           node {$ $} (y)
        (y) edge           node {$ $} (w);;
\end{tikzpicture}
\end{adjustbox}
\vspace{-5pt}
\caption{Placebo discontinuity design DAG}
\label{dag:nc}
\end{center}
\end{figure}

\subsection{Assumptions}\label{sec:assumptions}

The following assumptions identify $\tau_0$ even when the expected potential outcomes \hfill\break $d \mapsto \E{Y(0, D, U, \eta_y) \mid D = d}$ and $d\mapsto \E{Y(1, D, U, \eta_y) \mid D = d}$ are discontinuous at the cutoff $d=d^*$.

Throughout, we place assumptions in some neighborhood of the threshold $d^{*}$ defined as,
$$
 \mathcal{D}_-(\epsilon) = (d^* - \epsilon, d^*), \quad \mathcal{D}_+(\epsilon) = (d^*, d^* + \epsilon), \quad \mathcal{D}(\epsilon) = \mathcal{D}_-(\epsilon) \cup \mathcal{D}_+(\epsilon).
$$

\begin{assumption}[RDD]\label{assum:rdd}
    Assume that the following limits exist and are unequal:
    $$\lim_{\epsilon \downarrow 0} \left[ \E{A \mid D = d^{*} - \epsilon} \right],\quad \lim_{\epsilon \downarrow 0} \left[ \E{A \mid D = d^{*} + \epsilon} \right].$$
\end{assumption}

The existence of this limit, for the expectation of the treatment assignment $A$ given the running variable $D$, is  standard in the RDD identification literature \citep{hahn2001identification,lee2010regression}.

\begin{definition}[Sharp design]\label{def:sharp}
Under a ``sharp'' design, we assume that there is no heterogeneity in treatment assignment given the running variable ($\eta_a$ does not exist) and for all $\tilde{\epsilon} \in \mathcal{D}_{+}(\epsilon) - d^{*}$,
$$\E{A \mid D = d^{*} - \tilde{\epsilon}} = 0, \quad \E{A \mid D = d^{*} + \tilde{\epsilon}} = 1.$$
A ``fuzzy'' design is one that is not sharp, i.e. one that does not satisfy the equalities above.
\end{definition}

In a sharp design, $A$ is fully determined by $D$, implying no heterogeneity in treatment assignment (i.e., no $\eta_a$). In this case, Assumption~\ref{assum:rdd} is trivially satisfied.

We define $\mathcal{Z}(\epsilon)$ such that $\pr{Z \in \mathcal{Z}(\epsilon) \mid D \in \mathcal{D}(\epsilon)} = 1$.

\begin{assumption}[Placebo variables]\label{assum:proxy}
Assume there exists some $\epsilon>0$ such that for all $d\in \mathcal{D}(\epsilon)$ and for all $z\in \mathcal{Z}(\epsilon)$,
    \begin{enumerate}[label=\alph*]
        \item (Causal consistency): if $A=a$, $D=d$, and $Z=z$ then $Y=Y(a,d,z,U,\eta_y)$, and $W=W(a,d,z,U,\eta_w)$ almost surely. If $D=d$ and $Z=z$ then $A=A(d,z,U,\eta_a)$ almost surely.
        \item (Placebo selection on unobservables): $Z \indep \eta_{y}, \eta_{a} \mid  U, D$ and $\eta_{w} \indep D, Z \mid U$. 
        \item (Overlap): if $f(u)>0$ then $f(a,d,z\mid  u)>0$, where we assume the densities exist.
        \item (Exclusion): $Y(a,d,z,U,\eta_y)=Y(a,d,U,\eta_y)$, $W(a,d,z,U,\eta_w)=W(U,\eta_w)$, and $A(d, z, U, \eta_{a}) = A(d, \eta_{a})$ almost surely.
    \end{enumerate}
\end{assumption}

In summary, we have the following simplification local to the cutoff: 
$$Y = Y(A, D, U, \eta_y), \quad W = W(U, \eta_w), \quad A = A(D, \eta_a).$$

Causal consistency rules out network effects near the cutoff.

The main condition within placebo selection on observables is the implication that $Z\indep \eta_y, \eta_w |D,U$: conditional upon the running variable and unobserved confounder, the placebo treatment is as good as random for the target outcome and placebo outcome.\footnote{By weak union, $\eta_{w} \indep D, Z \mid U$ implies $\eta_{w} \indep Z \mid D,U$.} In the sharp design, $Z\indep \eta_a |D,U$ automatically holds; in the fuzzy design, it imposes that the placebo treatment is conditionally independent of the idiosyncratic aspect of treatment assignment. As shown in Appendix~\ref{sec:cont}, $Z\indep \eta_y, \eta_a,\eta_w |D,U$ together with an additional continuity condition suffice for our argument to work. In the main text, to simplify the exposition, we strengthen $\eta_{w} \indep Z \mid D, U$ to $\eta_{w} \indep D,Z \mid U$ and thereby avoid the additional continuity.

Overlap ensures there is no stratum of unobserved confounding for which the treatment, running variable, or placebo treatment are deterministic.
In the language of the RDD literature, we rule out complete control but allow imprecise control \citep{lee2010regression}.

Exclusion encodes the idea that the placebo treatment does not directly cause the outcome of interest; instead, it affects the outcome via the running variable. Exclusion also encodes the idea that the placebo outcome is not caused by the treatment, running variable, or placebo treatment. Finally, exclusion imposes that the treatment is pinned down by only the running variable (and possibly idiosyncratic noise in the fuzzy setting).

\begin{assumption}[Conditional independence]\label{assum:indep}
There exists some $\epsilon > 0$ such that for all $d \in \mathcal{D}(\epsilon)$, $\eta_{a} \indep \eta_{y}, U \mid D = d$.
\end{assumption}

Assumption~\ref{assum:indep} is vacuously satisfied by the sharp design. In the fuzzy design, it is analogous to \citet[Condition C3]{hahn2001identification}: it implies that conditional on $D$, the randomness in the treatment assignment mechanism is unrelated to the idiosyncratic aspect of potential outcomes and the unobserved confounding. This assumption in the fuzzy design may be relaxed by imposing homogeneity of effects; see Appendix~\ref{sec:homog}.

Assumptions \ref{assum:proxy} and \ref{assum:indep} can be summarized by the results below:

\begin{lemma}[Necessary independences]\label{lemma:indep}
    Assumption~\ref{assum:proxy} implies that for all $d \in \mathcal{D}(\epsilon)$ and $z \in \mathcal{Z}(\epsilon)$,
    \begin{enumerate}[label=\alph*]
            \item $Y \indep Z \mid  D, U$
            \item $W \indep D, Z \mid U$\footnote{The conditional independence restriction on our placebo outcome $W$ resembles the conditional independence assumption imposed on the outcome $Y$ in \cite{angrist2015wanna}. Here, the unobserved confounder $U$ in Lemma~\ref{lemma:indep}b plays the role of the covariates $X$, and allows us to attribute changes in the conditional distribution $f(w \mid d)$ at the cutoff to the unobserved confounder $U$.}
    \end{enumerate}
    Additionally under Assumption~\ref{assum:indep} it follows that,
        \begin{enumerate}[label=\alph*]
            \setcounter{enumi}{2}
            \item $A \indep U \mid D$
    \end{enumerate}
\end{lemma}

\begin{example}[Maimonides rule]\label{ex:maimonides}
    Recall the Maimonides rule example from Section~\ref{sec:example}, where class size $A$ is determined by grade-level enrollment $D$, and we study its effect on student achievement $Y$ in grade $t$. Administrator ability $U$ is unobserved, and strategic manipulation around the cutoff induces a jump in $f(d \mid u, \eta_y)$ (and hence $f(u, \eta_y \mid d)$): sophisticated administrators who are aware of the threshold $d^{*}$ may inflate enrollment to open an additional class. Let $Z$ denote the availability of students to transfer across grade levels, and $W$ denote achievement in grade $t-1$.

    $Y \indep Z \mid D, U$ requires that student transfer availability is independent of test scores conditional on enrollment in grade $t$ and administrator ability. This is plausible because transfer availability is determined by demographic factors which are unrelated to how students perform on exams under conditioning. \cite{angrist_maimonides_2019} use $Z$ as an instrument in their IV specification under the stronger condition $Y \indep Z \mid D$.

    $W \indep D, Z \mid U$ requires that prior-year test scores in grade $t-1$ are independent of current enrollment and transfer availability in grade $t$, conditional on administrator ability. A straightforward timing argument supports this: scores from earlier grades carry no information about current enrollment or transfer availability.

    $A \indep U \mid D$ requires that administrator heterogeneity can only affect class size through enrollment. This holds because the Maimonides rule is a deterministic function of $D$: once enrollment is fixed, the number of classes is mechanically determined by the rule, leaving no room for administrator discretion.

    Finally, the overlap condition in Assumption~\ref{assum:proxy}c requires that school officials cannot perfectly control enrollment $D$. This is credible because enrollment depends on many demographic factors outside the administrator's control, even if some manipulation occurs at the margin. This accommodates the jump in $d \mapsto f(d \mid u, \eta_y)$ that motivates our approach. \hfill$\triangleleft$
\end{example}

\begin{example}[Birth-weight classification]\label{ex:almond}
Consider the empirical setting of \cite{almond2005birth}, who study the effect of being classified as ``critically underweight'' on infant mortality, where the classification is triggered when an infant's reported birth-weight falls below the $1500$ gram threshold. The running variable $D$ is the infant's reported birth-weight; the treatment $A$ is the indicator for being classified as critically underweight; the outcome $Y$ is infant mortality; and the unobserved confounder $U$ is healthcare quality. The concern motivating the design is that $U$ confounds the running variable: better doctors may deliberately report birth-weights just below $1500$ grams so that their patients receive the additional care reserved for critically underweight infants. This strategic manipulation around the cutoff induces a jump in $f(d \mid u, \eta_{y})$, invalidating the standard continuity argument. To recover the treatment effect, let $Z$ be an indicator for whether the infant was born at night (when hospitals are understaffed), and $W$ be the mother's education level.

$Y \indep Z \mid D, U$, requires that the time of day at which the infant is born is independent of mortality once we condition on healthcare quality and reported birth-weight. This is plausible: holding healthcare quality fixed, the time of birth should have no remaining effect on whether the infant survives.

$W \indep D, Z \mid U$, requires that the mother's education level is independent of both the reported birth-weight and the time of birth given healthcare quality. Holding the quality of care fixed, there is no plausible channel through which the mother's education can influence the reported birth-weight or the timing of delivery, so this restriction likely holds.

$A \indep U \mid D$, requires that healthcare quality affects the weight classification only through the reported birth-weight. This holds by construction: the classification $A$ is a deterministic function of $D$, so once the reported birth-weight is fixed, whether the infant is labeled critically underweight is fully determined by the threshold rule, leaving no residual dependence on $U$.

Finally, the overlap condition in Assumption~\ref{assum:proxy}c requires that doctors cannot perfectly control the reported birth-weight $D$. This is reasonable: an infant weighing well above $1500$ grams cannot credibly be reported as weighing below the threshold. \hfill$\triangleleft$
\end{example}

\begin{assumption}[Continuity]\label{assum:stable}
For all $d \in D(\epsilon)$, assume that $d \mapsto \E{Y(1,D,U,\eta_y)\mid D=d, U=u}$ and $d \mapsto \E{Y(0,D,U,\eta_y) \mid D=d, U=u}$ are continuous for all $u$.
\end{assumption}

Assumption \ref{assum:stable} weakens the standard continuity assumption of the RDD literature: continuity of potential outcomes conditional upon the running variable $D=d$ near the cutoff $d^*$. Instead, it requires continuity of potential outcomes conditional upon both $D=d$ and $U=u$. 

While the density of student enrollment 
may jump at $d^{*}$, we anticipate that potential student achievement is continuous when fixing the administrator's aptitude $U=u$.

\begin{lemma}[Limit]\label{lemma:limit}
    Under a fuzzy design and Assumptions \ref{assum:rdd}, \ref{assum:proxy}, \ref{assum:indep}, and \ref{assum:stable},
    $$
    \tau_{0} = \frac{\E{ \lim_{\epsilon \downarrow 0}\left[\E{Y \mid D = d^{*} + \epsilon, U} - \E{Y \mid D = d^{*} - \epsilon, U}\right]  \mid D = d^{*}}}{\lim_{\epsilon \downarrow 0}\left[\E{A \mid D = d^{*} + \epsilon} -  \E{A \mid D = d^{*} - \epsilon}\right]}.
    $$
    
    In a sharp design, and under Assumptions \ref{assum:rdd}, \ref{assum:proxy}, and \ref{assum:stable},
    $$
    \tau_{0} = \E{ \lim_{\epsilon \downarrow 0}\left[\E{Y \mid D = d^{*} + \epsilon, U} - \E{Y \mid D = d^{*} - \epsilon, U}\right]  \mid D = d^{*}}.
    $$
\end{lemma}

Lemma \ref{lemma:limit} solves for the treatment effect $\tau_{0}$ in an interpretable way: for each value of the unobserved confounded $U=u$, conduct RDD conditional upon $U=u$, then average over $U$. In our running example, this is equivalent to averaging the treatment effect of a class size $(A)$ on test scores in year-$t$ $(Y)$ for schools with enrollment levels at $d^{*}$ over different administrator-ability levels $U$. 

Clearly, this calculation is infeasible, since $U$ is unobserved. We need additional assumptions to identify $\tau_{0}$ in terms of observed variables.

\subsection{Main identification result}

Section~\ref{sec:assumptions} established that the causal estimand $\tau_0$ can be written in terms of conditional expectations over unobserved confounding $(U)$. The key challenge now is to express $\tau_0$ in terms of observable variables. To do this, we adapt the ``confounding bridge'' approach \citep{tchetgen2020introduction}, which uses auxiliary variables to solve an integral equation that recovers the influence of $U$ on $Y$.

\begin{assumption}[Confounding bridge]\label{assum:bridge}
    For all $d\in \mathcal{D}_{+}(\epsilon)$, there exists a bounded solution to the following operator equation almost surely,
    \begin{equation*}\begin{aligned}
    \E{Y \mid  D = d , U} &= \int h_{+}(d - d^{*},w)\dv \pr{w\mid D=d,  U}.\\
    \end{aligned}\end{equation*}
    Similarly, for all $d \in \mathcal{D}_{-}(\epsilon)$ we have,
    \begin{equation*}\begin{aligned}
    \E{Y \mid  D = d , U} &= \int h_{-}(d - d^{*},w) \dv \pr{w\mid D=d, U}.\\
    \end{aligned}\end{equation*}
\end{assumption}

The confounding bridge assumption states that, for fixed values of the running variable $(D)$ and unobserved confounder $(U)$, the conditional expectation of the outcome $(Y)$ can be expressed as a function of a proxy variable $(W)$ that is associated with $(U)$. The functions $h_+$ and $h_-$ act as weighting kernels that summarize how $W$ captures the influence of $U$ on $Y$ on either side of the threshold.

\setcounter{example}{0}
\begin{example}[Maimonides rule]
    In the Maimonides rule example, for schools with the same number of enrolled students $(D = d)$, there exists a function $h_+(d - d^*, \cdot)$ such that we can recover expected student achievement using a reweighting of the distribution of test scores in year $t - 1$ $(W)$. The function $h_+$ does not depend on $U$ directly, but its integral with respect to the distribution of $W$ conditional on $U$ reconstructs the confounded conditional expectation of $Y$. This formulation allows us to bypass the unobserved $U$ by using an observed variable $W$ that is a proxy for $U$. \hfill$\triangleleft$
\end{example}

\begin{example}[Birth-weight classification]
    In the birth-weight example, Assumption~\ref{assum:bridge} posits the existence of a function $h_{+}$ that links the observed placebo outcome $W$ (the mother's education) to the unobserved confounder $U$ (healthcare quality). Concretely, among infants with the same reported birth-weight, we require that the confounded conditional expectation of mortality, $\E{Y \mid D = d, U}$, can be expressed as an average of $h_{+}(d - d^{*}, W)$ over the distribution of the mother's education given $U$. Crucially, $h_{+}$ does not depend on healthcare quality $U$ directly; it is a function only of the observed quantities $d - d^{*}$ and $W$. Yet integrating it against $\pr{w \mid D = d, U}$ reproduces the $U$-confounded conditional expectation. This allows us to recover the quantities of interest through the observed proxy $W$. \hfill$\triangleleft$
\end{example}

First we show that $h_0$ is a solution to an integral equation in terms of observed variables, though it may not be unique.

\begin{lemma}[Factuals I]\label{lemma:factuals}
    Lemma~\ref{lemma:indep} and Assumption~\ref{assum:bridge} imply that any $h_{+}$ that exists also solves, for all $d\in\mathcal{D}_{+}(\epsilon)$ and all $z\in\mathcal{Z}(\epsilon)$, 
    $$
    \E{Y\mid  D=d,Z=z}=\int h_{+}(d - d^{*},w)\dv \pr{w\mid  D=d,Z=z}.
    $$
    Analogously, for all $d \in \mathcal{D}_{-}(\epsilon)$,
    $$
    \E{Y\mid  D=d,Z=z}=\int h_{-}(d - d^{*},w)\dv \pr{w\mid  D=d,Z=z}.
    $$
\end{lemma}

Lemma~\ref{lemma:factuals} translates the bridge equation from a conditional expectation over $U$ to one over observed variables $Z$ and $W$. It shows that the same $h_0$ function also solves a second integral equation where both sides of the equation are estimable from the data.

\begin{assumption}[Completeness]\label{assum:complete}
For all $d\in\mathcal{D}(\epsilon)$ and all $z\in\mathcal{Z}(\epsilon)$,
 $\mathbb{E}\{f(U)\mid  D=d,Z=z\}=0$ if and only if $f(U)=0$ almost surely.
\end{assumption}

Assumption~\ref{assum:complete} is necessary to guarantee a unique solution to the equation in Lemma~\ref{lemma:factuals} (factuals).  Without this assumption, there could be multiple functions $h$ that satisfy the same moment condition, making point identification of $\tau_0$ more challenging. Future work may build on \citet{bennett2022inference} to give  weaker conditions under which $\tau_{0}$ can still be identified without assuming completeness. 

\begin{example}[Partially linear DGP]\label{ex:partially_linear}
Assumptions~\ref{assum:bridge} and~\ref{assum:complete} are structural requirements on the informativeness of the placebo variables, and they are generally untestable. To aid interpretation, we introduce a partially linear DGP in which both conditions take on a transparent form. Consider
\begin{align*}
 Y &= f(D - d^{*}) + \beta_{u} U + \eta_{y}, \qquad \E{\eta_{y} \mid D, Z, U} = 0\\
 W &= \alpha_{0} + \alpha_{u} U + \eta_{w}, \qquad\qquad\quad \E{\eta_{w} \mid D, Z, U} = 0,
\end{align*}
where $f \colon \mathcal{D}(\epsilon) \mapsto \mathbb{R}$ is continuous at all $d \in \mathcal{D}(\epsilon)$ such that $d \neq d^{*}$, $A = \1{D \geq d^{*}}$, and we assume that $W$ is bounded for all $d \in \mathcal{D}(\epsilon)$. 

Under this DGP, computing $\E{Y \mid D, U}$ and $\E{W \mid D, U}$ and eliminating $U$ shows that Assumption~\ref{assum:bridge} (confounding bridge) reduces to the condition $\alpha_u \neq 0$, i.e., the placebo outcome $W$ must be genuinely informative about the confounder $U$.\footnote{See Lemma~\ref{lemma:pl_sufficiency} for a formal justification.} Although this condition is untestable since $U$ is unobserved, any researcher who uses a placebo outcome to diagnose continuity violations must implicitly believe that $W$ is related to $U$ in this way; our framework makes that belief precise.

For Assumption~\ref{assum:complete} (completeness), suppose additionally that $Z$ and $U$ are both discrete. \cite{newey2003instrumental} show that a simple sufficient condition is that $Z$ takes on at least as many distinct values as $U$, together with nondegeneracy of the relevant conditional distributions. More generally, completeness is a rank condition on the conditional distribution of $Z$ given $U$, and a range of sufficient conditions and supporting examples are available in the econometrics literature \citep{andrews2017examples}. \hfill$\triangleleft$
\end{example}

\begin{lemma}[Factuals II]\label{lemma:factuals2}
    Lemma~\ref{lemma:indep} and Assumption~\ref{assum:complete} imply that if there exists a function $h_{+}$ that, for all $d\in\mathcal{D}_{+}(\epsilon)$ and all $z\in\mathcal{Z}(\epsilon)$ satisfies,
    $$
    \E{Y\mid  D=d,Z=z}=\int h_{+}(d - d^{*}, w)\dv \pr{w\mid  D=d,Z=z},
    $$
    then Assumption~\ref{assum:bridge} holds. An analogous result holds for a function $h_{-}$ and all $d \in \mathcal{D}_{-}(\epsilon)$.
\end{lemma}

Lemma~\ref{lemma:factuals2} completes the identification argument by showing that the existence of a function $h_0$ that solves the observed-data version of the bridge equation implies the existence of a confounding bridge in the structural model. This allows us to write $\tau_{0}$ in terms of functions of observed variables.

\begin{theorem}[Placebo identification]\label{thm:identify}
    Suppose Assumptions~\ref{assum:rdd},~\ref{assum:proxy},~\ref{assum:indep},~\ref{assum:stable}, and~\ref{assum:bridge} hold. Then 
    $$\tau_{0} = \frac{\lim_{\epsilon \downarrow 0} \int h_{+}(\epsilon, w) \dv \pr{w \mid D = d^{*}} -  \lim_{\epsilon \downarrow 0} \int  h_{-}(-\epsilon,w)\dv \pr{w\mid D=d^*}}{\lim_{\epsilon \downarrow 0}\left[\E{A \mid D = d^{*} + \epsilon} -  \E{A \mid D = d^{*} - \epsilon}\right]}$$
    If in addition Assumption~\ref{assum:complete} holds then the integrals of $h_0$ in the numerator are identified. Therefore $\tau_{0}$ is identified.
    
    In sharp design, under Assumptions~\ref{assum:rdd},~\ref{assum:proxy},~\ref{assum:stable}, and~\ref{assum:bridge}, then this simplifies to:
    $$\tau_{0} = \lim_{\epsilon \downarrow 0} \int h_{+}(\epsilon, w) \dv \pr{w \mid D = d^{*}} -  \lim_{\epsilon \downarrow 0} \int  h_{-}(-\epsilon,w)\dv \pr{w\mid D=d^*}$$
\end{theorem}

Theorem~\ref{thm:identify} provides a complete identification strategy---it shows how the discontinuity in potential outcomes conditional on unobserved $U$ can be corrected using an observable proxy $(W)$ and an auxiliary variable $(Z)$. The functions $h_+$ and $h_-$ effectively “invert” the confounding effect of $U$ using the relationship between $U$ and $W$. This inversion is unique when completeness holds and yields a point-identified causal estimand.

\setcounter{example}{2}
\begin{example}[Partially linear DGP]
    Recall the structural equations from Example~\ref{ex:partially_linear}, now written with treatment-specific regression functions on either side of the cutoff,
    \begin{align*}
         Y &= A f_{+}(D - d^{*}) + (1 - A) f_{-}(D - d^{*}) + \beta_{u} U + \eta_{y}, \qquad\; \E{\eta_{y} \mid D, Z, U} = 0\\
         W &= \alpha_{0} + \alpha_{u} U + \eta_{w}, \qquad\qquad\qquad\qquad\qquad\qquad\qquad\quad \E{\eta_{w} \mid D, Z, U} = 0,
    \end{align*}
    where $f_{+}, f_{-} \colon \mathcal{D}(\epsilon) \mapsto \mathbb{R}$ are continuous on either side of the cutoff.

    Lemma~\ref{lemma:pl_sufficiency} verifies that the confounding bridges
    \begin{align*}
        h_{+}(D - d^{*}, W) &= A \left\{f_{+}(D - d^{*}) - \frac{\beta_{u}}{\alpha_{u}}\alpha_{0}+ \frac{\beta_{u}}{\alpha_{u}} W\right\}\\
        h_{-}(D - d^{*}, W) &= (1 - A) \left\{f_{-}(D - d^{*}) - \frac{\beta_{u}}{\alpha_{u}}\alpha_{0}+ \frac{\beta_{u}}{\alpha_{u}} W\right\}
    \end{align*}
    satisfy the moment condition in Lemma~\ref{lemma:factuals} (factuals). This allows the moment condition in Lemma~\ref{lemma:limit} (limit) to be expressed entirely in terms of observable variables, permitting estimation.

    Applying Theorem~\ref{thm:identify} (identification), the RDD treatment effect is then identified as
    \begin{align*}
        \tau_{0} &= f_{+}(0) - f_{-}(0) = \E{\lim_{\epsilon \downarrow 0}\left[h_{+}(\epsilon, W) - h_{-}(-\epsilon, W)\right] \mid D = d^{*}}.
    \end{align*}
    That is, $\tau_{0}$ is recovered by taking the limiting difference of the two bridge functions at the cutoff, which removes the confounding contribution of $U$ and leaves only the causal effect of the treatment. \hfill$\triangleleft$
\end{example}
\section{Estimation: Local instrumental variable regression}\label{sec:estimation}

Theorem~\ref{thm:identify} identifies the RDD treatment effect by reweighting two confounding bridges $h_{-}(d, w)$ and $h_{+}(d, w)$. 
Recent literature to approximate these integral equations rely on kernel methods to non-parametrically estimate these confounding bridges. However, \citet{fan_variable_1992} show that the order of bias of these kernel estimators on boundary points ($d^{*}$) is high, and advocate the use of local linear regressions to estimate points on the boundary with less bias.

As our second contribution, we propose a local linear instrumental variable estimator, extending \citet{fan_variable_1992} to accommodate instruments. For tractability, we introduce a partially linear approximation to the confounding bridge in Assumption~\ref{assum:bridge}.  The final estimator decomposes into two intepretable terms; it adjusts the discontinuity in the target outcome by the discontinuity in the placebo outcome.

For clarity, we focus on the sharp design, i.e. the numerator in Theorem~\ref{thm:identify}. The fuzzy design is a straightforward extension, since the denominator in Theorem~\ref{thm:identify} is a standard RDD estimand.

For simplicity, we also focus on the exactly identified case where $\dim(z)=\dim(w)$. Our results naturally extend to the overidentified case $\dim(z)\geq \dim(w)$ by standard techniques.

\subsection{Locally approximating the confounding bridge}

In line with the identification strategy presented in Section~\ref{sec:double_proxy}, we estimate the treatment effect $\tau_{0}$ using a two-step procedure. First, we solve for the confounding bridges $h_{+}(\cdot, \cdot)$ and $h_{-}(\cdot, \cdot)$ at the cutoff $d^{*}$ using the moment condition from Lemma~\ref{lemma:factuals}. Then, invoking Theorem~\ref{thm:identify}, we compute the difference in their limits, $
\tau_{0} = \lim_{\epsilon \downarrow 0} \left[ \int h_{+}(\epsilon, w) \, d\mathbb{P}(w \mid D = d^{*}) - \int h_{-}(-\epsilon, w) \, d\mathbb{P}(w \mid D = d^{*}) \right].
$

To motivate our estimator, consider the class of partially linear functions,
\begin{equation}\begin{aligned}\label{eq:partial}
    h_{+}(d - d^{*}, w) = g_{+}(d - d^{*}) + w^{\top}\gamma_{+}, \qquad h_{-}(d - d^{*}, w) = g_{-}(d - d^{*}) + w^{\top}\gamma_{-},
\end{aligned}\end{equation}
where $g_{+} \colon \mathcal{D}_{+}(\epsilon) \rightarrow \mathbb{R}$ and $g_{-} \colon \mathcal{D}_{-}(\epsilon) \rightarrow \mathbb{R}$ are any continuously differentiable functions, and $\gamma_{+}$ plays the role of the best linear projection of $h_{+}(\cdot, \cdot)$ onto $W$ at $D=d^{*}$. 

Appendix~\ref{sec:sim_description} shows that partially linear potential outcomes imply partially linear confounding bridges. More generally, a partially linear confounding bridge may be viewed as a best-in-class approximation to a nonlinear confounding bridge; see Proposition~\ref{prop:estimand} below for a detailed characterization.

Within this class, the moment condition in Lemma~\ref{lemma:factuals} becomes
\begin{equation}\begin{aligned}\label{eq:lin_moment}
\E{Y - g_{+}(D - d^{*}) - W^{\top}\gamma_{+} \mid D = d, Z = z} = 0.
\end{aligned}\end{equation}

This expression closely resembles a partially linear instrumental variable moment condition, where $W$ plays the role of the endogenous regressor and $Z$ serves as the instrument. Therefore, we propose a local  instrumental variable regression approach to estimate the best approximation to $h_{+}(\cdot, \cdot)$ near $d^{*}$.

\subsection{Algorithm}

To define our estimator, we introduce some notation. Let $\alpha_{+} = (\alpha_{+, 0}, h_{n} \alpha_{+, 1})^{\top} = \left(g_{+}(0), h_{n} g_{+}'(0)\right)^{\top}$ concatenate the nonlinear component of $h_+$ and its derivative at the cutoff, where $h_n$ is the kernel bandwidth. We concatenate all of the parameters to be estimated as $\nu_{+}^{\top} = (\alpha_{+, 0}, \alpha_{+, 1}, \gamma_{+}^{\top})$.

We propose what appears to be a new local linear objective function, which we call local instrumental variable regression. Its first order condition, which yields the estimator $\hat{\nu}_+$, is 
\begin{equation*}\begin{aligned}
    \frac{1}{n} \sum_{i = 1}^{n}  \omega_{i, +} \begin{bmatrix}
     R_{i, 1} \\ Z_{i}
    \end{bmatrix} \left(Y_{i} -  R_{i, 1}^{\top} \hat{\alpha}_{+} - W_{i}^{\top} \hat{\gamma}_{+} \right) &= \pmb{0},
\end{aligned}\end{equation*}
where $\omega_{i, +} = \frac{1}{h_{n}} \1{D_{i} \geq d^{*}} \K{\frac{\vert{D_{i} - d^{*}}}{h_{n}}}$ are local kernel weights, $R_{i, 1} = \left(1, \frac{D_{i} - d^{*}}{h_{n}} \right)^{\top}$ is a transformed vector of the running variable, and $Z_{i} \in \mathbb{R}^{\dim(w)}$ is a vector of placebo treatments. This empirical moment is clearly a sample analogue of the population moment implied by equation~\eqref{eq:lin_moment}. The left limit objects are analogous.

We then use the coefficients estimated by local instrumental variable regression to construct our estimator of the treatment effect. The appropriate formula is immediate from Theorem~\ref{thm:identify} and equation~\eqref{eq:partial}. Formally, the approximate treatment effect and its estimator are
\begin{equation}\begin{aligned}\label{eq:naive_estimator}
    \tau_{\text{pdd}} &= g_{+}(0) + \lim_{\epsilon \downarrow 0} \E{W^{\top} \mid D = d^{*} + \epsilon} \gamma_{+} - \left(g_{-}(0) + \lim_{\epsilon \downarrow 0} \E{W^{\top} \mid D = d^{*} + \epsilon} \gamma_{-}\right)\\
    &=\alpha_{+, 0} + (\beta_{+, 0}^{w})^{\top} \gamma_{+} - \left\{\alpha_{-, 0} + (\beta_{+, 0}^{w})^{\top}\gamma_{-}\right\}\\
    \hat{\tau}_{\text{pdd}} &= \hat{\alpha}_{+, 0} + (\hat{\beta}_{+, 0}^{w})^{\top} \hat{\gamma}_{+} - \left\{\hat{\alpha}_{-, 0} + (\hat{\beta}_{+, 0}^{w})^{\top} \hat{\gamma}_{-}\right\},
\end{aligned}\end{equation}
where $\hat{\beta}_{+, 0}^{w}$ estimates the right-limit $\beta_{+, 0}^{w}=\lim_{\epsilon \downarrow 0} \E{W_{i} \mid D = d^{*} + \epsilon}$ using the local linear estimator of  \citet{hahn2001identification}:
$$H_{1} \hat{\beta}_{+}^{w_{j}} = \argmin_{\beta} \sum_{i = 1}^{n} \omega_{i, +} (W_{i, j} - R_{i, 1}^{\top}\beta )^{2}.$$
Then, we let 
$\hat{\beta}_{+, 0}^{w} = \left(\hat{\beta}_{+, 0}^{w_{1}}, \dots, \hat{\beta}_{+, 0}^{w_{\dim(w)}}\right)^{\top} \in \mathbb{R}^{\dim(w)}$ concatenate the first components of these estimated coefficients. The left limit objects are analogous.

Section~\ref{sec:inference} below shows that this estimator is consistent for $\tau_{\text{pdd}}$. After bias correction, which we defer to Section~\ref{sec:inference}, it is also asymptotically normal at the familiar rate of $n^{-2/5}$.

\subsection{Equivalence: RDD with a new adjustment term}

For any finite sample size, our proposed estimator is numerically equivalent to a highly interpretable procedure: take the standard RDD estimator based on discontinuity in the outcome ($Y$), and add an adjustment term based on the discontinuity in the placebo outcome ($W$). 

\begin{proposition}[Estimator equivalence]\label{prop:decomp}
    Our proposed estimator in equation~\eqref{eq:naive_estimator} is numerically equivalent to 
    \begin{equation*}\begin{aligned}
        \hat{\tau}_{\text{pdd}} &= \hat{\tau}_{\text{rdd}}^{y} - \left(\hat{\tau}_{\text{rdd}}^{w} \right)^\top \hat{\gamma}_{-},
    \end{aligned}\end{equation*}
    where $\hat{\tau}_{\text{rdd}}^{y} = \hat{\beta}_{+, 0}^{y} - \hat{\beta}_{-, 0}^{y} \in \mathbb{R}$ and $\hat{\tau}_{\text{rdd}}^{w} = \left(\hat{\beta}_{+, 0}^{w} - \hat{\beta}_{-, 0}^{w}\right) \in \mathbb{R}^{\dim(w)}$ are constructed from the familiar RDD local linear objective functions, namely
    \begin{equation*}\begin{aligned}
        \hat{\beta}_{+}^{y} = \argmin_{\beta} \frac{1}{n} \sum_{i = 1}^{n} \omega_{i, +}\left(Y_{i} - R_{i, 1}^{\top}\beta \right)^{2}, \qquad \hat{\beta}_{+}^{w_j} = \argmin_{\beta} \frac{1}{n} \sum_{i = 1}^{n} \omega_{i, +}\left(W_{i,j} - R_{i, 1}^{\top}\beta \right)^{2},
    \end{aligned}\end{equation*}
    and analogous objectives below the cutoff.
    Moreover, $\hat{\gamma}_-$ is the local instrumental variable regression of the residualized $Y^{\perp}$ on the residualized $W^{\perp}$, instrumenting with $Z$, below cutoff:
    \begin{equation*}\begin{aligned}
       \hat{\gamma}_{-}&=
       \left[\frac{1}{n} \sum_{i = 1}^{n}  \omega_{i, -} \{Z_{i} (W_{i}^{\perp})^{\top}\}\right]^{-1}
       \left\{\frac{1}{n} \sum_{i = 1}^{n}  \omega_{i, -} (Z_{i} Y_{i}^{\perp})\right\}.
    \end{aligned}\end{equation*}
    Here, $(Y_{i}^\perp)$ are the residuals from a local linear regression of $(Y_i)$ onto $(D_i)$, and $(W_{i}^\perp)$ are the residuals from a local linear regression of $(W_i)$ onto $(D_i)$. In other words, we residualize the target outcome and the placebo outcome using the running variable.\hfill$\triangleleft$
\end{proposition}

The RDD literature, summarized recently by e.g. \citet{cattaneo_practical_2024}, advocates for testing whether $\hat{\tau}_{\text{rdd}}^{w}$ is close to zero as a diagnostic for whether $f(d \mid u,\eta_y)$ (or, relatedly, $f(u,\eta_y \mid d)$) changes discontinuously at $d^{*}$. Our contribution is to show that this discontinuity is an essential ingredient for an adjustment term. If $\hat{\tau}_{\text{rdd}}^{w}$ is indeed negligible, our adjustment term is zero and our estimator $\hat{\tau}_{\text{pdd}}$ simplifies to the standard RDD estimate $\hat{\tau}_{\text{rdd}}^{y}$.

Within our adjustment term, the placebo outcome discontinuity ($\hat{\tau}_{\text{rdd}}^{w}$) is multiplied by weights ($\hat{\gamma}_{-}$). 

These weights are increasing in the correlation between the residualized outcome ($Y^{\perp}$) and the placebo treatment ($Z$). Intuitively, if this correlation is stronger, then there is more unobserved confounding on the target outcome, and hence we increase the weights. 

The weights are decreasing in the correlation between the residualized placebo outcome ($W^{\perp}$) and the placebo treatment ($Z$). Intuitively, if this correlation is weaker, then the placebo outcome is a weaker proxy for the unobserved confounding, and hence we compensate by increasing the weights on its estimated discontinuity. 

Proposition~\ref{prop:decomp} clarifies the danger of a very weak proxy, which would lead to numerical instability in this final aspect of the weights. Future work may study the behavior of our estimator when using a weak proxy.

\begin{remark}[Only left-adjustment]
    The decomposition in Proposition~\ref{prop:decomp} only estimates $\gamma_{-}$, and not $\gamma_{+}$, because our treatment effect of interest is the effect at the cutoff. As such, it is the effect on the treated, so only the untreated potential outcome needs adjustment.

    In more detail, our estimator is anchored on the distribution of $U$ just above the cutoff. It preserves the right-limit $\mathbb{E}[Y \mid D = d^{*} + \epsilon]$. However, it adjusts the left-limit using $\gamma_{-}$. The adjusted term allows us  to estimate the counterfactual below the cutoff $\mathbb{E}[\mathbb{E}[Y \mid D = d^{*} - \epsilon, U] \mid D = d^{*}]$, in the thought experiment where $U$ matches the distribution it would have followed above the cutoff. 
    
    An alternative causal parameter is $\tilde{\tau}_{0} = \mathbb{E}[Y(1, d^{*}, U, \eta_{y}) - Y(0, d^{*}, U, \eta_{y})]$, for which our estimator would have two adjustment terms: 
    $\tilde{\tau}_{\text{pdd}} = \hat{\tau}_{\text{rdd}}^{y} + \left(\overline{W} - \hat{\beta}_{+, 0}^{w}\right)^{\top}\hat{\gamma}_{+} - \left(\overline{W} - \hat{\beta}_{-, 0}^{w}\right)^{\top}\hat{\gamma}_{-}$. \hfill $\triangleleft$
\end{remark}

\section{Bias corrected inference}\label{sec:inference}

As our third contribution, we study the large sample properties of the estimator in Proposition~\ref{prop:decomp}. We prove it is consistent. Then, we prove that it is asymptotically normal at the familiar rate $n^{-2/5}$ after bias correction.

\subsection{Consistency}

While prior work used placebo outcomes and predetermined covariates to detect assumption violations, we propose a method to directly correct for them and thereby recover $\tau_{0}$. Below we state sufficient conditions for consistency.

Let $\mu_{+, y}(d - d^{*}) = \E{Y \mid D = d}$ for $d \geq d^{*}$ and $\mu_{-, w_{j}}(d - d^{*}) = \E{W_{i, j} \mid D = d}$ for $d < d^{*}$. Similarly, define $\mu_{-, y}(d - d^{*})$, $\mu_{+, w_{j}}(d - d^{*})$, $\mu_{-, w_{j}}(d - d^{*})$, $\mu_{+, z_{j}}(d - d^{*})$, and $\mu_{-, z_{j}}(d - d^{*})$. Let $\mu_{+, z_{j}w_{k}}(d - d^{*}) = \E{W_{i, k}Z_{i, j} \mid D = d}$.

Let $\sigma^{2}_{+, y}(d - d^{*}) = \E{\left(Y_{i} - g_{+}(d - d^{*}) - W_{i}^{\top}\gamma_{+}\right)^{2} \mid D_i = d}$ for $d \geq d^{*}$ and $\sigma^{2}_{-, w_j}(d - d^{*}) = \hfill\break \E{\left(W_{i, j} - \mu_{-, w_j}(d - d^{*})\right)^{2} \mid D_i = d}$ for $d < d^{*}$. Similarly, define $\sigma^{2}_{-, y}(d - d^{*})$, $\sigma^{2}_{+, w_{j}}(d - d^{*})$, $\sigma^{2}_{+, z_{j}}(d - d^{*})$, and $\sigma^{2}_{-, y}(d - d^{*})$.

\begin{assumption}[Regularity conditions]\label{assum:reg}
    Let the following conditions hold.
    \begin{enumerate}[label = \alph*]
        \item (Continuous mean)  For all $d \in \mathcal{D}_{+}(\epsilon)$ and $j \in \{1, \dots, \dim(w)\}$, $g_{+}(\cdot)$ and $\mu_{+, w_{j}}(\cdot)$ are 3-times continuously differentiable and bounded. Additionally, $\mu_{+, z_{j}}(\cdot)$ and $\mu_{+, z_{j}w_{k}}(\cdot)$ are continuous and bounded.
        
        For all $d \in \mathcal{D}_{-}(\epsilon)$ and $j \in \{1, \dots, \dim(w)\}$, $\mu_{-, y}(\cdot)$ and $\mu_{-, w_{j}}(\cdot)$ are 3-times continuously differentiable and bounded. Additionally, $\mu_{-, z_{j}}(\cdot)$ and $\mu_{-, z_{j}w_{k}}(\cdot)$ are continuous and bounded.

        \item (Continuous variance) For all $d \in \mathcal{D}_{+}(\epsilon)$ and $j, k \in \{1, \dots, \dim(w)\}$, $\sigma^{2}_{+, y}(\cdot)$, $\sigma^{2}_{+, w_{j}}(\cdot)$, and $\sigma^{2}_{+, z_j}(\cdot)$ are continuous and bounded away from $0$. For all $d \in \mathcal{D}_{-}(\epsilon)$, $\sigma^{2}_{-, y}(\cdot)$, $\sigma^{2}_{-, w_{j}}(\cdot)$, and $\sigma^{2}_{-, z_j}(\cdot)$ are continuous and bounded away from $0$.

        For all $d \in \mathcal{D}_{+}(\epsilon)$, $\rho_{+, y^{2}z}(d - d^{*})$, $\rho_{+, y^{2}zz^{\top}}(d - d^{*})$, $\rho_{+, yw_jz}(d - d^{*})$, $\rho_{+, yw_j}(d - d^{*})$, $\rho_{+, z_kw_j}(d - d^{*})$, and $\rho_{+, w_jw_k}(d - d^{*})$ are continuous. For all $d \in \mathcal{D}_{-}(\epsilon)$, $\rho_{-, y^{2}z}(d - d^{*})$, $\rho_{-, y^{2}zz^{\top}}(d - d^{*})$, $\rho_{-, yw_jz}(d - d^{*})$, $\rho_{-, yw_j}(d - d^{*})$, $\rho_{-, z_kw_j}(d - d^{*})$, and $\rho_{-, w_jw_k}(d - d^{*})$ are continuous.\footnote{The $\rho_{\pm, x}(\cdot)$ are covariance terms, with formulas given in Appendix~\ref{sec:appendix_bias}.}

        \item (Continuous distribution) For all $d \in \mathcal{D}(\epsilon)\backslash\{d^{*}\}$, $f_{d}(d)$ is bounded, continuous, and nonzero, where $f_{d}(d)$ is the distribution of $D$.
        
        \item (Bounded higher moments) For all $d \in \mathcal{D}(\epsilon)$ and $j \in \{1, \dots, \dim(w)\}$, there exists some $\zeta > 0$ and $M \in \mathbb{R}$ such that we have $\E{\vert{Y_{i} - g_{+}(d - d^{*}) - W_{i}^{\top}\gamma_{+}}^{2 + \zeta} \mid D_{i} = d}$, \hfill\break $\E{\vert{Y_{i} - g_{-}(d - d^{*}) - W_{i}^{\top}\gamma_{-}}^{2 + \zeta} \mid D_{i} = d}$, $\E{\vert{W_{i, j} - \mu_{+, w_{j}}(d - d^{*})}^{2 + \zeta} \mid D_{i} = d }$, \hfill\break $\E{\vert{W_{i, j} - \mu_{-, w_{j}}(d - d^{*})}^{2 + \zeta} \mid D_{i} = d},$ $\E{\vert{Z_{i, j}\left(Y_{i} - g_{+}(d - d^{*}) - W_{i}^{\top}\gamma_{+}\right)}^{2 + \zeta} \mid D_{i} = d} < M$.
    \end{enumerate}
\end{assumption}

The regularity conditions in Assumption~\ref{assum:reg} are variations of standard assumptions in the RDD literature \citep{porter2003estimation, calonico_regression_2019}. 

\begin{assumption}[Kernel conditions]\label{assum:kernel}
    Assume that the kernel function $K \colon \mathbb{R} \mapsto \mathbb{R}^{+}$ is symmetric, bounded, and continuous at $0$ with bounded support.
\end{assumption}

Assumption~\ref{assum:kernel} is satisfied by the window kernel $\K{u} = \1{\vert{u} \leq 1}$ and the triangle kernel $\K{u} = \left(1 - \vert{u}\right) \1{\vert{u} \leq 1}$, which are the most frequently used kernels in the RDD literature. It also allows for the Epanechnikov kernel $\K{u} = \frac{3}{4}(1 - u^{2}) \1{\vert{u} \leq 1}$. Under these conditions we can characterize the probability limit of our local instrumental variable estimator for any vanishing bandwidth  $1\gg  h_n \gg n^{-1}$.

\begin{proposition}[Probability limit]\label{prop:estimand}
    Let Assumptions~\ref{assum:reg}a-c and~\ref{assum:kernel} hold, and assume that $nh_{n} \rightarrow \infty$ and $h_{n} \rightarrow 0$. Then, it follows that,
   $$
        \hat{\tau}_{\text{pdd}} \inprob \tau_{\text{rdd}}^{y} - (\tau_{\text{rdd}}^{w})^{\top} \gamma_{-},
$$
    where $\gamma_{-} = \lim_{\epsilon \downarrow 0} \Cov{ Z_{i}, W_{i}^{\top} \mid D = d^{*} - \epsilon}^{-1} \Cov{ Z_{i}, Y_{i} \mid D_{i} = d^{*} - \epsilon}$ is a best local linear approximation, i.e. it solves the local projection problem 
   $$
        Y = c + W^{\top} \gamma_{-} + \eta, \qquad \lim_{\epsilon \downarrow 0}  \E{\eta  Z_{i} \mid D_{i} = d^{*} - \epsilon} =0,\qquad  \lim_{\epsilon \downarrow 0} \E{\eta \mid D_{i} = d^{*} - \epsilon} = 0.
  $$\hfill$\triangleleft$
\end{proposition}

Building on Proposition~\ref{prop:decomp} (estimator equivalence), which characterizes the finite sample estimate,
Proposition~\ref{prop:estimand} characterizes the asymptotic limit of the local estimator. Now, at the population level, we show that if the placebo outcome does not jump at the cutoff ($\tau^{w}_{\text{rdd}} = 0$), then the local instrumental variable estimand reduces to the traditional RDD estimand. 

The population limit of the weights admits a similar interpretation as before.  It also sheds light on the nature of our approximation. 
In particular, Proposition~\ref{prop:estimand} shows that $\gamma_{-}$ can be interpreted as a linear projection of $Y$ onto the distribution of $W$, instrumenting with $Z$, when $D$  approaches the cutoff from the left. In other words, $\gamma_{-}$ is the parameter which best approximates the confounding bridge in Lemma~\ref{lemma:factuals} (factuals) as a partially linear function of the placebo outcome $W$. 

If the confounding bridge $h_{-}(0, w)$ is very nonlinear in $W$ then this approximation can be poor. However, we advocate for this approach because full nonparametric estimation over $(D_{i},W_{i})$ would require us to uniformly estimate $h_{-}(0, w)$ as a function of $w$ to characterize the asymptotic bias of the local instrumental variable estimator. There would be $\dim(w)$ additional kernels and bandwidths in the estimator. As discussed by \citet{calonico_regression_2019}, localizing weights  for many variables would suffer from a curse of dimensionality, and hence empirical applications would become challenging. Our estimator does not suffer from this curse of dimensionality, and is therefore practical for empirical researchers. Future work may attempt to overcome these challenges, perhaps by placing additional structure on the effective dimension or sparsity of $W$ \citep{noack2021flexible}. 

\begin{corollary}[Consistency]\label{cor:consistency} Let Assumptions~\ref{assum:rdd},~\ref{assum:proxy},~\ref{assum:indep},~\ref{assum:stable},~\ref{assum:bridge},~\ref{assum:reg}a, and~\ref{assum:kernel} hold. Additionally, suppose that equation~\eqref{eq:partial} holds almost surely. Then, if $nh_{n} \rightarrow \infty$ and $h_{n} \rightarrow 0$, it follows that $\hat{\tau}_{\text{pdd}} \inprob \tau_{0}$.\hfill$\triangleleft$
\end{corollary}

Corollary~\ref {cor:consistency} shows that when the confounding bridge lies in the class of partially linear functions in equation~\eqref{eq:partial}, then the local instrumental variable estimator is consistent for the treatment effect $\tau_{0}$. As shown in Appendix~\ref{sec:sim_description}, partially linear potential outcomes imply partially linear confounding bridges. More generally, partial linearity of the confounding bridge may be a reasonable approximation. 

\subsection{Bias correction and asymptotic normality}\label{sec:rbc_pdd}

As shown in Proposition~\ref{prop:decomp}, our estimator can be decomposed into a sum of two components: (i) the RDD estimate of the outcome discontinuity $\hat{\tau}_{\text{rdd}}^{y}$, and (ii) the RDD estimate of the placebo outcome discontinuity, appropriately weighted as $(\hat{\tau}_{\text{rdd}}^{w})^{\top}\hat{\gamma}_{-}$. A key insight from \citet{calonico_robust_2014} is that, when the bandwidth $h_n$ is selected to optimize mean squared error (MSE), these discontinuity estimators suffer from non-negligible bias in their sampling distributions. This bias persists in large samples and leads to invalid inference, particularly when constructing confidence intervals.

Proposition~\ref{prop:bias} shows the asymptotic bias of the naive PDD estimator at the MSE-optimal bandwidth rate of $h_n \asymp n^{-1/5}$. 

\begin{proposition}[PDD bias]\label{prop:bias}
     Let Assumptions~\ref{assum:rdd},~\ref{assum:proxy},~\ref{assum:indep},~\ref{assum:stable},~\ref{assum:bridge},~\ref{assum:reg}, and~\ref{assum:kernel} hold, and assume that the bandwidths satisfy $h_n = c_h n^{-1/5}$ and $b_n = c_b n^{-1/5}$, where $b_n$ is the bandwidth of the bias correction. Additionally, suppose that equation~\eqref{eq:partial} holds almost surely. Then, it follows that,
    \begin{align*}
        \E{\sqrt{nh_{n}}(\hat{\tau}_{\text{pdd}} - \tau_{\text{pdd}})} &\to \mathbf{e}_{+, w}^{\top} \left(\Psi_{+, 1}^{-1} \eta_{+, \mathrm{IV}} -  \Psi_{-, 1}^{-1} \eta_{-, \mathrm{IV}}\right) + \left(\gamma_{+} - \gamma_{-}\right)^{\top} \pmat{\mathbf{e}_{0}^{\top}\Gamma_{+, 1}^{-1} \eta_{+, w_1} \\ \vdots \\ \mathbf{e}_{0}^{\top}\Gamma_{+, 1}^{-1} \eta_{+, w_q}}.
    \end{align*}
    Formulas for $\mathbf{e}_{+, w}, \Psi_{+, 1}, \eta_{+, \mathrm{IV}}, \Gamma_{+, 1},$ and $\eta_{+, w_j}$ are given in Appendix~\ref{sec:appendix_bias}. \hfill$\triangleleft$
\end{proposition}

To address this issue, we extend the robust bias correction technique developed by \citet{calonico_robust_2014}. In particular, we construct bias-corrected estimators that remain consistent and asymptotically normal, even when using MSE-optimal bandwidths. We then define the bias-corrected PDD estimator $\hat{\tau}_{\mathrm{pdd}}^{\mathrm{rbc}}$ as,
\begin{equation}\begin{aligned}\label{eq:pdd_rbc_def}
    \hat{\tau}_{\mathrm{pdd}}^{\mathrm{rbc}} &= \hat{\tau}_{\mathrm{pdd}} - \left(\hat{\mathbf{e}}_{+, w}^{\top} \left(\hat{\Psi}_{+, 1}^{-1}\hat{\eta}_{+, \mathrm{IV}} - \hat{\Psi}_{-, 1}^{-1}\hat{\eta}_{-, \mathrm{IV}}\right) +  \left(\hat{\gamma}_{+, 1} - \hat{\gamma}_{-, 1}\right)\mathbf{e}_{0}^{\top}\hat{\Gamma}_{+, 1}^{-1} \hat{\eta}_{+, w_1} + \cdots + \left(\hat{\gamma}_{+, q} - \hat{\gamma}_{-, q}\right)\mathbf{e}_{0}^{\top}\hat{\Gamma}_{+, 1}^{-1} \hat{\eta}_{+, w_q}\right).
\end{aligned}\end{equation}

The theorem below establishes the asymptotic normality of our estimator.

\begin{theorem}[Limit distribution]\label{thm:limit_dist}
   Let Assumptions~\ref{assum:rdd},~\ref{assum:proxy},~\ref{assum:indep},~\ref{assum:stable},~\ref{assum:bridge},~\ref{assum:reg}, and~\ref{assum:kernel} hold, and assume that the bandwidths satisfy $h_n = c_h n^{-1/5}$ and $b_n = c_b n^{-1/5}$, where $b_n$ is the bandwidth of the bias correction. Additionally, suppose that equation~\eqref{eq:partial} holds almost surely. Then, it follows that,
    \begin{equation*}\begin{aligned}
        \sqrt{nh_{n}} \left(\hat{\tau}_{\mathrm{pdd}}^{\mathrm{rbc}} - \tau_{0}\right) &\indist \Norm{0, V},
    \end{aligned}\end{equation*}
    where $V$ and a consistent estimator for $V$ are given in Appendix~\ref{sec:appendix_bias}.\hfill$\triangleleft$
\end{theorem}

Theorem~\ref{thm:limit_dist} establishes the asymptotic normality of the bias-corrected local instrumental variable estimator under the commonly used MSE-optimal bandwidth rate of $h_n\asymp n^{-1/5}$ \citep{imbens_regression_2008}. The bias correction bandwidth may be of the same order, i.e. $b_n\asymp h_n$. 

In practice, we implement a data-driven bandwidth choice $\hat{h}_n$ along the lines of e.g. \citet{imbens2012optimal}. Future work may develop a data-driven bandwidth algorithm to minimize the length of confidence intervals along the lines of \citet{calonico2018effect,calonico_optimal_2020}.

Theorem~\ref{thm:limit_dist}  justifies the use of standard Wald-type inference for $\hat{\tau}_{\text{pdd}}^{\text{rbc}}$, e.g. the construction of confidence intervals and hypothesis tests.

\section{Discussion}\label{sec:conclusion}

This paper develops a new identification and estimation framework for regression discontinuity designs when the standard continuity of potential outcomes assumption fails due to unobserved confounding. We show that by leveraging a placebo treatment and a placebo outcome, it is possible to recover the causal parameter even when the running variable's distribution exhibits discontinuities at the threshold. Our identification argument relies on conditional continuity given an unobserved confounder, and it employs the concept of a confounding bridge to encode the relationship between the unobserved confounder and its observed proxy. Under a completeness condition, we show that the integrated bridge function is identified from observed data, and the treatment effect at the cutoff can be recovered using integrals of observable variables.

To operationalize this strategy, we propose a local instrumental variable estimator that approximates the confounding bridge using a partially linear specification. We demonstrate that the estimator decomposes into a standard RDD term and an adjustment term driven by the discontinuity in the placebo outcome. This decomposition yields an interpretable adjustment that allows researchers to correct for violations of the standard RDD assumptions, rather than discarding such designs altogether. We establish conditions for consistency and bias-corrected inference, allowing  optimal bandwidths. Our results extend the applicability of RDD methods to settings with strategic behavior that is detected via placebo discontinuities, providing researchers with a tractable and theoretically grounded approach for handling unobserved confounding near the cutoff.


\spacingset{1}
\bibliographystyle{apalike}
\bibliography{main}
\spacingset{1.8}

\newpage
\appendix
\numberwithin{assumption}{section}
\numberwithin{lemma}{section}
\numberwithin{theorem}{section}

\section{Identification proof}\label{sec:id_proof}

\subsection{Necessary independences}

\begin{proof}[Proof of Lemma~\ref{lemma:indep} (necessary independences)]

We proceed in steps.

\begin{enumerate}
    \item First, we show that $A \indep Z \mid D, U$. We write,
\begin{equation*}\begin{aligned}
\p{A=a\mid  D=d,Z=z,U}&=\p{A(d,\eta_a)=a\mid  D=d,Z=z,U}\\
&=\p{A(d,\eta_a)=a\mid  D=d,U},
\end{aligned}\end{equation*}
where the first line follows from Assumption~\ref{assum:proxy}d (exclusion) and the second line follows from Assumption~\ref{assum:proxy}b (placebo selection on unobservables).

To show that $Y \indep Z \mid D, U$, first note that,
\begin{equation*}\begin{aligned}
\p{Y=y\mid  D=d,Z=z,U}&=\p{Y(A,d,U,\eta_y)=y\mid  D=d,Z=z,U}\\
&=\p{Y(A,d,U,\eta_y)=y\mid  D=d,U},
\end{aligned}\end{equation*}
where the first line follows from Assumption~\ref{assum:proxy}d (exclusion) and the second line follows from Assumption~\ref{assum:proxy}b (placebo selection on unobservables) and the result we have already shown: $(A \indep Z \mid D, U)$. 

    \item To show the second claim ($W \indep D, Z \mid U$), we have that,
    \begin{equation*}\begin{aligned}
        \p{W = w \mid D = d, Z = z, U} &= \p{W(U, \eta_{w}) = w \mid D = d, Z = z, U} \\
        &= \p{W(U, \eta_{w}) = w \mid U},
    \end{aligned}\end{equation*}
    where the first line follows from Assumption~\ref{assum:proxy}d (exclusion) and the second line follows from Assumption~\ref{assum:proxy}b (placebo selection on unobservables).
    
    \item Similarly, we prove the last claim ($A\indep U \mid D$): 
    \begin{equation*}\begin{aligned}
    \p{A \mid D=d,U=u} &= \p{A(d,\eta_{a}) \mid D=d,U=u} \\
    &= \p{A(d,\eta_{a}) \mid D=d}.
    \end{aligned}\end{equation*}
    where the first line follows from Assumption~\ref{assum:proxy}d (exclusion) and the second line follows from Assumption~\ref{assum:indep}b ($\eta_{a} \indep U \mid D$). 
\end{enumerate}

\end{proof}

\subsection{Limit}

\begin{proof}[Proof of Lemma~\ref{lemma:limit} (limit)]

    Assume a fuzzy design. We proceed in steps.

    \begin{enumerate}
        \item Notice that under Assumption~\ref{assum:proxy} (placebo variables), we are able to write,
    $$Y = Y(0,D,U,\eta_y) + A(D, \eta_{a})\left(Y(1, D, U, \eta_{y}) - Y(0,D,U,\eta_y)\right).$$
        \item Then, we have that,
    \begin{equation*}\begin{aligned}
        &\E{Y \mid D = d^{*} + \epsilon, U} - \E{Y \mid D = d^{*} - \epsilon, U} \\
        &= \E{Y(0,d^{*} + \epsilon,U,\eta_y) \mid D = d^{*} + \epsilon, U} - \E{Y(0,d^{*} - \epsilon,U,\eta_y) \mid D = d^{*} - \epsilon, U}\\
        &\quad + \E{A(D,\eta_a)\left(Y(1, D , U, \eta_{y}) - Y(0, D ,U,\eta_y)\right) \mid D = d^{*} + \epsilon, U} \\
        &\quad - \E{A(D,\eta_a)\left(Y(1, D, U, \eta_{y}) - Y(0, D ,U,\eta_y)\right) \mid D = d^{*} - \epsilon, U} \qquad \text{(I)}\\
        &= \E{Y(0,d^{*} + \epsilon,U,\eta_y) \mid D = d^{*} + \epsilon, U} - \E{Y(0,d^{*} - \epsilon,U,\eta_y) \mid D = d^{*} - \epsilon, U}\\
        &\quad + \E{A(D, \eta_{a}) \mid D = d^{*} + \epsilon, U} \E{Y(1, D , U, \eta_{y}) - Y(0, D ,U,\eta_y) \mid D = d^{*} + \epsilon, U} \\
        &\quad - \E{A(D, \eta_{a}) \mid D = d^{*} - \epsilon, U} \E{Y(1, D, U, \eta_{y}) - Y(0, D ,U,\eta_y) \mid D = d^{*} - \epsilon, U} \qquad \text{(II)}\\
        &= \E{Y(0,d^{*} + \epsilon,U,\eta_y) \mid D = d^{*} + \epsilon, U} - \E{Y(0,d^{*} - \epsilon,U,\eta_y) \mid D = d^{*} - \epsilon, U}\\
        &\quad + \E{A(D, \eta_{a}) \mid D = d^{*} + \epsilon} \E{Y(1, D , U, \eta_{y}) - Y(0, D ,U,\eta_y) \mid D = d^{*} + \epsilon, U} \\
        &\quad - \E{A(D, \eta_{a}) \mid D = d^{*} - \epsilon} \E{Y(1, D, U, \eta_{y}) - Y(0, D ,U,\eta_y) \mid D = d^{*} - \epsilon, U} \qquad \text{(III)},
    \end{aligned}\end{equation*}
    where (I) follows from the expansion in step 1. The factoring of $A(D, \eta_{a})$ in (II) follows from Assumption~\ref{assum:indep} ($\eta_{a} \indep \eta_{y}, U \mid D = d$) and weak union, whereby $\eta_{a} \indep \eta_{y} \mid D = d, U$. Dropping the conditioning on $U$ in (III) follows from Lemma~\ref{lemma:indep}c ($A \indep U \mid D$).
    \item By Assumptions~\ref{assum:rdd} (RDD) and \ref{assum:stable} (continuity), we can take the limit of the equality above to show,
    \end{enumerate}
    \begin{equation*}\begin{aligned}
    &\lim_{\epsilon \downarrow 0}\left[\E{Y \mid D = d^{*} + \epsilon, U} - \E{Y \mid D = d^{*} - \epsilon, U}\right] \\
    &= \lim_{\epsilon \downarrow 0}\left[\E{Y(0,d^{*} + \epsilon,U,\eta_y) \mid D = d^{*} + \epsilon, U} - \E{Y(0,d^{*} - \epsilon,U,\eta_y) \mid D = d^{*} - \epsilon, U}\right] \\
    &\quad + \lim_{\epsilon \downarrow 0}\left[ \E{A(D, \eta_{a}) \mid D = d^{*} + \epsilon} \right] \lim_{\epsilon \downarrow 0}\left[ \E{Y(1, D , U, \eta_{y}) - Y(0, D ,U,\eta_y) \mid D = d^{*} + \epsilon, U} \right] \\
    &\quad - \lim_{\epsilon \downarrow 0}\left[ \E{A(D, \eta_{a}) \mid D = d^{*} - \epsilon} \right] \lim_{\epsilon \downarrow 0}\left[ \E{Y(1, D , U, \eta_{y}) - Y(0, D ,U,\eta_y) \mid D = d^{*} - \epsilon,U} \right] \\
    &= \lim_{\epsilon \downarrow 0}\left[\E{A(D, \eta_{a}) \mid D = d^{*} + \epsilon} -  \E{A(D, \eta_{a}) \mid D = d^{*} - \epsilon}\right]\\
    &\quad \times \E{Y(1, D , U, \eta_{y}) - Y(0, D ,U,\eta_y) \mid D = d^{*}, U}.
    \end{aligned}\end{equation*}
Specifically, to go from the second line to the third line, the initial terms cancel, and the later terms factorize.
    
    Rearranging the above, we then find that,
    \begin{equation*}\begin{aligned}
    &\E{Y(1, D , U, \eta_{y}) - Y(0, D ,U,\eta_y) \mid D = d^{*}, U} \\
    &= \frac{\lim_{\epsilon \downarrow 0}\left[\E{Y \mid D = d^{*} + \epsilon, U} - \E{Y \mid D = d^{*} - \epsilon, U}\right]}{\lim_{\epsilon \downarrow 0}\left[\E{A \mid D = d^{*} + \epsilon} -  \E{A \mid D = d^{*} - \epsilon}\right]}.
    \end{aligned}\end{equation*}
    \item This implies
    \begin{equation*}\begin{aligned}
     &\E{Y(1, D , U, \eta_{y}) - Y(0, D ,U,\eta_y) \mid D = d^{*}} \\
     &=
     \E{\E{Y(1, D , U, \eta_{y}) - Y(0, D ,U,\eta_y) \mid D = d^{*}, U}\mid D = d^{*} } \\
     &=
     \frac{\E{\lim_{\epsilon \downarrow 0}\left[ \E{Y \mid D = d^{*} + \epsilon, U} - \E{Y \mid D = d^{*} - \epsilon, U}\right] \mid D = d^{*}  }}{\lim_{\epsilon \downarrow 0}\left[ \E{A \mid D = d^{*} + \epsilon} -  \E{A \mid D = d^{*} - \epsilon}  \right] }\\
    \end{aligned}\end{equation*}
    where first equality follows from the law of iterated expectations and the second equality substitutes in our result above.\\

    The argument for the sharp design is the same, recognizing that Assumption~\ref{assum:indep} vacuously holds and simplifying the denominator.
\end{proof}

\subsection{Factuals}

\begin{lemma}[Towards factuals]\label{lemma:for_factuals}
     Lemma~\ref{lemma:indep} implies for all $d\in\mathcal{D}(\epsilon)$ and all $z\in\mathcal{Z}(\epsilon)$,
     \begin{equation*}\begin{aligned}
         \pr{y\mid  d,z} &= \int \pr{y\mid d,u}\dv \pr{u\mid  d,z}\\
         \pr{w\mid  d,z} &=  \int \pr{w\mid d, u} \dv \pr{u\mid  d,z} = \int \pr{w\mid u} \dv \pr{u\mid  d,z}
     \end{aligned}\end{equation*}
\end{lemma}

\begin{proof}[Proof of Lemma~\ref{lemma:for_factuals} (towards factuals)]

    To lighten notation, we abbreviate $(D=d,Z=z)=(d,z)$ and similarly for other variables.
    Then, we have that,
    \begin{equation*}\begin{aligned}
        \pr{y\mid  d,z}
        &=\int \pr{y,u\mid  d,z}\dv u 
        =\int \pr{y\mid  d,z,u}\dv \pr{u\mid  d,z}
        =\int \pr{y\mid  d,u}\dv \pr{u\mid  d,z} \\
        \pr{w\mid  d,z}
        &= \int  \pr{w,u\mid  d,z}\dv u 
        = \int  \pr{w\mid  d,z,u}\dv \pr{u\mid  d,z}
        =\int  \pr{w\mid d, u} \dv \pr{u\mid  d,z}\\
        &=\int  \pr{w\mid u} \dv \pr{u\mid  d,z},
    \end{aligned}\end{equation*}
    where we use that $Y \indep Z \mid U, D$ and $W \indep D, Z \mid U$  from Lemma~\ref{lemma:indep}.
\end{proof}

\begin{proof}[Proof of Lemma~\ref{lemma:factuals} (factuals I)]
   We have that,
    \begin{equation*}\begin{aligned}
     \E{Y\mid  d,z}
     &=\int y \dv \pr{y\mid  d,z} \\
     &=\int y \dv \pr{y\mid  d,u}\dv \pr{u\mid  d,z} \qquad \text{(I)}\\
         &=\int \E{Y\mid  d,u}\dv \pr{u\mid  d,z} \\
         &=\int h_{0}(d - d^{*},w)\dv \pr{w\mid  d,u}\dv \pr{u\mid  d,z} \qquad \text{(II)}\\
         &=\int h_{0}(d - d^{*}, w)\dv \pr{w\mid  d,z}  \qquad \text{(III)},
    \end{aligned}\end{equation*}
    where (I) and (III) follow from Lemma~\ref{lemma:for_factuals} (towards factuals), and (II) follows from Assumption~\ref{assum:bridge} (confounding bridge).
\end{proof}

\begin{proof}[Proof of Lemma~\ref{lemma:factuals2} (factuals II)]
    Suppose a solution to $
    \E{Y\mid d,z}=\int h_{0}(d - d^{*}, w)\dv \pr{w\mid d,z}$ exists. Then by  Lemma~\ref{lemma:for_factuals} (towards factuals),
    \begin{equation*}\begin{aligned}
         \E{Y\mid d,z}
         &=\int h_{0}(d - d^{*}, w)\dv \pr{w\mid d,z} \\
         &= \int h_{0}(d - d^{*}, w)\dv \pr{w\mid d, u}\dv \pr{u\mid  d,z}.
    \end{aligned}\end{equation*}

    Furthermore, using Lemma~\ref{lemma:for_factuals} (towards factuals), we can also write that,
    \begin{equation*}\begin{aligned}
         \E{Y\mid d,z}
         &=\int y \dv \pr{y\mid  d,z} \\
         &=\int y \dv \pr{y\mid  d, u} \dv \pr{u\mid  d,z} \\
         &=\int \mathbb{E}(Y\mid  d,u)\dv \pr{u\mid  d,z}.
    \end{aligned}\end{equation*}
    Therefore by Assumption~\ref{assum:complete} (completeness), we can equate the objects within the integrals.
\end{proof}

\subsection{Main result}

\begin{proof}[Proof of Theorem~\ref{thm:identify} (placebo identification)]

We proceed in steps.

\begin{enumerate}
    \item First, note that
    \begin{equation*}\begin{aligned}
     \E{\int h_{0}(d - d^{*}, w) \dv \pr{w\mid D = d, U  } \bigmid D = d'} &= \E{\int h_{0}(d - d^{*}, w) \dv \pr{w\mid U  } \bigmid D = d'} \\
     &= \E{\int h_{0}(d - d^{*}, w) \dv \pr{w\mid D=d', U  } \bigmid D = d'} \\
     &= \int h_{0}(d - d^{*}, w)\dv \pr{w \mid D = d'},
    \end{aligned}\end{equation*}
    where the first and second line follow from Lemma~\ref{lemma:indep}b $(W \indep D \mid U)$ and the last line uses the law of iterated expectations.

    \item Also recall from the proof of Lemma~\ref{lemma:limit} (limit) we showed that,
\begin{equation*}\begin{aligned}
        &\lim_{\epsilon \downarrow 0}\left[\E{Y \mid D = d^{*} + \epsilon, U} - \E{Y \mid D = d^{*} - \epsilon, U}\right]\\
        &= \lim_{\epsilon \downarrow 0}\left[\E{A(D, \eta_{a}) \mid D = d^{*} + \epsilon} -  \E{A(D, \eta_{a}) \mid D = d^{*} - \epsilon}\right]\\
        &\quad \times \E{Y(1, D , U, \eta_{y}) - Y(0, D ,U,\eta_y) \mid D = d^{*}, U}.
    \end{aligned}\end{equation*}
The limits on the right hand side exist due to Assumption~\ref{assum:rdd} (RDD). Therefore the limit on the left hand side also exists.

    \item Therefore, we can write the numerator of the expression from Lemma~\ref{lemma:limit} (limit) as,\small
    \begin{equation*}\begin{aligned}
    &\E{ \lim_{\epsilon \downarrow 0} \left\{ \E{Y \mid D = d^{*} + \epsilon, U} -  \E{Y \mid D = d^{*} - \epsilon, U} \bigmid D = d^{*} \right\}} \\
    &= \E{ \lim_{\epsilon \downarrow 0} \left\{\int h_{+}(\epsilon,w)\dv \pr{w\mid D=d^*+\epsilon, U} - \int h_{-}(-\epsilon,w)\dv \pr{w\mid D=d^*-\epsilon, U} \right\} \bigmid D = d^{*} } \\
    &= \lim_{\epsilon \downarrow 0}\E{ \int h_{+}(\epsilon,w)\dv \pr{w\mid D=d^*+\epsilon, U} -   \int h_{-}(-\epsilon,w)\dv \pr{w\mid D=d^*-\epsilon, U} \bigmid D = d^{*} }\\
    &=  \lim_{\epsilon \downarrow 0} \int h_{+}(\epsilon, w) \dv \pr{w \mid D = d^{*}} -  \lim_{\epsilon \downarrow 0} \int  h_{-}(-\epsilon,w)\dv \pr{w\mid D=d^*}
    \end{aligned}\end{equation*}
    \normalsize
    where the first equality follows from Assumption~\ref{assum:bridge} (confounding bridge), the second equality uses the dominated convergence theorem and the boundedness of Assumption~\ref{assum:bridge} (confounding bridge), and the last equality follows from step 1.
    
\item Finally, we argue that the overall expression is unique. To begin, note that $\E{Y \mid d, U}$ is unique. Furthermore, in the proof of Lemma~\ref{lemma:factuals2} (factuals II), we showed that  $\E{Y\mid  d,U}=\int h_{0}(d - d^{*}, w)\dv \mathbb{P}(w\mid d, U) $ almost surely by appealing to Assumption~\ref{assum:complete}, so the latter integral is unique. 
\end{enumerate}
\end{proof}

\section{Relaxing Assumption~\ref{assum:proxy}b}\label{sec:cont}

Assumption~\ref{assum:proxy}b implies the necessary condition $Z\indep \eta_y, \eta_a,\eta_w |D,U$. 

In this appendix, we show that this necessary condition, together with an additional continuity condition, lead to the same main result. Intuitively, we trade off an assumption on placebo exogeneity for an assumption on placebo continuity, i.e. independence for functional form.

Our goal is to define an alternative set of assumptions under which the conclusion of Theorem~\ref{thm:identify} continues to hold, replacing Assumption~\ref{assum:proxy}b with weaker conditions. We summarize the invariances of the full structural system in the next section. A graphical summary of the causal structure is provided in Figure~\ref{dag:relax_w}.

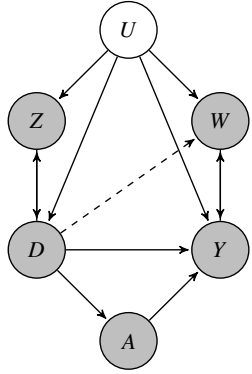
\begin{figure}[htbp]
\begin{center}
\begin{adjustbox}{width=.2\textwidth}
\begin{tikzpicture}[->,>=stealth',shorten >=1pt,auto,node distance=2cm,
                    semithick]
  \tikzstyle{every state}=[draw=black,text=black]

  \node[state]         (u) [fill = white]                        {$U$};
  \node[state]         (w) [below right of=u, fill=lightgray]    {$W$};
  \node[state]         (z) [below left of=u, fill=lightgray]     {$Z$};
  \node[state]         (d) [below of=z, fill=lightgray]             {$D$};
  \node[state]         (y) [below of=w, fill=lightgray]             {$Y$};
  \node[state]         (a) [below right of=d, fill=lightgray]    {$A$};

  \path (u) edge           node {$ $} (z)
            edge           node {$ $} (w)
            edge           node {$ $} (d)
            edge           node {$ $} (y)
        (z) edge           node {$ $} (d)
        (w) edge           node {$ $} (y)
        (d) edge           node {$ $} (y)
            edge           node {$ $} (a)
            edge[dashed]   node {$ $} (w)
            edge           node {$ $} (z)
        (a) edge           node {$ $} (y)
        (y) edge           node {$ $} (w);;
\end{tikzpicture}
\end{adjustbox}
\caption{Additional continuity DAG}
\label{dag:relax_w}
\end{center}
\end{figure}

\subsection{Extension}

\begin{assumption}[Placebo variables]\label{assum_w:proxy}
Assume there exists some $\epsilon>0$ such that for all $d\in \mathcal{D}(\epsilon)$ and for all $z\in \mathcal{Z}(\epsilon)$,
    \begin{enumerate}[label=\alph*]
        \item (Causal consistency): if $A=a$, $D=d$, and $Z=z$ then $Y=Y(a,d,z,U,\eta_y)$, and $W=W(a,d,z,U,\eta_w)$ almost surely. If $D=d$ and $Z=z$ then $A=A(d,z,U,\eta_{A})$ almost surely.
        \item (Placebo selection on unobservables): $Z \indep \eta_{y}, \eta_{a},\eta_w \mid  U, D$. 
        \item (Overlap): if $f(u)>0$ then $f(a,d,z\mid  u)>0$, where we assume the densities exist.
        \item (Exclusion): $Y(a,d,z,U,\eta_y)=Y(a,d,U,\eta_y)$, $W(a,d,z,U,\eta_w)=W(U,\eta_w)$, and $A(d, z, U, \eta_{a}) = A(d, \eta_{a})$ almost surely.
    \end{enumerate}
\end{assumption} 

Assumptions~\ref{assum_w:proxy}a,c,d are the same as those in Assumption~\ref{assum:proxy}. However, Assumption~\ref{assum_w:proxy}b is weaker than Assumption~\ref{assum:proxy}b by weak union. This relaxation makes placebo selection on observables a condition that is exclusively about the placebo treatment.

\begin{lemma}[Necessary independences]\label{lemma_w:indep}
    Assumption~\ref{assum_w:proxy} implies that for all $d \in \mathcal{D}(\epsilon)$ and $z \in \mathcal{Z}(\epsilon)$,
    \begin{enumerate}[label=\alph*]
            \item $Y \indep Z \mid  D, U$
            \item $W \indep  Z \mid D, U$
    \end{enumerate}
    Additionally under Assumption~\ref{assum:indep} it follows that,
        \begin{enumerate}[label=\alph*]
            \setcounter{enumi}{2}
            \item $A \indep U \mid D$
    \end{enumerate}
\end{lemma}

Lemma~\ref{lemma_w:indep}b is weaker than Lemma~\ref{lemma:indep}b by weak union. The rest of the lemma remains the same.

\begin{assumption}[Confounding bridge continuity]\label{assum_w:stable}
    Assume that the mapping \hfill\break $d \mapsto \E{h_{+}(d-d^*,W)\mid D=d, U=u}$ is continuous for all $d \in \mathcal{D}_{+}(\epsilon)$ and all $u$. Similarly, assume that $d \mapsto \E{h_{-}(d-d^*,W) \mid D=d, U=u}$ is continuous for all $d \in \mathcal{D}_{-}(\epsilon)$ and all $u$.
\end{assumption}

This additional continuity condition generalizes the standard RDD continuity condition (Assumption~\ref{assum:stable}). In particular, it replaces the potential outcome function with the confounding bridge function. Intuitively, if confounding is continuous, then we can allow more complex forms of it.

\begin{theorem}[Placebo identification]\label{thm_w:identify}
    Suppose Assumptions~\ref{assum:rdd},~\ref{assum:indep},~\ref{assum:stable},~\ref{assum:bridge},~\ref{assum_w:proxy}, and~\ref{assum_w:stable} hold. Then 
    $$\tau_{0} = \frac{\int h_{+}(0, w) \dv \pr{w \mid D = d^{*}} -  \int  h_{-}(0,w)\dv \pr{w\mid D=d^*}}{\lim_{\epsilon \downarrow 0}\left[\E{A \mid D = d^{*} + \epsilon} -  \E{A \mid D = d^{*} - \epsilon}\right]}$$
    If in addition Assumption~\ref{assum:complete} holds then the integrals of $h_0$ in the numerator are identified. Therefore $\tau_{0}$ is identified. \\
    
    In sharp design, under Assumptions~\ref{assum:rdd},~\ref{assum:stable},~\ref{assum:bridge},~\ref{assum_w:proxy}, and~\ref{assum_w:stable}, then this simplifies to:
    $$\tau_{0} = \lim_{\epsilon \downarrow 0} \int h_{+}(\epsilon, w) \dv \pr{w \mid D = d^{*}} -  \lim_{\epsilon \downarrow 0} \int  h_{-}(-\epsilon,w)\dv \pr{w\mid D=d^*}$$
\end{theorem}

Theorem~\ref{thm_w:identify} recovers our main identification result (Theorem~\ref{thm:identify}) under an alternative set of assumptions, as desired.

\subsection{Proof of extension}

\begin{proof}[Proof of Lemma~\ref{lemma_w:indep}]
    Compared to the proof of Lemma~\ref{lemma:indep}, the only difference is in the second claim ($W \indep Z \mid D, U$). We have that,
    \begin{equation*}\begin{aligned}
        \p{W = w \mid D = d, Z = z, U} &= \p{W(U, \eta_{w}) = w \mid D = d, Z = z, U} \\
        &= \p{W(U, \eta_{w}) = w \mid D, U},
    \end{aligned}\end{equation*}
    where the first line follows from Assumption~\ref{assum_w:proxy}d (exclusion) and the second line follows from Assumption~\ref{assum_w:proxy}b (placebo selection on unobservables).
\end{proof}

\begin{lemma}[Towards factuals]\label{lemma_w:for_factuals}
     Lemma~\ref{lemma_w:indep} implies for all $d\in\mathcal{D}(\epsilon)$ and all $z\in\mathcal{Z}(\epsilon)$,
     \begin{equation*}\begin{aligned}
         \pr{y\mid  d,z} &= \int \pr{y\mid d,u}\dv \pr{u\mid  d,z}\\
         \pr{w\mid  d,z} &=  \int \pr{w\mid d, u} \dv \pr{u\mid  d,z}.
     \end{aligned}\end{equation*}
\end{lemma}

\begin{proof}
    The argument is identical to Lemma~\ref{lemma:for_factuals}, excluding the final step.
\end{proof}

\begin{proof}[Proof of Theorem~\ref{thm_w:identify}]
Notice that Lemma~\ref{lemma:limit} remains unchanged under the conditions of this appendix, since it only appeals to Lemma~\ref{lemma:indep}c, which remains unchanged. Therefore, we proceed in steps similar to the proof of Theorem~\ref{thm:identify}. 

\begin{enumerate}
    \item From the proof of Lemma~\ref{lemma:limit} (limit) we showed that,
    \begin{equation*}\begin{aligned}
            &\lim_{\epsilon \downarrow 0}\left[\E{Y \mid D = d^{*} + \epsilon, U} - \E{Y \mid D = d^{*} - \epsilon, U}\right]\\
            &= \lim_{\epsilon \downarrow 0}\left[\E{A(D, \eta_{a}) \mid D = d^{*} + \epsilon} -  \E{A(D, \eta_{a}) \mid D = d^{*} - \epsilon}\right]\\
            &\quad \times \E{Y(1, D , U, \eta_{y}) - Y(0, D ,U,\eta_y) \mid D = d^{*}, U}.
        \end{aligned}\end{equation*}
    The limits on the right hand side exist due to Assumption~\ref{assum:rdd} (RDD). Therefore the limit on the left hand side also exists.
    
    \item We can write the numerator of the expression from Lemma~\ref{lemma:limit} (limit) as,
    \small
    \begin{equation*}\begin{aligned}
    &\E{ \lim_{\epsilon \downarrow 0} \left\{ \E{Y \mid D = d^{*} + \epsilon, U} - \E{Y \mid D = d^{*} - \epsilon, U} \right\} \bigmid D = d^{*}} \\
    &= \E{ \lim_{\epsilon \downarrow 0} \left\{\int h_{+}(\epsilon,w)\dv \pr{w\mid D=d^*+\epsilon, U}  - \int h_{-}(-\epsilon,w)\dv \pr{w\mid D=d^*-\epsilon, U} \right\} \bigmid D = d^{*} } \\
    &= \E{ \int h_{+}(0,w)\dv \pr{w\mid D=d^*, U} \bigmid D = d^{*}} - \E{ \int h_{-}(0,w)\dv \pr{w\mid D=d^*, U} \bigmid D = d^{*} } \\
    &= \int h_{+}(0,w)\dv \pr{w\mid D=d^*} -\int h_{-}(0,w)\dv \pr{w\mid D=d^*}
    \end{aligned}\end{equation*}
    \normalsize
    where the first equality follows from Assumption~\ref{assum:bridge} (confounding bridge), the second equality uses Assumption~\ref{assum_w:stable}, and the third uses the law of iterated expectations.
    
\item Finally, we argue that the overall expression is unique. To begin, note that $\E{Y \mid d, U}$ is unique. Furthermore, in the proof of Lemma~\ref{lemma:factuals2} (factuals II), we showed that  $\E{Y\mid  d,U}=\int h_{0}(d - d^{*}, w)\dv \mathbb{P}(w\mid d, U) $ almost surely by appealing to Assumption~\ref{assum:complete}, so the latter integral is unique. 
\end{enumerate}
\end{proof}
\section{Relaxing Assumption~\ref{assum:indep}}\label{sec:homog}

Assumption~\ref{assum:indep} holds trivially in the sharp design. In the fuzzy design, however, it imposes a nontrivial restriction on the treatment assignment mechanism.

In this appendix, we show that Assumption~\ref{assum:indep} can be relaxed in the fuzzy design if treatment effects are homogeneous. Intuitively, we trade off a restriction on the treatment mechanism for an assumption on the outcome model, i.e. independence for functional form.

Our goal is to define an alternative set of assumptions under which the conclusion of Lemma~\ref{lemma:limit} continues to hold, replacing Assumption~\ref{assum:indep} with weaker conditions. We summarize the invariances of the full structural system in the next section. A graphical summary of the causal structure is provided in Figure~\ref{dag:nc_homo}.

\begin{figure}[htbp]
\begin{center}
\begin{adjustbox}{width=.2\textwidth}
\begin{tikzpicture}[->,>=stealth',shorten >=1pt,auto,node distance=2cm,
                    semithick]
  \tikzstyle{every state}=[draw=black,text=black]

  \node[state]         (u) [fill = white]                        {$U$};
  \node[state]         (w) [below right of=u, fill=lightgray]    {$W$};
  \node[state]         (z) [below left of=u, fill=lightgray]     {$Z$};
  \node[state]         (d) [below of=z, fill=lightgray]             {$D$};
  \node[state]         (y) [below of=w, fill=lightgray]             {$Y$};
  \node[state]         (a) [below right of=d, fill=lightgray]    {$A$};

  \path (u) edge           node {$ $} (z)
            edge           node {$ $} (w)
            edge           node {$ $} (d)
            edge           node {$ $} (y)
            edge           node {$ $} (a)
        (z) edge           node {$ $} (d)
        (w) edge           node {$ $} (y)
        (d) edge           node {$ $} (y)
            edge           node {$ $} (a)
            edge           node {$ $} (z)
        (a) edge           node {$ $} (y)
        (y) edge           node {$ $} (w);;
\end{tikzpicture}
\end{adjustbox}
\caption{Homogeneous effect DAG}
\label{dag:nc_homo}
\end{center}
\end{figure}
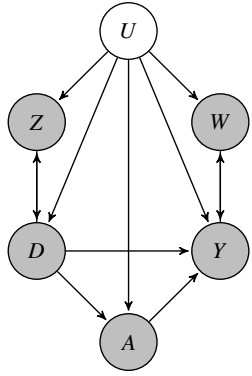

\subsection{Extension}

\begin{assumption}[RDD]\label{assum_h:rdd}
    Assume that the following limits exist,
    $$\lim_{\epsilon \downarrow 0} \left[ \E{A \mid D = d^{*} - \epsilon, U} \right],\quad \lim_{\epsilon \downarrow 0} \left[ \E{A \mid D = d^{*} + \epsilon, U} \right].$$
\end{assumption}

This condition modifies the standard RDD treatment assignment assumption (Assumption~\ref{assum:rdd}). It requires that a discontinuity in expected treatment probability at the cutoff after conditioning on $U$.

\begin{assumption}[Placebo variables]\label{assum_h:proxy}
Assume there exists some $\epsilon>0$ such that for all $d\in \mathcal{D}(\epsilon)$ and for all $z\in \mathcal{Z}(\epsilon)$,
    \begin{enumerate}[label=\alph*]
        \item (Causal consistency): if $A=a$, $D=d$, and $Z=z$ then $Y=Y(a,d,z,U,\eta_y)$, $W=W(a,d,z,U,\eta_w)$. If $D=d$ and $Z=z$ then $A=A(d,z,U,\eta_{A})$ almost surely.
        \item (Selection on unobservables): $Z \indep \eta_{y},  \eta_{a} \mid  U, D$ and $\eta_{w} \indep Z, D \mid U$. 
        \item (Overlap): if $f(u)>0$ then $f(a,d,z\mid  u)>0$, where we assume the densities exist.
        \item (Exclusion): $Y(a,d,z,U,\eta_y)=Y(a,d,U,\eta_y)$, $W(a,d,z,U,\eta_w)=W(U,\eta_w)$, and $A(d, z, U, \eta_{a}) = A(d, U, \eta_{a})$ almost surely.
    \end{enumerate}
\end{assumption}

Assumptions~\ref{assum_h:proxy}a–c are equivalent to those in Assumption~\ref{assum:proxy}. However, Assumption~\ref{assum_h:proxy}d relaxes  Assumption~\ref{assum:proxy}d, allowing for unobserved confounding in the treatment mechanism. Under Assumption~\ref{assum_h:proxy}, locally, our notation simplifies to 
$$
Y=Y(A, D,U,\eta_y),\quad W=W(U,\eta_w),\quad A = A(D, U, \eta_{a}),
$$
and can be summarized by the result below.

\begin{lemma}[Necessary independences]\label{lemma_h:indep}
    Assumption~\ref{assum_h:proxy} implies that for all $d \in \mathcal{D}(\epsilon)$ and $z \in \mathcal{Z}(\epsilon)$, 
    \begin{enumerate}[label = \alph*]
        \item $A, Y \indep Z \mid D, U$
        \item $W \indep D, Z \mid U$
    \end{enumerate}
\end{lemma}

\begin{assumption}[Homogeneity]\label{assum_h:homo}
    Under homogenous treatment effects, we assume that for $d \in \mathcal{D}(\epsilon)$ the treatment effect is constant, i.e.
    $$\tau_{0} = Y(1, d, U, \eta_{y}) - Y(0, d, U, \eta_{y})$$
    almost surely.
\end{assumption}

We use Assumption~\ref{assum_h:homo} in place of  Assumption~\ref{assum:indep} in Section~\ref{sec:double_proxy}.We impose that treatment effects are homogeneous across individuals near the cutoff, while allowing treatment assignment to depend on unobserved confounders that affect the potential outcomes.

Together with Assumptions~\ref{assum_h:rdd}, \ref{assum_h:proxy}, and \ref{assum:stable}, this yields the following identification result:

\begin{lemma}[Limit]\label{lemma_h:limit}
    Under Assumptions \ref{assum_h:rdd}, \ref{assum_h:proxy}, \ref{assum_h:homo}, and \ref{assum:stable},
    $$
    \tau_{0} = \frac{\E{\lim_{\epsilon \downarrow 0} \left[\E{Y \mid D = d^{*} + \epsilon, U} - \E{Y \mid D = d^{*} - \epsilon, U} \right]\mid D = d^{*}}}{\E{\lim_{\epsilon \downarrow 0}\left[ \E{A \mid D = d^{*} + \epsilon, U} - \E{A \mid D = d^{*} - \epsilon, U} \right] \mid D = d^{*}}},
    $$
    when $\E{\lim_{\epsilon \downarrow 0}\left[ \E{A \mid D = d^{*} + \epsilon, U} - \E{A \mid D = d^{*} - \epsilon, U} \right] \mid D = d^{*}} \not= 0$. 
\end{lemma}

Lemma~\ref{lemma_h:limit} expresses the treatment effect $\tau_0$ as the ratio of two discontinuities, each computed after conditioning on $U$ and then averaging over $U$. This formulation permits identification even when the treatment $(A)$ is confounded by unobservables. 

For estimation, we can apply the same confounding bridge approach from Section~\ref{sec:estimation} to recover the counterfactual term $ \E{ \lim_{\epsilon \downarrow 0}[\E{A \mid D = d^{*} - \epsilon, U}] \mid D = d^{*}}$ by using the observed proxies for $U$. The final estimator will look like a ratio of estimators in Section~\ref{sec:estimation}, using $Y$ in the numerator and $A$ in the denominator.

\subsection{Proof of extension}

\begin{proof}[Proof of Lemma~\ref{lemma_h:indep}]
    To show that $A \indep Z \mid D, U$, we write
    \begin{equation*}\begin{aligned}
    \p{A=a\mid  D=d,Z=z,U}&=\p{A(d,U,\eta_a)=a\mid  D=d,Z=z,U}\\
    &=\p{A(d,U,\eta_a)=a \mid  D=d,U},
    \end{aligned}\end{equation*}
    where the second line follows from Assumption~\ref{assum_h:proxy}b (placebo selection on unobservables).
    Otherwise, the result is the same as Lemma~\ref{lemma:indep}.
\end{proof}

\begin{proof}[Proof of Lemma~\ref{lemma_h:limit}]
    Assume a fuzzy design. We proceed in steps.

    \begin{enumerate}
    
    \item Notice that under Assumption~\ref{assum_h:proxy} (proxy controls), we are able to write,
    $$Y = Y(0,D,U,\eta_y) + A(D, U, \eta_{a})\left(Y(1, D, U, \eta_{y}) - Y(0,D,U,\eta_y)\right).$$

    \item Using the decomposition above, we can write
    \begin{equation*}\begin{aligned}
        &\E{Y \mid D = d^{*} + \epsilon, U} - \E{Y \mid D = d^{*} - \epsilon, U} \\
        &= \E{Y(0,d^{*} + \epsilon,U,\eta_y) \mid D = d^{*} + \epsilon, U} - \E{Y(0,d^{*} - \epsilon,U,\eta_y) \mid D = d^{*} - \epsilon, U}\\
        &\quad + \E{A(d^{*} + \epsilon, U, \eta_{a})\left(Y(1, D , U, \eta_{y}) - Y(0, D ,U,\eta_y)\right) \mid D = d^{*} + \epsilon, U} \\
        &\quad - \E{A(d^{*} - \epsilon, U, \eta_{a})\left(Y(1, D, U, \eta_{y}) - Y(0, D ,U,\eta_y)\right) \mid D = d^{*} - \epsilon, U}\\
        &= \E{Y(0,d^{*} + \epsilon,U,\eta_y) \mid D = d^{*} + \epsilon, U} - \E{Y(0,d^{*} - \epsilon,U,\eta_y) \mid D = d^{*} - \epsilon, U}\\
        &\quad + \tau_{0} \E{A(d^{*} + \epsilon, U, \eta_{a}) \mid D = d^{*} + \epsilon, U} -  \tau_{0} \E{A(d^{*} - \epsilon, U, \eta_{a}) \mid D = d^{*} - \epsilon, U},
    \end{aligned}\end{equation*}
    where we use Assumption~\ref{assum_h:homo} to derive the second equality.

    \item By Assumptions~\ref{assum_h:rdd} (RDD) and \ref{assum:stable} (continuity), we can take the limit of the equality above to show,
    \begin{equation*}\begin{aligned}
        &\lim_{\epsilon \downarrow 0}\left[\E{Y \mid D = d^{*} + \epsilon, U} - \E{Y \mid D = d^{*} - \epsilon, U}\right] \\
        &= \lim_{\epsilon \downarrow 0}\left[\E{Y(0,d^{*} + \epsilon,U,\eta_y) \mid D = d^{*} + \epsilon, U} - \E{Y(0,d^{*} - \epsilon,U,\eta_y) \mid D = d^{*} - \epsilon, U}\right] \\
        &\quad + \tau_{0}\lim_{\epsilon \downarrow 0}\left[ \E{A(D, U, \eta_{a}) \mid D = d^{*} + \epsilon, U} - \E{A(D, U, \eta_{a}) \mid D = d^{*} - \epsilon, U}\right] \\
        &= \tau_{0}\lim_{\epsilon \downarrow 0}\left[ \E{A(D, U, \eta_{a}) \mid D = d^{*} + \epsilon, U} - \E{A(D, U, \eta_{a}) \mid D = d^{*} - \epsilon, U}\right].
    \end{aligned}\end{equation*}
    Specifically, to go from the second line to the third line, the initial terms cancel, and the later terms factorize.
    
  \item Taking conditional expectations of both sides, we have that 
    \begin{equation*}\begin{aligned}
        &\E{\lim_{\epsilon \downarrow 0}\left[\E{Y \mid D = d^{*} + \epsilon, U} - \E{Y \mid D = d^{*} - \epsilon, U}\right] \mid D = d^{*}} \\
        &= \tau_{0} \E{\lim_{\epsilon \downarrow 0}\left[ \E{A(D, U, \eta_{a}) \mid D = d^{*} + \epsilon, U} - \E{A(D, U, \eta_{a}) \mid D = d^{*} - \epsilon, U}\right] \mid D = d^{*}}.
    \end{aligned}\end{equation*}
    Rearranging and gives
    $$
     \tau_{0} = \frac{\E{\lim_{\epsilon \downarrow 0} \left[\E{Y \mid D = d^{*} + \epsilon, U} - \E{Y \mid D = d^{*} - \epsilon, U} \right]\mid D = d^{*}}}{\E{\lim_{\epsilon \downarrow 0}\left[ \E{A \mid D = d^{*} + \epsilon, U} - \E{A \mid D = d^{*} - \epsilon, U} \right] \mid D = d^{*}}}.
    $$
    \end{enumerate}
\end{proof}
\section{Equivalence proof}\label{sec:decomposition}

\subsection{Notation}

Let $\mathbf{e}_{0} = (1, 0)^{\top}$ and $\mathbf{e}_{1} = (0, 1)^{\top}$. Furthermore, let $q=\dim(w)$ and define 
\begin{equation*}\begin{aligned}
    \omega_{i, +} &= \frac{1}{h_{n}}\1{D_{i} \geq d^{*}} \K{\frac{\vert{D_{i} - d^{*}}}{h_{n}}} \in \mathbb{R}\\
    R_{i, 1} &= \left(1, \left(\frac{D_{i} - d^{*}}{h_{n}}\right) \right)^{\top} \in \mathbb{R}^{2}\\
    R_{i, 2} &= \left(1, \left(\frac{D_{i} - d^{*}}{b_{n}}\right), \left(\frac{D_{i} - d^{*}}{b_{n}}\right)^{2} \right)^{\top} \in \mathbb{R}^{3} \\
    \mathbf{H}_{p} &= \diag\left(1, \dots, h_{n}^{p}\right) \in \mathbb{R}^{(p + 1) \times (p + 1)}\\
\end{aligned}\hspace{4em}\begin{aligned}
    \mathbf{K}_{+} &= \diag(\omega_{1,+}, \dots, \omega_{n,+})\in \mathbb{R}^{n \times n}\\
    \mathbf{R}_{p} &= \left(R_{1, p}, \dots, R_{n, p} \right)^{\top} \in \mathbb{R}^{n \times (p + 1)}\\
    \mathbf{Z} &= \left( Z_{1}, \dots,  Z_{n}\right)^{\top} \in \mathbb{R}^{n \times q}\\
    \mathbf{W} &= \left(W_{1}, \dots, W_{n}\right)^{\top} \in \mathbb{R}^{n \times q}.\\
\end{aligned}\end{equation*}

\subsection{Estimator recap}

We estimate $\nu_{+}^{\top} = (\alpha_{+, 0}, \alpha_{+, 1}, \gamma_{+}^{\top})$ by solving a local linear IV objective function with first order conditions,
\begin{equation*}\begin{aligned}
    \frac{1}{n} \sum_{i = 1}^{n}  \omega_{i, +} \begin{bmatrix}
     R_{i, 1} \\  Z_{i}
    \end{bmatrix} \left(Y_{i} -  R_{i, 1}^{\top} \mathbf{H}_{1} \hat{\alpha}_{+} - W_{i}^{\top}\hat{\gamma}_{+} \right) &= \pmb{0}.
\end{aligned}\end{equation*}

Then we have that,
\begin{equation*}\begin{aligned}
    \hat{\nu}_{+} &= \left(\begin{bmatrix}
        \mathbf{R}_{1}^{\top} \\
        \mathbf{Z}^{\top}
    \end{bmatrix}\mathbf{K}_{+} \begin{bmatrix}
        \mathbf{R}_{1} & \mathbf{W}
    \end{bmatrix}\right)^{-1}\left(\begin{bmatrix}
        \mathbf{R}_{1}^{\top} \\
        \mathbf{Z}^{\top}
    \end{bmatrix} \mathbf{K}_{+} \mathbf{Y} \right)
    =\begin{bmatrix}
        \mathbf{R}_{1}^{\top} \mathbf{K}_{+}\mathbf{R}_{1} & \mathbf{R}_{1}^{\top} \mathbf{K}_{+}\mathbf{W} \\
        \mathbf{Z}^{\top} \mathbf{K}_{+} \mathbf{R}_{1} & \mathbf{Z}^{\top} \mathbf{K}_{+} \mathbf{W}
    \end{bmatrix}^{-1} \begin{bmatrix} \mathbf{R}_{1}^{\top} \mathbf{K}_{+} \mathbf{Y} \\ \mathbf{Z}^{\top}\mathbf{K}_{+} \mathbf{Y} \end{bmatrix}.
\end{aligned}\end{equation*}

We estimate $\beta_{+}^{w_{j}}$ using the local linear estimator of \citet{hahn2001identification},
\begin{equation*}\begin{aligned}
    \mathbf{H}_{1}\hat{\beta}_{+}^{w_{j}} &= \argmin_{\beta} \frac{1}{n} \sum_{i = 1}^{n} \omega_{i, +}\left(W_{i,j} - R_{i, 1}^{\top}\beta \right)^{2} = \left(\mathbf{R}_{1}^{\top}\mathbf{K}_{+}\mathbf{R}_{1}\right)^{-1}\mathbf{R}_{1}^{\top}\mathbf{K}_{+}\mathbf{W}_{,j}.
\end{aligned}\end{equation*}
Then, we let 
$\beta_{+, 0}^{w} = \left(\beta_{+, 0}^{w_{1}}, \dots, \beta_{+, 0}^{w_{q}}\right)^{\top} \in \mathbb{R}^{q}$ and $\beta_{+}^{w} = \left(\beta_{+}^{w_{1}}, \dots, \beta_{+}^{w_{q}}\right) \in \mathbb{R}^{2 \times q}$.

\subsection{Main result}

\begin{proof}[Proof of Proposition~\ref{prop:decomp} (estimator decomposition)]

We proceed in steps, using regression algebra. 
\begin{enumerate}
    \item Using the formula for the inverse of a block matrix,
\small
\begin{equation*}\begin{aligned}
    \hat{\nu}_{+} &= \begin{bmatrix}
        \mathbf{R}_{1}^{\top} \mathbf{K}_{+}\mathbf{R}_{1} & \mathbf{R}_{1}^{\top} \mathbf{K}_{+}\mathbf{W} \\
        \mathbf{Z}^{\top} \mathbf{K}_{+} \mathbf{R}_{1} & \mathbf{Z}^{\top} \mathbf{K}_{+} \mathbf{W}
    \end{bmatrix}^{-1} \begin{bmatrix} \mathbf{R}_{1}^{\top} \mathbf{K}_{+} \mathbf{Y} \\ \mathbf{Z}^{\top}\mathbf{K}_{+} \mathbf{Y} \end{bmatrix} \\
    &= \begin{bmatrix}
    (\mathbf{R}_{1}^{\top} \mathbf{K}_{+} \mathbf{R}_{1})^{-1} + (\mathbf{R}_{1}^{\top} \mathbf{K}_{+} \mathbf{R}_{1})^{-1} \mathbf{R}_{1}^{\top} \mathbf{K}_{+} \mathbf{W} \, \mathbf{Q}_+^{-1} \, \mathbf{Z}^{\top} \mathbf{K}_{+} \mathbf{R}_{1} (\mathbf{R}_{1}^{\top} \mathbf{K}_{+} \mathbf{R}_{1})^{-1}
    &
    - (\mathbf{R}_{1}^{\top} \mathbf{K}_{+} \mathbf{R}_{1})^{-1} \mathbf{R}_{1}^{\top} \mathbf{K}_{+} \mathbf{W} \, \mathbf{Q}_+^{-1}
    \\
    - \mathbf{Q}_+^{-1} \mathbf{Z}^{\top} \mathbf{K}_{+} \mathbf{R}_{1} (\mathbf{R}_{1}^{\top} \mathbf{K}_{+} \mathbf{R}_{1})^{-1}
    &
    \mathbf{Q}_+^{-1} \end{bmatrix}  \begin{bmatrix} \mathbf{R}_{1}^{\top} \mathbf{K}_{+} \mathbf{Y} \\ \mathbf{Z}^{\top}\mathbf{K}_{+} \mathbf{Y} \end{bmatrix}
\end{aligned}\end{equation*}
\normalsize
where $\mathbf{Q}_+ =\mathbf{Z}^{\top} \mathbf{K}_{+} \left\{\mathbf{I} - \mathbf{R}_{1} (\mathbf{R}_{1}^{\top} \mathbf{K}_{+} \mathbf{R}_{1})^{-1} \mathbf{R}_{1}^{\top} \mathbf{K}_{+} \right\} \mathbf{W}$ is the Schur complement. Then, it follows that for $\hat{\alpha}_{+} = (\hat{g}_{+}(0), h_{n} \hat{g}_{+}^{(1)}(0))^{\top}$, we have that,
\begin{equation*}\begin{aligned}
   \hat{\alpha}_{+} &= (\mathbf{R}_{1}^{\top} \mathbf{K}_{+} \mathbf{R}_{1})^{-1}\mathbf{R}_{1}^{\top}\mathbf{K}_{+}\mathbf{Y} \\
    &\quad + (\mathbf{R}_{1}^{\top} \mathbf{K}_{+} \mathbf{R}_{1})^{-1} \mathbf{R}_{1}^{\top} \mathbf{K}_{+} \mathbf{W} \, \mathbf{Q}_+^{-1} \, \mathbf{Z}^{\top} \mathbf{K}_{+} \mathbf{R}_{1} (\mathbf{R}_{1}^{\top} \mathbf{K}_{+} \mathbf{R}_{1})^{-1} \mathbf{R}_{1}^{\top}\mathbf{K}_{+}\mathbf{Y} \\
    &\quad - (\mathbf{R}_{1}^{\top} \mathbf{K}_{+} \mathbf{R}_{1})^{-1} \mathbf{R}_{1}^{\top} \mathbf{K}_{+} \mathbf{W} \, \mathbf{Q}_+^{-1} \mathbf{Z}^{\top}\mathbf{K}_{+}\mathbf{Y}\\
    &= \hat{\beta}_{+}^y + \hat{\beta}_{+}^{w} \left\{\mathbf{Q}_+^{-1} \, \mathbf{Z}^{\top} \mathbf{K}_{+} \mathbf{R}_{1} (\mathbf{R}_{1}^{\top} \mathbf{K}_{+} \mathbf{R}_{1})^{-1} \mathbf{R}_{1}^{\top}\mathbf{K}_{+}\mathbf{Y}
    - \, \mathbf{Q}_+^{-1} \mathbf{Z}^{\top}\mathbf{K}_{+}\mathbf{Y}\right\}\\
    &=\hat{\beta}_{+}^y - \hat{\beta}_{+}^{w} \left[\mathbf{Q}_+^{-1} \mathbf{Z}^{\top}\mathbf{K}_{+}\{\mathbf{I} - \mathbf{R}_{1}(\mathbf{R}_{1}^{\top}\mathbf{K}_{+}\mathbf{R}_{1})^{-1}\mathbf{R}_{1}^{\top}\mathbf{K}_{+}\}\mathbf{Y}\right]\\
    &= \hat{\beta}_{+}^{y} - \hat{\beta}_{+}^{w} \hat{\gamma}_{+},
\end{aligned}\end{equation*}
where $\hat{\beta}_{+}^{y}$ and $\hat{\beta}_{+}^{w_{j}}$ solve the local linear objective functions,
\begin{equation*}\begin{aligned}
    \hat{\beta}_{+}^{y} = \argmin_{\beta}  \sum_{i = 1}^{n} \omega_{i, +}\left(Y_{i} - R_{i, 1}^{\top}\beta \right)^{2}, \qquad
    \hat{\beta}_{+}^{w_{j}} = \argmin_{\beta} \sum_{i = 1}^{n} \omega_{i, +}\left(W_{i, j} - R_{i, 1}^{\top}\beta \right)^{2}.
\end{aligned}\end{equation*}

    \item We can go through the same steps to show that $\hat{\alpha}_{-} = \hat{\beta}_{-}^{y} - \hat{\beta}_{-}^{w} \hat{\gamma}_{-}$. 

    \item Plugging these expressions into our estimator for $\hat{\tau}_{\text{pdd}}$ in Equation~\eqref{eq:naive_estimator}, we find that,
\begin{equation*}\begin{aligned}
    \hat{\tau}_{\text{pdd}} &= \hat{\alpha}_{+, 0} + \left(\hat{\beta}_{+, 0}^{w}\right)^{\top} \hat{\gamma}_{+} - \left\{\hat{\alpha}_{-, 0} + \left(\hat{\beta}_{+, 0}^{w}\right)^{\top} \hat{\gamma}_{-}\right\}\\
    &= \hat{\beta}_{+, 0}^{y} - \left(\hat{\beta}_{+, 0}^{w}\right)^{\top} \hat{\gamma}_{+} + \left(\hat{\beta}_{+, 0}^{w}\right)^{\top} \hat{\gamma}_{+} - \left(\hat{\beta}_{-, 0}^{y} - \left(\hat{\beta}_{-, 0}^{w}\right)^{\top} \hat{\gamma}_{-} + \left(\hat{\beta}_{+, 0}^{w}\right)^{\top} \hat{\gamma}_{-}\right)\\
    &= \hat{\beta}_{+, 0}^{y} - \hat{\beta}_{-, 0}^{y} - \left(\hat{\beta}_{+, 0}^{w} -  \hat{\beta}_{-, 0}^{w}\right)^{\top} \hat{\gamma}_{-}\\
    &=\hat{\tau}_{\text{rdd}}^{y} - \left(\hat{\tau}_{\text{rdd}}^{w} \right)^\top \hat{\gamma}_{-}.
\end{aligned}\end{equation*}

\item Finally, we show the claim that $\hat{\gamma}_{-} = \left[\frac{1}{n} \sum_{i = 1}^{n} \omega_{i, -} \{Z_{i} (W_{i}^{\perp})^{\top}\}\right]^{-1} \left\{\frac{1}{n} \sum_{i = 1}^{n}  \omega_{i, -} (Z_{i} Y_{i}^{\perp})\right\}$. Expanding out our expression for $\hat{\gamma}_{-}$, we have that,
\begin{equation*}\begin{aligned}
     \hat{\gamma}_{-}&=\mathbf{Q}_-^{-1} \mathbf{Z}^{\top}\mathbf{K}_{-}\{\mathbf{I} - \mathbf{R}_{1}(\mathbf{R}_{1}^{\top}\mathbf{K}_{-}\mathbf{R}_{1})^{-1}\mathbf{R}_{1}^{\top}\mathbf{K}_{-}\}\mathbf{Y}\\
    &= 
    \left[\mathbf{Z}^{\top}\mathbf{K}_{-}\{\mathbf{I} - \mathbf{R}_{1}(\mathbf{R}_{1}^{\top}\mathbf{K}_{-}\mathbf{R}_{1})^{-1}\mathbf{R}_{1}^{\top}\mathbf{K}_{-}\}\mathbf{W}\right]^{-1} \left[ \mathbf{Z}^{\top}\mathbf{K}_{-}\{\mathbf{I} - \mathbf{R}_{1}(\mathbf{R}_{1}^{\top}\mathbf{K}_{-}\mathbf{R}_{1})^{-1}\mathbf{R}_{1}^{\top}\mathbf{K}_{-}\}\mathbf{Y}\right] \\
    &= \left[\sum_{i = 1}^{n} \omega_{i, -} Z_{i} \left\{W_{i}^{\top} - R_{i, 1}^{\top}(\mathbf{R}_{1}^{\top}\mathbf{K}_{-}\mathbf{R}_{1})^{-1}\mathbf{R}_{1}^{\top}\mathbf{K}_{-}\mathbf{W} \right\} \right]^{-1} 
    \left[\sum_{i = 1}^{n} \omega_{i, -} Z_{i} \left\{Y_{i} - R_{i, 1}^{\top}(\mathbf{R}_{1}^{\top}\mathbf{K}_{-}\mathbf{R}_{1})^{-1}\mathbf{R}_{1}^{\top}\mathbf{K}_{-}\mathbf{Y} \right\} \right]\\
    &= \left[ \frac{1}{n} \sum_{i = 1}^{n} \omega_{i, -} Z_{i} \left\{W_{i}^{\top} - R_{i, 1}^{\top}\left(\frac{1}{n}\sum_{i = 1}^{n} \omega_{i, -} R_{i, 1}R_{i, 1}^{\top}\right)^{-1} \frac{1}{n}\sum_{i = 1}^{n} \omega_{i, -} R_{i, 1} W_{i}^{\top} \right\} \right]^{-1} \\
    &\quad \times \left[ \frac{1}{n} \sum_{i = 1}^{n} \omega_{i, -} Z_{i} \left\{Y_{i} - R_{i, 1}^{\top}\left(\frac{1}{n}\sum_{i = 1}^{n} \omega_{i, -} R_{i, 1}R_{i, 1}^{\top}\right)^{-1} \frac{1}{n}\sum_{i = 1}^{n} \omega_{i, -} R_{i, 1} Y_{i} \right\} \right]\\
    &=\left[\frac{1}{n} \sum_{i = 1}^{n} \omega_{i, -} \{Z_{i} (W_{i}^{\perp})^{\top}\}\right]^{-1} \left\{\frac{1}{n} \sum_{i = 1}^{n}  \omega_{i, -} (Z_{i} Y_{i}^{\perp})\right\}.
\end{aligned}\end{equation*}
\end{enumerate}

\end{proof}
\section{Consistency proof}

\subsection{Main result}

\begin{proof}[Proof of Proposition~\ref{prop:estimand} (estimand)]

We proceed in steps.

\begin{enumerate}
    \item By Proposition~\ref{prop:decomp},
    \begin{equation*}\begin{aligned}
        \hat{\tau}_{\text{pdd}} &= \hat{\beta}_{+, 0}^{y} - \hat{\beta}_{-, 0}^{y} - \left(\hat{\beta}_{+, 0}^{w} -  \hat{\beta}_{-, 0}^{w}\right)^{\top} \hat{\gamma}_{-}.
    \end{aligned}\end{equation*}

    Standard results from e.g. \citet[Section 4.1]{hahn2001identification} show that under our conditions,
    \begin{equation*}\begin{aligned}
        \hat{\beta}_{+, 0}^{y} - \hat{\beta}_{-, 0}^{y} &\inprob \lim_{\epsilon \downarrow 0}\left\{\E{Y \mid D = d^{*} + \epsilon} - \E{Y \mid D = d^{*} - \epsilon} \right\}\\
        \hat{\beta}_{+, 0}^{w} - \hat{\beta}_{-, 0}^{w} &\inprob \lim_{\epsilon \downarrow 0}\left\{\E{W \mid D = d^{*} + \epsilon} - \E{W \mid D = d^{*} - \epsilon} \right\}.
    \end{aligned}\end{equation*}
    Therefore, appealing to the continuous mapping theorem, it suffices to show that $\hat{\gamma}_{-} \inprob \gamma_{-}$. 
    \item Within the proof of Proposition~\ref{prop:decomp}, we have shown that for \hfill\break $\mathbf{Q}_- = \mathbf{Z}^{\top} \mathbf{K}_{-} \left\{\mathbf{I} - \mathbf{R}_{1} (\mathbf{R}_{1}^{\top} \mathbf{K}_{-} \mathbf{R}_{1})^{-1} \mathbf{R}_{1}^{\top} \mathbf{K}_{-} \right\} \mathbf{W}$,
    \begin{equation*}\begin{aligned}
        \hat{\gamma}_{-} &= \mathbf{Q}_-^{-1} \mathbf{Z}^{\top}\mathbf{K}_{-}\{\mathbf{I} - \mathbf{R}_{1}(\mathbf{R}_{1}^{\top}\mathbf{K}_{-}\mathbf{R}_{1})^{-1}\mathbf{R}_{1}^{\top}\mathbf{K}_{-}\}\mathbf{Y}.
    \end{aligned}\end{equation*}
   
    Using Lemma~\ref{lemma:q-lim} ($\mathbf{Q}_-$-limit) and the continuous mapping theorem, we can write that,
    $$n \mathbf{Q}_-^{-1} \inprob \kappa^{-1} \lim_{\epsilon \downarrow 0} \Cov{ Z_{i}, W_{i}^{\top} \mid D = d^{*} - \epsilon}^{-1}, \qquad \kappa = f_{d, +}(d^{*}) \int_{0}^{\infty} K(u) \dv u.$$ 

   By Lemma~\ref{lemma:q-lim} ($\mathbf{Q}_-$-limit), replacing $\mathbf{W}$ with $\mathbf{Y}$, we can analogously show that,
    $$\frac{1}{n} \mathbf{Z}^{\top}\mathbf{K}_{-}\{\mathbf{I} - \mathbf{R}_{1}(\mathbf{R}_{1}^{\top}\mathbf{K}_{-}\mathbf{R}_{1})^{-1}\mathbf{R}_{1}^{\top}\mathbf{K}_{-}\}\mathbf{Y} \inprob \kappa \lim_{\epsilon \downarrow 0} \Cov{ Z_{i}, Y_{i} \mid D_{i} = d^{*} - \epsilon}.$$
    
      Therefore, applying the continuous mapping and dominated convergence theorems, it follows that,
    $$\hat{\gamma}_{-} \inprob \lim_{\epsilon \downarrow 0} \Cov{ Z_{i}, W_{i}^{\top} \mid D = d^{*} - \epsilon}^{-1} \Cov{ Z_{i}, Y_{i} \mid D_{i} = d^{*} - \epsilon}.$$
    
    \item   We verify the claim that $\gamma_{-}$ is the best linear predictor for 
    \begin{equation*}\begin{aligned}
        Y = c + W^{\top} \gamma_{-} + \eta, \qquad \lim_{\epsilon \downarrow 0}  \E{\eta  Z_{i} \mid D_{i} = d^{*} - \epsilon} =0,\qquad  \lim_{\epsilon \downarrow 0} \E{\eta \mid D_{i} = d^{*} - \epsilon} = 0.
    \end{aligned}\end{equation*}
    Substituting for $Y$ in the probability limit gives
    \begin{equation*}\begin{aligned}
        \lim_{\epsilon \downarrow 0} \ &\Cov{ Z_{i}, W_{i}^{\top} \mid D_{i} = d^{*} - \epsilon}^{-1} \Cov{ Z_{i}, Y_{i} \mid D_{i} = d^{*} - \epsilon}\\
        &= \lim_{\epsilon \downarrow 0} \Cov{ Z_{i}, W_{i}^{\top} \mid D_{i} = d^{*} - \epsilon}^{-1}\left\{ \Cov{ Z_{i}, c + W_{i}^{\top}\gamma_{-} + \eta \mid D_{i} = d^{*} - \epsilon}\right\}\\
        &= \lim_{\epsilon \downarrow 0} \Cov{ Z_{i}, W_{i}^{\top} \mid D_{i} = d^{*} - \epsilon}^{-1} \left\{\Cov{ Z_{i}, W_{i}^{\top} \mid D_{i} = d^{*} - \epsilon} \gamma_{-} + \Cov{ Z_{i}, \eta \mid D_{i} = d^{*} - \epsilon}\right\}\\
        &= \gamma_{-}.
    \end{aligned}\end{equation*}
\end{enumerate}
\end{proof}

\begin{proof}[Proof of Corollary~\ref{cor:consistency} (consistency)]

We proceed in steps. We focus only on showing convergence above the cutoff, and convergence below the cutoff follows from the same argument applied to $\alpha_{-, 0}$ and $\gamma_{-}$.\footnote{We will not use cancellation, so that the same argument gives convergence below the cutoff.}  We want to show that,
\begin{equation*}\begin{aligned}
    \hat{\alpha}_{+, 0} + (\hat{\beta}_{+, 0}^{w})^{\top} \hat{\gamma}_{+} &\inprob \lim_{\epsilon \downarrow 0} \E{h_{+}(\epsilon, W) \mid D = d^{*}}.
\end{aligned}\end{equation*}

\begin{enumerate}
\item  We express the limit as a sum.

Substituting in $h_{+}(d - d^{*}, w)$ from Equation~\eqref{eq:partial}, 
\begin{equation*}\begin{aligned}
  \lim_{\epsilon \downarrow 0} \E{h_{+}(\epsilon, W) \mid D = d^{*}} 
    &= g_{+}(0) + \E{W^{\top} \mid D = d^{*}}\gamma_{+}\\
    &= \alpha_{+, 0} + \left(\beta_{+, 0}^{w}\right)^{\top}\gamma_{+}.
\end{aligned}\end{equation*}
    \item It is helpful to characterize a partially linear representation for $Y$. 
    
    By Lemma~\ref{lemma:factuals} (factuals) and Equation~\eqref{eq:partial}, for all $d \in \mathcal{D}_{+}(\epsilon)$ recall that we can write,
    \begin{equation*}\begin{aligned}
        \E{Y - h_{+}(D-d^*, W) \mid D = d, Z = z} &= \E{Y - g_{+}(D_{i} - d^{*}) - W^{\top}\gamma_{+} \mid D = d, Z = z} = 0.
    \end{aligned}\end{equation*}
    Using this moment condition, it follows that for some $\delta \in [0, D - d^{*}]$ we can use a Taylor expansion to represent, 
   \begin{equation*}\begin{aligned}
    Y &= g_{+}(D - d^{*}) + W^{\top}\gamma_{+} + \tilde{\eta} \\
        &= g_{+}(0) + g_{+}^{(1)}(0)(D - d^{*}) + W^{\top}\gamma_{+} + \tilde{\eta} + \frac{1}{2} g_{+}^{(2)}(\delta)(D - d^{*})^{2}\\
        &= \alpha_{+, 0} + \alpha_{+, 1}(D - d^{*}) + W^{\top}\gamma_{+} + \tilde{\eta} + \frac{1}{2} g_{+}^{(2)}(\delta)(D - d^{*})^{2},
    \end{aligned}\end{equation*}
    where $\E{\tilde{\eta} \mid D=d, Z=z} = 0$ for all $d \in \mathcal{D}_{+}(\epsilon)$ and $z \in \mathcal{Z}_{+}(\epsilon)$.

    \item Using the representation in step 2, we now show that $\alpha_{+, 0} = \beta_{+, 0}^{y} - (\beta_{+, 0}^{w})^{\top} \gamma_{+}$. 
    
    Applying the conditional expectation $\E{\cdot \mid Z,D = d^{*} + \epsilon}$  on both sides,
    $$
    \E{Y\mid Z,D = d^{*} + \epsilon}=\alpha_{+, 0} + \alpha_{+, 1}\epsilon + \E{W^{\top} \mid Z ,D = d^{*} + \epsilon }\gamma_{+} + \frac{1}{2} g_{+}^{(2)}(\delta)\epsilon^{2}.
    $$
    Further taking the conditional expectation $\E{\cdot \mid D = d^{*} + \epsilon}$  on both sides and appealing to the law of iterated expectations,
    $$
    \E{Y\mid D = d^{*} + \epsilon}=\alpha_{+, 0} + \alpha_{+, 1}\epsilon + \E{W^{\top} \mid D = d^{*} + \epsilon }\gamma_{+} + \frac{1}{2} g_{+}^{(2)}(\delta)\epsilon^{2}.
    $$
    Finally, taking the limit on both sides, 
    $$
    \lim_{\epsilon \downarrow 0} \E{Y \mid D_{i} = d^{*} + \epsilon} = \alpha_{+, 0} + \lim_{\epsilon \downarrow 0} \E{W_{i}^{\top} \mid D_{i} = d^{*} + \epsilon} \gamma_{+}.
    $$
    Rearranging gives $\alpha_{+, 0} = \beta_{+, 0}^{y} - (\beta_{+, 0}^{w})^{\top}\gamma_{+}$.

  \item Next, we show that $\hat{\gamma}_{+} \inprob \gamma_{+}$. 

    As argued in the proof of Proposition~\ref{prop:estimand} (estimand) we have that,
    \small
    \begin{equation*}\begin{aligned}
        \hat{\gamma}_{+}
        &\inprob \lim_{\epsilon \downarrow 0} \Cov{ Z_{i}, W_{i}^{\top} \mid D = d^{*} + \epsilon}^{-1} \Cov{ Z_{i}, Y_{i} \mid D_{i} = d^{*} + \epsilon}\\
        &= \lim_{\epsilon \downarrow 0} \Cov{ Z_{i}, W_{i}^{\top} \mid D = d^{*} + \epsilon}^{-1} \Cov{ Z_{i},  \alpha_{+, 0} + \alpha_{+, 1}(D_{i} - d^{*}) + W_{i}^{\top}\gamma_{+} + \tilde{\eta} + \frac{1}{2} g_{+}^{(2)}(\delta)(D_{i} - d^{*})^{2} \mid D_{i} = d^{*} + \epsilon}\\
        &= \lim_{\epsilon \downarrow 0} \Cov{ Z_{i}, W_{i}^{\top} \mid D = d^{*} + \epsilon}^{-1} \Cov{ Z_{i}, W_{i}^{\top} \mid D_{i} = d^{*} + \epsilon} \gamma_{+} \\
        &\quad + \lim_{\epsilon \downarrow 0} \Cov{ Z_{i}, W_{i}^{\top} \mid D = d^{*} + \epsilon}^{-1} \Cov{ Z_{i}, \tilde{\eta} \mid D_{i} = d^{*} + \epsilon}\\
        &= \gamma_{+},
    \end{aligned}\end{equation*}
    \normalsize
    where the first equality uses step 2 and the third equality uses
    that 
    \begin{equation*}\begin{aligned}
    \Cov{ Z_{i}, \tilde{\eta} \mid D_{i} = d^{*} + \epsilon} 
    &= \E{Z_{i}\tilde{\eta} \mid D_{i} = d^{*} + \epsilon}-\E{Z_{i} \mid D_{i} = d^{*} + \epsilon}\E{\tilde{\eta} \mid D_{i} = d^{*} + \epsilon} \\
    &= \E{Z_{i} \E{\tilde{\eta} \mid Z_{i}, D_{i} = d^{*} + \epsilon}}
    -\E{Z_{i} \mid D_{i} = d^{*} + \epsilon}\E{\E{\tilde{\eta} \mid Z_i, D_{i} = d^{*} + \epsilon}} \\
    &= 0. \end{aligned}\end{equation*}

    \item  We also show that $\hat{\alpha}_{+, 0} \inprob \alpha_{+, 0}$. 
    
    In the proof of Proposition~\ref{prop:decomp} (estimator decomposition), we showed that $\hat{\alpha}_{+, 0} = \hat{\beta}_{+, 0}^{y} - (\hat{\beta}_{+, 0}^{w})^{\top}\hat{\gamma}_{+}$. In step 3, we showed that $\alpha_{+, 0} = \beta_{+, 0}^{y} - (\beta_{+, 0}^{w})^{\top} \gamma_{+}$. 
    
    Recall that $\hat{\beta}_{+, 0}^{y} \inprob \beta_{+, 0}^{y}$ and $\hat{\beta}_{+, 0}^{w} \inprob \beta_{+, 0}^{w}$ by standard arguments \citep[Section 4.1]{hahn2001identification}. Using that $\hat{\gamma}_{+} \inprob \gamma_{+}$ from step 4, the continuous mapping theorem implies that $\hat{\beta}_{+, 0}^{y} - (\hat{\beta}_{+, 0}^{w})^{\top}\hat{\gamma}_{+} \inprob \beta_{+, 0}^{y} - (\beta_{+, 0}^{w})^{\top} \gamma_{+}$. Therefore, $\hat{\alpha}_{+, 0} \inprob \alpha_{+, 0}$.

    \item Collecting results and applying the continuous mapping theorem we have shown that,
    \begin{equation*}\begin{aligned}
        \hat{\alpha}_{+, 0} + (\hat{\beta}^{w}_{+, 0})^{\top}\hat{\gamma}_{+} &\inprob \alpha_{+, 0} + (\beta_{+, 0}^{w})^{\top} \gamma_{+}\\
        &= \lim_{\epsilon \downarrow 0} \E{h_{+}(\epsilon, W) \mid D = d^{*}}
    \end{aligned}\end{equation*}
    where the convergence uses steps 4 and 5, and the equality uses step 1.
\end{enumerate}

\end{proof}

\subsection{Technical lemma}

\begin{lemma}[$\mathbf{Q}_-$-limit]\label{lemma:q-lim}
    Let Assumptions~\ref{assum:reg}a-c and~\ref{assum:kernel} hold, and assume that $h \rightarrow 0$ and $nh \rightarrow \infty$. Then,
    $$\frac{1}{n} \mathbf{Q}_- \inprob  \kappa  \lim_{\epsilon \downarrow 0} \Cov{ Z_{i}, W_{i}^{\top} \mid D = d^{*} - \epsilon},$$
    for $\kappa = f_{d, -}(d^{*}) \int_{0}^{\infty} K(u) \dv u$.\hfill$\triangleleft$
\end{lemma}

\begin{proof}[Proof of Lemma~\ref{lemma:q-lim}]
   
    We proceed in steps.

\begin{enumerate}

\item  As shown in the proof of Proposition~\ref{prop:decomp},
    $$\mathbf{Q}_- = \mathbf{Z}^{\top} \mathbf{K}_{-} \left\{\mathbf{I} - \mathbf{R}_{1} (\mathbf{R}_{1}^{\top} \mathbf{K}_{-} \mathbf{R}_{1})^{-1} \mathbf{R}_{1}^{\top} \mathbf{K}_{-} \right\} \mathbf{W}.$$  Therefore, we have the decomposition 
    \begin{equation*}\begin{aligned}
        \frac{1}{n}\mathbf{Q}_- &= \frac{1}{n} \sum_{i = 1}^{n} \omega_{i, -}  Z_{i} \left\{W_{i}^{\top} -  R_{i, 1}^{\top} (\mathbf{R}_{1}^{\top} \mathbf{K}_{-} \mathbf{R}_{1})^{-1} \mathbf{R}_{1}^{\top} \mathbf{K}_{-} \mathbf{W} \right\}\\
        &= \frac{1}{n} \sum_{i = 1}^{n} \omega_{i, -}  Z_{i} \left(W_{i}^{\top} -  \mathbf{e}_{0}^{\top}(\mathbf{R}_{1}^{\top} \mathbf{K}_{-} \mathbf{R}_{1})^{-1} \mathbf{R}_{1}^{\top} \mathbf{K}_{-} \mathbf{W}\right) \\
        &\quad - \frac{1}{n} \sum_{i = 1}^{n} \omega_{i, -}  Z_{i} \left\{\left(\frac{D_{i} - d^{*}}{h_{n}} \right)\mathbf{e}^{\top}_{1}(\mathbf{R}_{1}^{\top} \mathbf{K}_{-} \mathbf{R}_{1})^{-1} \mathbf{R}_{1}^{\top} \mathbf{K}_{-} \mathbf{W} \right\}.
    \end{aligned}\end{equation*}
   We analyze each term below.

\item  Under Assumptions~\ref{assum:reg}a and~\ref{assum:kernel}, results from \citet[Section 4.1]{hahn2001identification} imply that 
    \begin{align*}
     &\mathbf{e}_{0}^{\top}(\mathbf{R}_{1}^{\top} \mathbf{K}_{-} \mathbf{R}_{1})^{-1} \mathbf{R}_{1}^{\top} \mathbf{K}_{-} \mathbf{W} \inprob (\beta^w_{-, 0})^{\top}=\lim_{\epsilon \downarrow 0} \E{W_{i}^{\top} \mid D_{i} = d^{*} - \epsilon} \\
        &\frac{1}{h_{n}}\mathbf{e}_{1}^{\top}(\mathbf{R}_{1}^{\top} \mathbf{K}_{-} \mathbf{R}_{1})^{-1} \mathbf{R}_{1}^{\top} \mathbf{K}_{-} \mathbf{W} \inprob (\beta^w_{-, 1})^{\top}.
    \end{align*}

    \item Here, we show that,
    \begin{align*}
    \left(\frac{1}{nh_{n}} \sum_{i = 1}^{n} \1{D_{i} < d^{*}} \K{\frac{D_{i} - d^{*}}{h_{n}}}  Z_{i}(W_{i} -  \beta^w_{-, 0} )^{\top}  \right) &= \frac{1}{h_{n}} \E{\1{D_{i} < d^{*}} \K{\frac{D_{i} - d^{*}}{h_{n}}}  Z_{i}(W_{i} -  \beta^w_{-, 0} )^{\top}} + \op{1}.
    \end{align*}
    
    Set $\mathbf{X}_{i} = \frac{1}{h_{n}}\1{D_{i} < d^{*}} \K{\frac{D_{i} - d^{*}}{h_{n}}} Z_{i}(W_{i} -  \beta^w_{-, 0} )^{\top} \in \R^{q \times q}$. For each component $\ell \in \{(1, 1), \dots, (q, q)\}$, we show $\frac{1}{n}\sum_{i} X_{i, \ell} = \E{X_{i, \ell}} + \op{1}$ by bounding $\Var{\frac{1}{n}\sum_{i}X_{i, \ell}} = \frac{1}{n}\Var{X_{1, \ell}} \leq \frac{1}{n}\norm{X_{1, \ell}}_{L_2}^{2}$ and applying Chebyshev's inequality. Computing the $L_2$ bound, by the law of iterated expectations we have,
        \begin{align*}
            h_{n}\E{X_{i, \ell}^{2}} &= \frac{1}{h_{n}}\E{\1{D_i < d^{*}}\K{\frac{D_{i} - d^{*}}{h_{n}}} Z_{i, k}^{2} \left(W_{i, j} - \beta_{-, 0}^{w_j}\right)^{2}}\\
            &= \int_{-\infty}^{0} \K{u}^{2} f_{d}(uh_{n} + d^{*}) \E{Z_{i, k}^{2} \left(W_{i, j} - \beta_{-, 0}^{w_j} \right)^{2} \mid D = uh_n + d^{*}} \dv u\\
            &\to \rho_{-, z_kw_j}(0) f_{d, -}(d^{*}) \int_{-\infty}^{0} \K{u}^{2} < \infty,
        \end{align*}
        where the squaring is element-wise and the limit/integral swap uses Assumptions~\ref{assum:reg} and~\ref{assum:kernel}. This shows that $\norm{X_{1, \ell}}_{L_2}^{2} = O(1/h_{n})$ and thus $\frac{1}{n}\norm{X_{1, \ell}}_{L_2}^{2} = O\left(\tfrac{1}{nh_{n}}\right) \to 0$. Then from Chebyshev's inequality,
        \begin{align*}
            \pr{\vert{\frac{1}{n} \sum_{i = 1}^{n} (X_{i, \ell} - \E{X_{i, \ell}})} > \epsilon} \leq \frac{1}{\epsilon^{2}} \frac{1}{n} \norm{X_{1, \ell}}_{L_{2}}^{2} \to 0,
        \end{align*}
        and the conclusion follows.

    \item Consider the first term in the decomposition. We will show that
    $$
    \frac{1}{n} \sum_{i = 1}^{n} \omega_{i, -}  Z_{i} \left(W_{i}^{\top} -  \mathbf{e}_{0}^{\top}(\mathbf{R}_{1}^{\top} \mathbf{K}_{-} \mathbf{R}_{1})^{-1} \mathbf{R}_{1}^{\top} \mathbf{K}_{-} \mathbf{W}\right) \inprob  \kappa \lim_{\epsilon \downarrow 0} \Cov{ Z_{i}, W_{i}^{\top} \mid D_{i} = d^{*} - \epsilon}.
    $$
    By steps 2 and 3,
    \begin{equation*}\begin{aligned}
             & \frac{1}{n} \sum_{i = 1}^{n} \omega_{i, -}  Z_{i} \left(W_{i}^{\top} -  \mathbf{e}_{0}^{\top}(\mathbf{R}_{1}^{\top} \mathbf{K}_{-} \mathbf{R}_{1})^{-1} \mathbf{R}_{1}^{\top} \mathbf{K}_{-} \mathbf{W}\right) \\
            &= \frac{1}{n} \sum_{i = 1}^{n} \omega_{i, -}  Z_{i} (W_{i} -  \beta^w_{-, 0} )^{\top} + \op{1} \qquad  \\
            &= \left(\frac{1}{nh_{n}} \sum_{i = 1}^{n} \1{D_{i} < d^{*}} \K{\frac{D_{i} - d^{*}}{h_{n}}}  Z_{i}(W_{i} -  \beta^w_{-, 0} )^{\top}  \right) + \op{1}\\
            &= \frac{1}{h_{n}} \E{\1{D_{i} < d^{*}} \K{\frac{D_{i} - d^{*}}{h_{n}}}  Z_{i}(W_{i} -  \beta^w_{-, 0} )^{\top}} + \op{1}.
        \end{aligned}\end{equation*}
    We use the law of iterated expectations and $u$-substitution with $u = \frac{d - d^{*}}{h_{n}} $ to express the former term as,
    \begin{equation*}\begin{aligned}
        & \frac{1}{h_{n}}\E{\1{D_{i} < d^{*}} \K{\frac{D_{i} - d^{*}}{h_{n}}} \E{ Z_{i}(W_{i} -  \beta^w_{-, 0} )^{\top} \bigmid D_{i}} } \\
        &=\frac{1}{h_{n}}  \int_{-\infty}^{d^*}f_d(d)\1{d< d^{*}} \K{\frac{d- d^{*}}{h_{n}}}  \E{ Z (W  -  \beta^w_{-, 0} )^{\top} \bigmid D=d} \dv d  \\
         &=\frac{1}{h_{n}}  \int_{-\infty}^{0}f_d(d^{*} + uh_{n})\K{-u}  \E{ Z (W  -  \beta^w_{-, 0} )^{\top} \bigmid D=d^{*} + uh_{n}} h_n\dv u  \\
             &=  \int_{-\infty}^{0}f_d(d^{*} + uh_{n})\K{-u}  \E{ Z (W  -  \beta^w_{-, 0} )^{\top} \bigmid D=d^{*} + uh_{n}} \dv u  \\ 
        &= \kappa \lim_{\epsilon \downarrow 0} \Cov{ Z , W ^{\top} \mid D  = d^{*} - \epsilon}+o_p(1).
    \end{aligned}\end{equation*}
    The last line uses the dominated convergence theorem and the boundedness of $\E{Z_{i} W_{i}^{\top} \mid D_{i} = \cdot}$ and $\mu_{-,z}(\cdot)$ from Assumption~\ref{assum:reg}, taking the limit as $h_{n} \rightarrow 0$. Specifically, we argue that
    \begin{align*}
        &\lim_{h_n\downarrow 0} \int_{-\infty}^{0}f_d(d^{*} + uh_{n})\K{-u}  \E{ Z (W  -  \beta^w_{-, 0} )^{\top} \bigmid D=d^{*} + uh_{n}} \dv u \\
        &=  \int_{-\infty}^{0}f_d(d^{*})\K{-u}  \lim_{\epsilon \downarrow 0} \E{ Z (W  -  \beta^w_{-, 0} )^{\top} \bigmid D=d^{*} - \epsilon } \dv u \\
        &=f_d(d^{*}) \lim_{\epsilon \downarrow 0} \E{ Z (W  -  \beta^w_{-, 0} )^{\top} \bigmid D=d^{*} - \epsilon} \int_{-\infty}^{0}\K{-u}   \dv u \\
        &= \kappa \lim_{\epsilon \downarrow 0} \Cov{ Z , W ^{\top} \mid D  = d^{*} - \epsilon}.
    \end{align*}

      \item Consider the second term in the decomposition. We will show that
    $$\frac{1}{n} \sum_{i = 1}^{n} \omega_{i, -} Z_{i} \left\{  \left(\frac{D_{i} - d^{*}}{h_{n}}\right) \mathbf{e}_{1}^{\top}(\mathbf{R}_{1}^{\top} \mathbf{K}_{-} \mathbf{R}_{1})^{-1} \mathbf{R}_{1}^{\top} \mathbf{K}_{-} \mathbf{W} \right\} \inprob 0.$$ 
   By step 2 and an analagous argument as in step 3, 
    \begin{equation*}\begin{aligned}
        &\frac{1}{n} \sum_{i = 1}^{n} \omega_{i, -} Z_{i} \left\{\left(\frac{D_{i} - d^{*}}{h_{n}}\right) \mathbf{e}_{1}^{\top}(\mathbf{R}_{1}^{\top} \mathbf{K}_{-} \mathbf{R}_{1})^{-1} \mathbf{R}_{1}^{\top} \mathbf{K}_{-} \mathbf{W} \right\} \\
         &=\frac{1}{n} \sum_{i = 1}^{n} \omega_{i, -} Z_{i} \left(D_{i} - d^{*}\right) (\beta^w_{-, 1})^{\top} +o_p(1) \\
        &= \frac{1}{n h_{n}} \sum_{i = 1}^{n} \1{D_{i} < d^{*}} \K{\frac{D_{i} - d^{*}}{h_{n}}} Z_{i}  \left(D_{i} - d^{*}\right) (\beta_{-, 1}^{w})^{\top}   + \op{1}\\
        &= \frac{1}{h_{n}} \E{\1{D_{i} < d^{*}} \K{\frac{D_{i} - d^{*}}{h_{n}}} Z_{i}\left(D_{i} - d^{*}\right) (\beta_{-, 1}^{w})^{\top} } + \op{1}.
 \end{aligned}\end{equation*}
 We use the law of iterated expectations and $u$-substitution with $u = \frac{d - d^{*}}{h_{n}} $ to express the former term as
        \begin{equation*}\begin{aligned}
        &\frac{1}{h_{n}} \E{\1{D_{i} < d^{*}} \K{\frac{D_{i} - d^{*}}{h_{n}}} \E{Z_{i} \mid D_{i}}\left(D_{i} - d^{*}\right)  }(\beta_{-, 1}^{w})^{\top} \\
                &= \frac{1}{h_{n}} \int_{-\infty}^{d^*}f_d(d)\1{d< d^{*}} \K{\frac{d- d^{*}}{h_{n}}} \E{Z  \mid D=d}\left(d - d^{*}\right) \dv d \, (\beta_{-, 1}^{w})^{\top}  \\
        &=  \frac{1}{h_{n}} \int_{-\infty}^{0} f_{d}(d^{*} + uh_{n})\K{-u} \, \E{Z \mid D = d^{*} + uh_{n}} \, uh_{n} \,  h_n\dv u \, (\beta_{-, 1}^{w})^{\top}\\
         &=   h_n \int_{-\infty}^{0} f_{d}(d^{*} + uh_{n})\K{-u} \, \E{Z \mid D = d^{*} + uh_{n}} \, u \,  \dv u \, (\beta_{-, 1}^{w})^{\top} \\
        &= \op{1}.
    \end{aligned}\end{equation*}
   The last line uses the dominated convergence theorem and the boundedness of $\mu_{-, z}(\cdot)$ from Assumption~\ref{assum:reg}, taking the limit as $h_{n} \rightarrow 0$. Specifically, we argue that
    \begin{align*}
        \lim_{h_n\downarrow 0}\int_{-\infty}^{0} f_{d}(d^{*} + uh_{n})\K{-u} \, \E{Z \mid D = d^{*} + uh_{n}} \, u \,  \dv u
        &= \int_{-\infty}^{0} f_{d, +}(d^{*})\K{-u} \, \mu_{-, z}(0) \, u \,  \dv u\\
        &=f_{d, -}(d^{*})\mu_{-, z}(0)\int_{-\infty}^{0} \K{-u} \,  \, u \,  \dv u
    \end{align*}
\end{enumerate}

\end{proof}
\section{Bias corrected inference}\label{sec:appendix_bias}

\subsection{Notation}

Here, we compile all the notation used in the appendix. When we introduce a term for the first time, we try to redefine it in the proofs and lemmas in this section.

We write the ratio of the two bandwidths as
\begin{equation*}
    \xi \coloneq \frac{h_{n}}{b_{n}}, \qquad\text{equivalently}\qquad b_{n} = \frac{1}{\xi} h_{n},
\end{equation*}
so that $h_{n}^{p}/b_{n}^{p} = \xi^{p}$ for every $p$. Under the MSE-optimal rates $h_{n} = c_{h} n^{-1/5}$ and $b_{n} = c_{b} n^{-1/5}$, the ratio $\xi = c_{h}/c_{b}$ is a fixed constant.

For the kernel constants, define
\begin{equation*}
    \gamma_{k, j} = \int_{0}^{\infty} u^{k} K(u)^{j} \dv u, \qquad
    \bar{\gamma}_{k, 2}(\xi) = \int_{0}^{\infty} u^{k} K(u) K(\xi u) \dv u.
\end{equation*}
We note that $\bar{\gamma}_{k, 2}(1) = \gamma_{k, 2}$, and when $\xi$ is clear from context we write $\bar{\gamma}_{k, 2} = \bar{\gamma}_{k, 2}(\xi)$.

Throughout, $g_{+}(\cdot)$ and $g_{-}(\cdot)$ denote the outcome regression functions just above and below the cutoff, and $\mu_{+, w_j}(\cdot)$, $\mu_{+, z_j}(\cdot)$ the corresponding conditional-mean functions of $W_{i, j}$ and $Z_{i, j}$; superscript $(m)$ denotes the $m^{\mathrm{th}}$ derivative, evaluated at $0$ for the one-sided limit at the cutoff. We also write $C \coloneq c_{h}^{5/2}$. We let $\mathbf{e}_{k}$ be a conformable matrix of zeros, with a 1 in the $(k + 1)^{\mathrm{th}}$ position. For example, written $\mathbf{e}_{2}^{\top}\Gamma_{+, 2}$, we have $\mathbf{e}_{2} = (0, 0, 1)^{\top}$, and writting $\mathbf{e}_{2}^{\top} \Psi_{+, 2}$, we have $\mathbf{e}_{2} = (0, 0, 1, 0, \dots, 0) \in \R^{3 + q}$.

We collect the basic objects below.
\small\begin{equation*}\begin{aligned}
    f_{d, +}(d^{*}) &\coloneq \lim_{\epsilon \downarrow 0^{+}} f_{d}(d^{*} + \epsilon)\\
    f_{d, -}(d^{*}) &\coloneq \lim_{\epsilon \downarrow 0^{+}} f_{d}(d^{*} - \epsilon)\\
    \mu_{+, w}(0) &\coloneq (\mu_{+, w_{1}}(0), \dots, \mu_{+, w_{q}}(0))^{\top} \in \mathbb{R}^{q} \\ 
    \mu_{+, z}(0) &\coloneq (\mu_{+, z_{1}}(0), \dots, \mu_{+, z_{q}}(0))^{\top} \in \mathbb{R}^{q} \\ 
    \mu_{+, zw^{\top}}(0) &= \pmat{\mu_{+, z_{1}w_{1}}(0) & \cdots & \mu_{+, z_{1}w_{q}}(0)\\ \vdots & \ddots & \vdots \\ \mu_{+, z_{q}w_{1}}(0) & \cdots & \mu_{+, z_{q}w_{q}}(0)} \in \R^{q \times q}
\end{aligned}\hspace{1em}\begin{aligned}
    \hat{\beta}_{+, 0}^{w_j} &\coloneq \mathbf{e}_{0}^{\top} \hat{\Gamma}_{+, 1}^{-1} \left(\frac{1}{n} \sum_{i = 1}^{n} \omega_{+, i} R_{i, 1} W_{i, j}\right) \in \mathbb{R}\\
    \hat{\beta}_{+, 0}^{w} &\coloneq \left(\hat{\beta}_{+, 0}^{w_1}, \dots, \hat{\beta}_{+, 0}^{w_q}\right)^{\top} \in \R^{q}\\
    \mathbf{e}_{+, w} &\coloneq (1, 0, \lim_{\epsilon \to 0^{+}} \E{W_{i}^{\top} \mid D = d^{*} + \epsilon})^{\top} \in \mathbb{R}^{2 + q}\\
    &= (1, 0, \beta_{+, 0}^{w_1}, \dots, \beta_{+, 0}^{w_q})\\
    \mathbf{H}_{1} &\coloneq \diag\left(1, h_n\right) \in \R^{2} \\
    \mathbf{B}_{2} &\coloneq \diag\left(1, b_{n}, b_{n}^{2}\right) \in \mathbb{R}^{3} \\
    \omega_{+, i} &\coloneq \frac{1}{h_{n}} \1{D_{i} \geq d^{*}} \K{\frac{D_{i} - d^{*}}{h_{n}}}\\
    \delta_{+, i} &\coloneq \frac{1}{b_{n}} \1{D_{i} \geq d^{*}} \K{\frac{D_{i} - d^{*}}{b_{n}}}
\end{aligned}\end{equation*}
\normalsize
Here $R_{i, 1} = \left(1, \frac{D_{i} - d^{*}}{h_{n}}\right)^{\top}$ is localized at the bandwidth $h_{n}$, whereas $R_{i, 2} = \left(1, \frac{D_{i} - d^{*}}{b_{n}}, \left(\frac{D_{i} - d^{*}}{b_{n}}\right)^{2}\right)^{\top}$ and $\delta_{+, i}$ are localized at the bandwidth $b_{n}$. For a vector $\mathbf{v}$ we use $\mathbf{v}^{2} = (v_{1}^{2}, \dots, v_{n}^{2})^{\top}$, and $\mathbf{e}_{p}$ denotes a conformable vector of zeros with a $1$ in the $(p + 1)^{\mathrm{th}}$ element.

The $``-"$ analogues of every object below are defined by replacing each $``+"$ subscript, and the associated one-sided limits and weights, with their $``-"$ counterparts. The structural coefficients $\gamma_{+}, \gamma_{-} \in \R^{q}$ on $W_{i}$ (Equation~\eqref{eq:naive_estimator}) are estimated by $\hat{\gamma}_{+}, \hat{\gamma}_{-}$, with components $\gamma_{\pm, j}$. The remaining regression coefficients and their estimators are
\small
\begin{equation*}\begin{aligned}
    \alpha_{+} &\coloneq \mathbf{H}_{1}(g_{+}(0), g_{+}^{(1)}(0))^{\top} = (\alpha_{+, 0}, \alpha_{+, 1})^{\top}\\
    \nu_{+} &\coloneq (\alpha_{+, 0}, \alpha_{+, 1}, \gamma_{+}^{\top})^{\top}, \qquad \hat{\nu}_{+} \coloneq \hat{\Psi}_{+, 1}^{-1}\left(\tfrac{1}{n}\textstyle\sum_{i} \omega_{i, +}\bmat{R_{i, 1}\\Z_i}Y_i\right)\\
    \beta_{+, 0}^{w_j} &\coloneq \mu_{+, w_j}(0), \qquad \beta_{+}^{w_j} \coloneq \mathbf{H}_{1}(\mu_{+, w_j}(0), \mu_{+, w_j}^{(1)}(0))^{\top}\\
    \pi_{+} &\coloneq \mathbf{B}_{2}\left(g_{+}(0), g_{+}^{(1)}(0), \tfrac{1}{2}g_{+}^{(2)}(0)\right)^{\top}\\
    \kappa_{+}^{w_j} &\coloneq \mathbf{B}_{2}\left(\mu_{+, w_j}(0), \mu_{+, w_j}^{(1)}(0), \tfrac{1}{2}\mu_{+, w_j}^{(2)}(0)\right)^{\top}\\
    \hat{V} &\coloneq \hat{\mathbf{s}}_{+}^{\top} \left(\frac{1}{n h_n}\sum_{i} X_{i, +}X_{i, +}^{\top}\right) \hat{\mathbf{s}}_{+} + \hat{\mathbf{s}}_{-}^{\top} \left(\frac{1}{n h_n}\sum_{i} \tilde{X}_{i, -}\tilde{X}_{i, -}^{\top}\right) \hat{\mathbf{s}}_{-}
\end{aligned}\hspace{1.5em}\begin{aligned}
    \mathbf{e}_{2}^{\top}\hat{\pi}_{+} &\coloneq \mathbf{e}_{2}^{\top}\hat{\Psi}_{+, 2}^{-1}\left(\tfrac{1}{n}\textstyle\sum_{i}\delta_{+, i}\bmat{R_{i, 2}\\Z_i}Y_i\right)\\
    \mathbf{e}_{2}^{\top}\hat{\kappa}_{+}^{w_j} &\coloneq \mathbf{e}_{2}^{\top}\hat{\Gamma}_{+, 2}^{-1}\left(\tfrac{1}{n}\textstyle\sum_{i}\delta_{+, i}R_{i, 2}W_{i, j}\right)\\
    \hat{g}_{+}^{(2)}(0) &\coloneq \tfrac{2}{b_n^{2}}\mathbf{e}_{2}^{\top}\hat{\pi}_{+}\\
    \hat{\mu}_{+, w_j}^{(2)}(0) &\coloneq \tfrac{2}{b_n^{2}}\mathbf{e}_{2}^{\top}\hat{\kappa}_{+}^{w_j}\\
    \hat{\mathbf{e}}_{+, w} &\coloneq \left(1, 0, \left(\hat{\beta}_{+, 0}^{w}\right)^{\top}\right)^{\top} \in \R^{2 + q}\\
    V &\coloneq \mathbf{s}_{+}^{\top} \Sigma_{x, +} \mathbf{s}_{+} + \mathbf{s}_{-}^{\top} \Sigma_{x, -} \mathbf{s}_{-}
\end{aligned}\end{equation*}
\normalsize
The bias-correction design matrices and the curvature weights, bias terms, and their estimators are
\small\begin{equation*}\begin{aligned}
        \hat{\Gamma}_{+, 1} &\coloneq \frac{1}{n} \sum_{i = 1}^{n} \omega_{i, +} R_{i, 1}R_{i, 1}^{\top}  \in \mathbb{R}^{2 \times 2}\\
    \hat{\Psi}_{+, 1} &\coloneq \frac{1}{n} \sum_{i = 1}^{n} \omega_{i, +} \bmat{R_{i, 1}R_{i, 1}^{\top} & R_{i, 1} W_{i}^{\top} \\  Z_{i} R_{i, 1}^{\top} & Z_{i} W_{i}^{\top}} \in \mathbb{R}^{(2 + q) \times (2 + q)}\\
    \hat{\Gamma}_{+, 2} &\coloneq \tfrac{1}{n}\textstyle\sum_{i}\delta_{+, i}R_{i, 2}R_{i, 2}^{\top} \in \R^{3 \times 3}\\
    \hat{\Psi}_{+, 2} &\coloneq \tfrac{1}{n}\textstyle\sum_{i}\delta_{+, i}\bmat{R_{i, 2}R_{i, 2}^{\top} & R_{i, 2}W_i^{\top}\\Z_iR_{i, 2}^{\top} & Z_iW_i^{\top}} \in \R^{(3 + q) \times (3 + q)}\\
    \Gamma_{+, 1} &\coloneq f_{d, +}(d^{*})\bmat{\gamma_{0, 1} & \gamma_{1, 1}\\\gamma_{1, 1} & \gamma_{2, 1}} \in \R^{2 \times 2}\\
    \Gamma_{+, 2} &\coloneq f_{d, +}(d^{*})\bmat{\gamma_{0, 1} & \gamma_{1, 1} & \gamma_{2, 1}\\\gamma_{1, 1} & \gamma_{2, 1} & \gamma_{3, 1}\\\gamma_{2, 1} & \gamma_{3, 1} & \gamma_{4, 1}} \in \R^{3 \times 3}\\
\end{aligned}\hspace{1.5em}\begin{aligned}
    \hat{\Omega}_{+} &\coloneq \tfrac{1}{n}\textstyle\sum_{i}\omega_{+, i}\bmat{R_{i, 1}\\Z_i}\left(\tfrac{D_i - d^{*}}{h_n}\right)^{2} \in \R^{2 + q}\\
    \hat{\Lambda}_{+} &\coloneq \tfrac{1}{n}\textstyle\sum_{i}\omega_{+, i}R_{i, 1}\left(\tfrac{D_i - d^{*}}{h_n}\right)^{2} \in \R^{2}\\
    \Omega_{+} &\coloneq f_{d, +}(d^{*})\left(\gamma_{2, 1}, \gamma_{3, 1}, \gamma_{2, 1}\mu_{+, z}(0)^{\top}\right)^{\top} \in \R^{2 + q}\\
    \Lambda_{+} &\coloneq f_{d, +}(d^{*})\left(\gamma_{2, 1}, \gamma_{3, 1}\right)^{\top} \in \R^{2}\\
    \eta_{+, \mathrm{IV}} &\coloneq \tfrac{C}{2}g_{+}^{(2)}(0)\Omega_{+}\\
    \hat{\eta}_{+, \mathrm{IV}} &\coloneq \tfrac{h_n^{2}}{2}\hat{\Omega}_{+}\hat{g}_{+}^{(2)}(0)\\
    \eta_{+, w_j} &\coloneq \tfrac{C}{2}\Lambda_{+}\mu_{+, w_j}^{(2)}(0)\\
    \hat{\eta}_{+, w_j} &\coloneq \tfrac{h_n^{2}}{2}\hat{\Lambda}_{+}\hat{\mu}_{+, w_j}^{(2)}(0)
\end{aligned}\end{equation*}
\normalsize
$\Psi_{+, 1}$ and $\Psi_{+, 2}$ denote the probability limits of $\hat{\Psi}_{+, 1}$ and $\hat{\Psi}_{+, 2}$, with explicit forms in Lemmas~\ref{lem:psi1_converge} and~\ref{lem:psi2_conv}; $\tau_{\mathrm{pdd}}$ is the PDD estimand; $\hat{\tau}_{\mathrm{pdd}}$ is its estimator in Equation~\eqref{eq:naive_estimator}; and $\hat{\tau}_{\mathrm{pdd}}^{\mathrm{rbc}}$ is the bias-corrected estimator in Equation~\eqref{eq:pdd_rbc_def_app}.

The conditional second-moment functions of the residuals that enter Lemma~\ref{lem:pdd_variance} are defined as follows. Writing $\varepsilon_{i, y} = Y_i - g_{+}(D_i - d^{*}) - W_i^{\top}\gamma_{+}$ and $\varepsilon_{i, w_j} = W_{i, j} - \mu_{+, w_j}(D_i - d^{*})$ for the above-cutoff residuals (with $``-"$ analogues replacing $g_{+}, \gamma_{+}, \mu_{+, w_j}$ by $g_{-}, \gamma_{-}, \mu_{-, w_j}$), and taking one-sided limits at the cutoff,
\small\begin{equation*}\begin{aligned}
    \sigma^{2}_{+, y}(0) &\coloneq \lim_{\epsilon \downarrow 0^{+}} \E{\varepsilon_{i, y}^{2} \mid D_i = d^{*} + \epsilon} \in \R\\
    \sigma^{2}_{+, w_j}(0) &\coloneq \lim_{\epsilon \downarrow 0^{+}} \E{\varepsilon_{i, w_j}^{2} \mid D_i = d^{*} + \epsilon} \in \R\\
    \rho_{+, yw_j}(0) &\coloneq \lim_{\epsilon \downarrow 0^{+}} \E{\varepsilon_{i, y}\varepsilon_{i, w_j} \mid D_i = d^{*} + \epsilon} \in \R\\
    \rho_{+, z_kw_j}(0) &\coloneq \lim_{\epsilon \downarrow 0^{+}} \E{Z_{i, k}\varepsilon_{i, w_j} \mid D = d^{*} + \epsilon} \in \R
\end{aligned}\hspace{2em}\begin{aligned}
    \rho_{+, y^{2}z}(0) &\coloneq \lim_{\epsilon \downarrow 0^{+}} \E{\varepsilon_{i, y}^{2} Z_i \mid D_i = d^{*} + \epsilon} \in \R^{q}\\
    \rho_{+, y^{2}zz^{\top}}(0) &\coloneq \lim_{\epsilon \downarrow 0^{+}} \E{\varepsilon_{i, y}^{2} Z_i Z_i^{\top} \mid D_i = d^{*} + \epsilon} \in \R^{q \times q}\\
    \rho_{+, yw_jz}(0) &\coloneq \lim_{\epsilon \downarrow 0^{+}} \E{\varepsilon_{i, y}\varepsilon_{i, w_j} Z_i \mid D_i = d^{*} + \epsilon} \in \R^{q}
\end{aligned}\end{equation*}
\normalsize
Since $\E{\varepsilon_{i, y} \mid D_i, Z_i} = 0$ and $\E{\varepsilon_{i, w_j} \mid D_i} = 0$, these coincide with the corresponding one-sided conditional variances and covariances at the cutoff.

\subsection{Characterizing asymptotic bias}
Let $\hat{\mathbf{e}}_{+, w} = \left(1, 0, \left(\hat{\beta}_{+, 0}^{w}\right)^{\top}\right)^{\top}$ and $\nu_{+} = (\alpha_{+, 0}, \alpha_{+, 1}, \gamma_{+}^{\top})^{\top} \in \mathbb{R}^{2 + q}$, where $\alpha_{+} = \mathbf{H}_{1} (g_{+}(0), g_{+}^{(1)}(0))^{\top}$. Using this notation, we can express our estimator from Equation~\eqref{eq:naive_estimator} as,
\begin{align*}
    \hat{\tau}_{\text{pdd}} - \tau_{\text{pdd}} &= \left\{\hat{\alpha}_{+, 0} + \left(\hat{\beta}_{+, 0}^{w}\right)^{\top} \hat{\gamma}_{+} - \alpha_{+, 0} - \left(\beta_{+, 0}^{w}\right)^{\top}\gamma_{+}\right\} - \left\{\hat{\alpha}_{-, 0} + \left(\hat{\beta}_{+, 0}^{w}\right)^{\top} \hat{\gamma}_{-} - \alpha_{-, 0} - \left(\beta_{+, 0}^{w}\right)^{\top}\gamma_{-}\right\}\\ 
    &= \hat{\mathbf{e}}_{+, w}^{\top} \left\{\hat{\nu}_{+} - \nu_{+}\right\} + \nu_{+}^{\top}\left(\hat{\mathbf{e}}_{+, w} - \mathbf{e}_{+, w}\right)  - \hat{\mathbf{e}}_{+, w}^{\top} \left\{\hat{\nu}_{-} - \nu_{-} \right\} - \nu_{-}^{\top}\left(\hat{\mathbf{e}}_{+, w} - \mathbf{e}_{+, w}\right)\\
    &= \hat{\mathbf{e}}_{+, w}^{\top} \left\{\hat{\nu}_{+} - \nu_{+}\right\} + \gamma_{+}^{\top}\left(\hat{\beta}_{+, 0}^{w} - \beta_{+, 0}^{w} \right)  - \hat{\mathbf{e}}_{+, w}^{\top} \left\{\hat{\nu}_{-} - \nu_{-} \right\} - \gamma_{-}^{\top}\left(\hat{\beta}_{+, 0}^{w} - \beta_{+, 0}^{w} \right)\\
    \hat{\nu}_{+} &\coloneq \hat{\Psi}_{+, 1}^{-1} \left( \frac{1}{n} \sum_{i = 1}^{n} \omega_{i, +} \bmat{R_{i, 1} \\ Z_{i}} Y_{i}\right).
\end{align*}

With these definitions, to characterize the asymptotic bias we decompose the differences $\left(\hat{\nu}_{+} - \nu_{+}\right)$ and $(\hat{\beta}_{+, 0}^{w} - \beta_{+, 0}^{w})^{\top}$ as follows,
\begin{equation}\begin{aligned}\label{eq:nu_decompose}
    \left(\hat{\nu}_{+} - \nu_{+}\right) &=  \hat{\Psi}_{+, 1}^{-1} \left( \frac{1}{n} \sum_{i = 1}^{n} \omega_{i, +} \bmat{R_{i, 1} \\ Z_{i}} (Y_{i} - R_{i, 1}^{\top} \alpha_{+} - W_{i}^{\top} \gamma_{+})\right)\\
    &= \hat{\Psi}_{+, 1}^{-1} \left( \frac{1}{n} \sum_{i = 1}^{n} \omega_{i, +} \bmat{R_{i, 1} \\ Z_{i}} (Y_{i} - g_{+}(D_{i} - d^{*}) - W_{i}^{\top} \gamma_{+})\right) \\
    &\qquad + \hat{\Psi}_{+, 1}^{-1} \left( \frac{1}{n} \sum_{i = 1}^{n} \omega_{i, +} \bmat{R_{i, 1} \\ Z_{i}} (g_{+}(D_{i} - d^{*}) - R_{i, 1}^{\top} \alpha_{+})\right)\\
    \left(\hat{\beta}_{+, 0}^{w} - \beta_{+, 0}^{w}\right)^{\top} &= \mathbf{e}_{0}^{\top} \hat{\Gamma}_{+, 1}^{-1} \left(\frac{1}{n} \sum_{i = 1}^{n} \omega_{+, i} R_{i, 1} \left\{W_{i}^{\top} - \mu_{+, w}(D_{i} - d^{*})^{\top} \right\}\right) \\
    &\qquad + \mathbf{e}_{0}^{\top} \hat{\Gamma}_{+, 1}^{-1} \left(\frac{1}{n} \sum_{i = 1}^{n} \omega_{+, i} R_{i, 1} \left\{\mu_{+, w}(D_{i} - d^{*})^{\top} - R_{i, 1}^{\top} \beta_{+}^{w}\right\}\right)
\end{aligned}\end{equation}

The two lemmas below characterize the limit of the leading terms above at the MSE-optimal rate $h_{n} \asymp c_{h} n^{-1/5}$.

\begin{lemma}[$\hat{\Psi}_{+, 1}$ convergence]\label{lem:psi1_converge}
    Under the conditions of Theorem~\ref{thm:limit_dist}, $\hat{\Psi}_{+, 1} \inprob \Psi_{+, 1}$, where
    \begin{align*}
        \Psi_{+, 1} &\coloneq f_{d, +}(d^{*}) \pmat{\bmat{\gamma_{0, 1} & \gamma_{1, 1}\\ \gamma_{1, 1} & \gamma_{2, 1}} & \bmat{\gamma_{0, 1} \\ \gamma_{1, 1}} \mu_{+, w}(0)^{\top}\\
        \mu_{+, z}(0) \bmat{\gamma_{0, 1} & \gamma_{1, 1}} & \gamma_{0, 1} \mu_{+, zw^{\top}}(0)}.
    \end{align*}\hfill$\triangleleft$
\end{lemma}

\begin{lemma}[$\hat{\Gamma}_{+, 1}$ convergence]\label{lem:Gamma1_converge}
    Under the conditions of Theorem~\ref{thm:limit_dist}, $\hat{\Gamma}_{+, 1} \inprob \Gamma_{+, 1} \coloneq f_{d, +}(d^{*}) \left(\begin{smallmatrix}\gamma_{0, 1} & \gamma_{1, 1}\\ \gamma_{1, 1} & \gamma_{2, 1}\end{smallmatrix}\right)$. \hfill$\triangleleft$
\end{lemma}

\begin{lemma}[IV bias convergence]\label{lem:IV_bias_converge}
    Under the conditions of Theorem~\ref{thm:limit_dist}, for $C = c_{h}^{5/2}$, we have that,
    \begin{align*}
        \sqrt{nh_{n}}\left(\hat{\nu}_{+} - \nu_{+}\right) &= \hat{\Psi}_{+, 1}^{-1} \left( \sqrt{\frac{h_{n}}{n}} \sum_{i = 1}^{n} \omega_{i, +} \bmat{R_{i, 1} \\ Z_{i}} (Y_{i} - g_{+}(D_{i} - d^{*}) - W_{i}^{\top} \gamma_{+})\right) + \Psi_{+, 1}^{-1} \eta_{+, \mathrm{IV}} + \op{1}\\
        \eta_{+, \mathrm{IV}} &\coloneq \frac{C}{2} g_{+}^{(2)}(0)  \Omega_{+}, \qquad
        \Omega_{+} \coloneq f_{d, +}(d^{*}) \left(\gamma_{2, 1}, \; \gamma_{3, 1} , \; \gamma_{2, 1} \mu_{+, z}(0)^{\top}\right)^{\top}.
    \end{align*}\hfill$\triangleleft$
\end{lemma}

\begin{lemma}[Local linear bias convergence]\label{lem:rdd_bias_converge}
    Under the conditions of Theorem~\ref{thm:limit_dist}, for $C = c_{h}^{5/2}$, we have that,
    \begin{align*}
        \sqrt{nh_n}\left(\hat{\beta}_{+, 0}^{w_j} - \beta_{+, 0}^{w_j}\right) &= \mathbf{e}_{0}^{\top} \hat{\Gamma}_{+, 1}^{-1} \left(\sqrt{\frac{h_n}{n}} \sum_{i = 1}^{n} \omega_{+, i} R_{i, 1} \left\{W_{i, j} - \mu_{+, w_j}(D_{i} - d^{*}) \right\}\right) + \mathbf{e}_{0}^{\top} \Gamma_{+, 1}^{-1} \eta_{+, w_j} + \op{1}\\
        \eta_{+, w_j} &\coloneq \frac{C}{2} \Lambda_{+} \mu_{+, w_j}^{(2)}(0), \qquad
        \Lambda_{+} \coloneq f_{d, +}(d^{*}) \left(\gamma_{2, 1}, \gamma_{3, 1}\right)^{\top}.
    \end{align*}\hfill$\triangleleft$
\end{lemma}

\subsection{Bias correction}

To correct the asymptotic bias above, we estimate the bias terms $\eta_{+, \cdot}$ from Proposition~\ref{prop:bias} (PDD bias) and subtract them from $\hat{\tau}_{\text{pdd}}$.

Recall $\eta_{+, \mathrm{IV}} = \frac{C}{2} g_{+}^{(2)}(0) \Omega_{+}$. To estimate $\frac{1}{2} g_{+}^{(2)}(0)$, define $\pi_{+} = \mathbf{B}_{2} \left(g_{+}(0), g_{+}^{(1)}(0), \frac{1}{2} g_{+}^{(2)}(0) \right)^{\top}$ and estimate $\hat{\pi}_{+}$ by solving,
\begin{align*}
    0 &= \frac{1}{n} \sum_{i = 1}^{n} \delta_{+, i} \bmat{R_{i, 2} \\ Z_{i}} \left(Y_{i} - R_{i, 2}^{\top} \hat{\pi}_{+} - W_{i}^{\top} \tilde{\gamma}_{+} \right)
    \;\Rightarrow\; \mathbf{e}_{2}^{\top} \hat{\pi}_{+} = \mathbf{e}_{2}^{\top} \hat{\Psi}_{+, 2}^{-1}\left(\frac{1}{n} \sum_{i = 1}^{n} \delta_{+, i} \bmat{R_{i, 2} \\ Z_{i}} Y_{i} \right),
\end{align*}
where $\hat{\Psi}_{+, 2} \coloneq \frac{1}{n} \sum_{i = 1}^{n} \delta_{+, i} \left(\begin{smallmatrix} R_{i, 2}R_{i, 2}^{\top} & R_{i, 2}W_{i}^{\top} \\ Z_i R_{i, 2}^{\top} & Z_i W_i^{\top}\end{smallmatrix}\right)$. Letting $\mathbf{e}_{2} = (0, 0, 1)^{\top} \in \R^{3}$, we define,
\begin{align*}
    \hat{\eta}_{+, \mathrm{IV}} &\coloneq \frac{h_{n}^{2}}{2} \hat{\Omega}_{+} \hat{g}_{+}^{(2)}(0), \qquad
    \hat{g}_{+}^{(2)}(0) \coloneq \frac{2}{b_{n}^{2}} \mathbf{e}_{2}^{\top} \hat{\pi}_{+}, \qquad
    \hat{\Omega}_{+} \coloneq \frac{1}{n} \sum_{i = 1}^{n} \omega_{+, i} \pmat{R_{i, 1} \\ Z_{i}} \left( \frac{D_{i} - d^{*}}{h_{n}}\right)^{2}.
\end{align*}

For $\eta_{+, w_j} = \frac{C}{2} \Lambda_{+} \mu_{+, w_j}^{(2)}(0) $, we apply the same correction as \cite{calonico_robust_2014}. For $\kappa_{+}^{w_{j}} = \mathbf{B}_{2} \left(\mu_{+, w_{j}}(0), \mu_{+, w_{j}}^{(1)}(0), \frac{1}{2} \mu_{+, w_{j}}^{(2)}(0) \right)^{\top}$, estimate $\hat{\kappa}_{+}^{w_{j}}$ by solving,
\begin{align*}
    0 &= \frac{1}{n} \sum_{i = 1}^{n} \delta_{+, i} R_{i, 2}\left(W_{i, j} - R_{i, 2}^{\top} \hat{\kappa}_{+}^{w_{j}} \right)
    \;\Rightarrow\; \mathbf{e}_{2}^{\top} \hat{\kappa}_{+}^{w_{j}} = \mathbf{e}_{2}^{\top} \hat{\Gamma}_{+, 2}^{-1} \left(\frac{1}{n} \sum_{i = 1}^{n} \delta_{+, i} R_{i, 2} W_{i, j} \right),
\end{align*}
where $\hat{\Gamma}_{+, 2} \coloneq \frac{1}{n} \sum_{i = 1}^{n} \delta_{+, i} R_{i, 2} R_{i, 2}^{\top}$. We then define,
\begin{align*}
    \hat{\eta}_{+, w_{j}} &\coloneq \frac{h_{n}^{2}}{2} \hat{\Lambda}_{+} \hat{\mu}_{+, w_{j}}^{(2)}(0), \qquad
    \hat{\mu}_{+, w_{j}}^{(2)}(0) \coloneq \frac{2}{b_{n}^{2}} \mathbf{e}_{2}^{\top} \hat{\kappa}_{+}^{w_{j}}, \qquad
    \hat{\Lambda}_{+} \coloneq \frac{1}{n} \sum_{i = 1}^{n} \omega_{+, i} R_{i, 1} \left(\frac{D_i - d^{*}}{h_{n}} \right)^{2}.
\end{align*}

The lemmas below characterize the limits of the estimators above.

\begin{lemma}[$\Lambda$-convergence]\label{lem:lambda_conv}
    Under the conditions of Theorem~\ref{thm:limit_dist}, we have that $\hat{\Lambda}_{+} \inprob \Lambda_{+}$, where,
    \begin{align*}
        \hat{\Lambda}_{+} &\coloneq \frac{1}{n} \sum_{i = 1}^{n} \omega_{+, i} R_{i, 1} \left(\frac{D_i - d^{*}}{h_{n}} \right)^{2}, \qquad
        \Lambda_{+} \coloneq f_{d, +}(d^{*}) \left(\gamma_{2, 1}, \; \gamma_{3, 1}\right)^{\top}.
    \end{align*}\hfill$\triangleleft$
\end{lemma}

\begin{lemma}[$\Omega$-convergence]\label{lem:omega_conv}
    Under the conditions of Theorem~\ref{thm:limit_dist}, we have that $\hat{\Omega}_{+} \inprob \Omega_{+}$, where,
    \begin{align*}
        \hat{\Omega}_{+} &\coloneq \frac{1}{n} \sum_{i = 1}^{n} \omega_{+, i} \pmat{R_{i, 1} \\ Z_{i}} \left( \frac{D_{i} - d^{*}}{h_{n}}\right)^{2}, \qquad \Omega_{+} \coloneq f_{d, +}(d^{*}) \left(\gamma_{2, 1}, \; \gamma_{3, 1} , \; \gamma_{2, 1} \mu_{+, z}(0)^{\top}\right)^{\top}.
    \end{align*}\hfill$\triangleleft$
\end{lemma}

\begin{lemma}[$\Gamma_{+, 2}$-convergence]\label{lem:Gamma2_conv}
    Under the conditions of Theorem~\ref{thm:limit_dist}, we have that $\hat{\Gamma}_{+, 2} \inprob \Gamma_{+, 2}$, where,
    \begin{align*}
        \hat{\Gamma}_{+, 2} &\coloneq \frac{1}{n} \sum_{i = 1}^{n} \delta_{+, i} R_{i, 2} R_{i, 2}^{\top}, \qquad \Gamma_{+, 2} \coloneq f_{d, +}(d^{*}) \bmat{\gamma_{0, 1} & \gamma_{1, 1} & \gamma_{2, 1} \\ \gamma_{1, 1} & \gamma_{2, 1} & \gamma_{3, 1} \\ \gamma_{2, 1} & \gamma_{3, 1} & \gamma_{4, 1}}.
    \end{align*}\hfill$\triangleleft$
\end{lemma}

\begin{lemma}[$\Psi_{+, 2}$-convergence]\label{lem:psi2_conv}
    Under the conditions of Theorem~\ref{thm:limit_dist}, we have that $\hat{\Psi}_{+, 2} \inprob \Psi_{+, 2}$, where
    \begin{align*}
        \hat{\Psi}_{+, 2} &\coloneq \frac{1}{n} \sum_{i = 1}^{n} \delta_{+, i} \pmat{ R_{i, 2}R_{i, 2}^{\top} & R_{i, 2}W_{i}^{\top} \\ Z_i R_{i, 2}^{\top} & Z_i W_i^{\top}} , \qquad \Psi_{+, 2} \coloneq f_{d, +}(d^{*}) \pmat{ \bmat{\gamma_{0, 1} & \gamma_{1, 1} & \gamma_{2, 1} \\ \gamma_{1, 1} & \gamma_{2, 1} & \gamma_{3, 1} \\ \gamma_{2, 1} & \gamma_{3, 1} & \gamma_{4, 1}} & \bmat{\gamma_{0, 1} \\ \gamma_{1, 1} \\ \gamma_{2, 1}} \mu_{+, w}(0)^{\top}\\
        \mu_{+, z}(0) \bmat{\gamma_{0, 1} & \gamma_{1, 1} & \gamma_{2, 1}} & \gamma_{0, 1} \mu_{+, zw^{\top}}(0)}.
    \end{align*}\hfill$\triangleleft$
\end{lemma}

We define the bias-corrected PDD estimator $\hat{\tau}_{\mathrm{pdd}}^{\mathrm{rbc}}$ as,
\begin{equation}\label{eq:pdd_rbc_def_app}\begin{aligned}
    \hat{\tau}_{\mathrm{pdd}}^{\mathrm{rbc}} &= \hat{\tau}_{\mathrm{pdd}} - \left(\hat{\mathbf{e}}_{+, w}^{\top} \left(\hat{\Psi}_{+, 1}^{-1}\hat{\eta}_{+, \mathrm{IV}} - \hat{\Psi}_{-, 1}^{-1}\hat{\eta}_{-, \mathrm{IV}}\right) +  \sum_{j = 1}^{q} \left(\hat{\gamma}_{+, j} - \hat{\gamma}_{-, j}\right)\mathbf{e}_{0}^{\top}\hat{\Gamma}_{+, 1}^{-1} \hat{\eta}_{+, w_j} \right).
\end{aligned}\end{equation}

\subsection{Central limit theorem}

In this section, we state the final results necessary to reach the conclusion of Theorem~\ref{thm:limit_dist}.

\begin{lemma}[$\hat{\eta}_{\mathrm{IV}}$-decomposition]\label{lem:eta_IV_decomp}
    Under the conditions of Theorem~\ref{thm:limit_dist},
    \begin{align*}
        \hat{\Psi}_{+, 1}^{-1}\sqrt{nh_{n}} \hat{\eta}_{+, \mathrm{IV}} &=  \sqrt{nh_{n}}\,h_{n}^{2}  \hat{\Psi}_{+, 1}^{-1} \hat{\Omega}_{+} \mathbf{e}_{2}^{\top} \hat{\Psi}_{+, 2}^{-1}\left(\frac{1}{n b_{n}^{2}} \sum_{i = 1}^{n} \delta_{+, i} \bmat{R_{i, 2} \\ Z_{i}} \left(Y_{i} - g_{+}(D_i - d^{*}) - W_i^{\top} \gamma_{+}\right) \right) \\
        &\qquad + \Psi_{+, 1}^{-1} \eta_{+, \mathrm{IV}} + \op{1}.
    \end{align*}\hfill$\triangleleft$
\end{lemma}

\begin{lemma}[$\hat{\eta}_{w}$-decomposition]\label{lem:eta_w_decomp}
    Under the conditions of Theorem~\ref{thm:limit_dist},
    \begin{align*}
        \hat{\Gamma}_{+, 1}^{-1} \sqrt{nh_{n}} \hat{\eta}_{+, w_{j}} &=  \sqrt{nh_{n}}\,h_{n}^{2} \hat{\Gamma}_{+, 1}^{-1} \hat{\Lambda}_{+} \mathbf{e}_{2}^{\top} \hat{\Gamma}_{+, 2}^{-1}\left(\frac{1}{n b_{n}^{2}} \sum_{i = 1}^{n} \delta_{+, i} R_{i, 2} \left(W_{i, j} - \mu_{+, w_{j}}\left(D_{i} - d^{*} \right) \right) \right) \\
        &\qquad + \Gamma_{+, 1}^{-1} \eta_{+, w_{j}} + \op{1}.
    \end{align*}\hfill$\triangleleft$
\end{lemma}

We collect the scores of the terms above into the vectors below. The blocks paired with $\hat{\Psi}_{\cdot, 2}$ and $\hat{\Gamma}_{\cdot, 2}$ use $R_{i, 2}$ and the kernel $\K{(D_{i} - d^{*})/b_{n}}$; the remaining blocks use $R_{i, 1}$ and $\K{(D_{i} - d^{*})/h_{n}}$.
\small\begin{equation*}\begin{aligned}
    \hat{\mathbf{s}}_{+} &\coloneq \pmat{\left\{\hat{\mathbf{e}}_{+, w}^{\top} \hat{\Psi}_{+, 1}^{-1}\right\}^{\top} \\
                                    -\left\{\xi^{3}  \hat{\mathbf{e}}_{+, w}^{\top} \hat{\Psi}_{+, 1}^{-1} \hat{\Omega}_{+} \mathbf{e}_{2}^{\top} \hat{\Psi}_{+, 2}^{-1}\right\}^{\top} \\
                                    \left\{(\gamma_{+, 1} - \gamma_{-, 1}) \mathbf{e}_{0}^{\top} \hat{\Gamma}_{+, 1}^{-1}\right\}^{\top} \\
                                    \vdots \\
                                    \left\{(\gamma_{+, q} - \gamma_{-, q}) \mathbf{e}_{0}^{\top} \hat{\Gamma}_{+, 1}^{-1}\right\}^{\top} \\
                                    -\left\{\xi^{3} (\hat{\gamma}_{+, 1} - \hat{\gamma}_{-, 1}) \mathbf{e}_{0}^{\top} \hat{\Gamma}_{+, 1}^{-1} \hat{\Lambda}_{+} \mathbf{e}_{2}^{\top} \hat{\Gamma}_{+, 2}^{-1} \right\}^{\top}\\
                                    \vdots \\
                                    -\left\{\xi^{3} (\hat{\gamma}_{+, q} - \hat{\gamma}_{-, q}) \mathbf{e}_{0}^{\top} \hat{\Gamma}_{+, 1}^{-1} \hat{\Lambda}_{+} \mathbf{e}_{2}^{\top} \hat{\Gamma}_{+, 2}^{-1}\right\}^{\top}
                                }\\
    \hat{\mathbf{s}}_{-} &\coloneq \pmat{-\left\{\hat{\mathbf{e}}_{+, w}^{\top} \hat{\Psi}_{-, 1}^{-1}\right\}^{\top} \\
                                    \left\{\xi^{3}  \hat{\mathbf{e}}_{+, w}^{\top} \hat{\Psi}_{-, 1}^{-1} \hat{\Omega}_{-} \mathbf{e}_{2}^{\top} \hat{\Psi}_{-, 2}^{-1}\right\}^{\top}}\\
    \hat{\mathbf{s}} &= \left(\hat{\mathbf{s}}_{+}^{\top}, \hat{\mathbf{s}}_{-}^{\top}\right)^{\top}\\
    X_{i} &\coloneq \left(X_{i, +}^{\top}, X_{i, -}^{\top}\right)^{\top}\\
    \mathbf{X}_{n} &\coloneq \frac{1}{\sqrt{nh_n}} \sum_{i = 1}^{n} X_{i}
    \end{aligned}\hspace{1em}\begin{aligned}
    X_{i, +} &\coloneq \pmat{  \1{D_{i} \geq d^{*}}\K{\frac{D_{i} - d^{*}}{h_{n}}} \bmat{R_{i, 1} \\ Z_{i}} \left(Y_{i} - g_{+}(D_{i} - d^{*}) - W_{i}^{\top}\gamma_{+}\right)\\
                        \1{D_{i} \geq d^{*}}\K{\frac{D_{i} - d^{*}}{b_{n}}} \bmat{R_{i, 2} \\ Z_{i}}\left(Y_{i} - g_{+}(D_{i} - d^{*}) - W_{i}^{\top}\gamma_{+}\right) \\
                        \1{D_{i} \geq d^{*}}\K{\frac{D_{i} - d^{*}}{h_{n}}} R_{i, 1} \left\{W_{i, 1} - \mu_{+, w_1}(D_i - d^{*}) \right\} \\
                        \vdots \\
                        \1{D_{i} \geq d^{*}}\K{\frac{D_{i} - d^{*}}{h_{n}}} R_{i, 1} \left\{W_{i, q} - \mu_{+, w_q}(D_i - d^{*})\right\}\\ 
                        \1{D_{i} \geq d^{*}}\K{\frac{D_{i} - d^{*}}{b_{n}}} R_{i, 2} \left\{W_{i, 1} - \mu_{+, w_1}(D_i - d^{*}) \right\} \\
                        \vdots \\
                        \1{D_{i} \geq d^{*}}\K{\frac{D_{i} - d^{*}}{b_{n}}} R_{i, 2} \left\{W_{i, q} - \mu_{+, w_q}(D_i - d^{*}) \right\}}\\
    X_{i, -} &\coloneq \pmat{ \1{D_{i} < d^{*}}\K{\frac{D_{i} - d^{*}}{h_{n}}} \bmat{R_{i, 1} \\ Z_{i}} \left(Y_{i} - g_{-}(D_{i} - d^{*}) - W_{i}^{\top}\gamma_{-}\right)\\
                        \1{D_{i} < d^{*}}\K{\frac{D_{i} - d^{*}}{b_{n}}} \bmat{R_{i, 2} \\ Z_{i}}\left(Y_{i} - g_{-}(D_{i} - d^{*}) - W_{i}^{\top}\gamma_{-}\right)}
\end{aligned}\end{equation*}\normalsize

\begin{corollary}[$\hat{\tau}_{\mathrm{pdd}}^{\mathrm{rbc}}$ vector characterization]\label{cor:clt_decomp} Under the conditions of Theorem~\ref{thm:limit_dist},
    $\sqrt{nh_n}\left(\hat{\tau}_{\mathrm{pdd}}^{\mathrm{rbc}} - \tau\right) = \hat{\mathbf{s}}^{\top} \mathbf{X}_{n} + \op{1}$. \hfill$\triangleleft$
\end{corollary}

Using the results above, we are ready to prove Theorem~\ref{thm:limit_dist}.

\begin{proof}[Proof of Theorem~\ref{thm:limit_dist}]
    Without loss of generality, assume $\dim(W) = q = 1$; for $q > 1$ the argument is identical with additional cross terms. By Corollary~\ref{cor:clt_decomp}, $\sqrt{nh_n}(\hat{\tau}_{\mathrm{pdd}}^{\mathrm{rbc}} - \tau) = \hat{\mathbf{s}}^{\top} \mathbf{X}_{n} + \op{1}$, where (suppressing the $``-"$ terms) $\hat{\mathbf{s}}_{+} \in \R^{12}$ and $X_{i, +} \in \R^{12}$ are as below, 
    \begin{equation*}\begin{aligned}
    \hat{\mathbf{s}}_{+} &\coloneq \pmat{\left\{\hat{\mathbf{e}}_{+, w}^{\top} \hat{\Psi}_{+, 1}^{-1}\right\}^{\top} \\
                                    -\left\{\left(\sqrt{nh_{n}}h_{n}^{2} \xi\right)  \hat{\mathbf{e}}_{+, w}^{\top} \hat{\Psi}_{+, 1}^{-1} \hat{\Omega}_{+} \mathbf{e}_{2}^{\top} \hat{\Psi}_{+, 2}^{-1}\right\}^{\top} \\
                                    \left\{(\gamma_{+, 1} - \gamma_{-, 1}) \mathbf{e}_{0}^{\top} \hat{\Gamma}_{+, 1}^{-1}\right\}^{\top} \\
                                    -\left\{\left(\sqrt{nh_{n}}h_{n}^{2} \xi\right) (\hat{\gamma}_{+, 1} - \hat{\gamma}_{-, 1}) \mathbf{e}_{0}^{\top} \hat{\Gamma}_{+, 1}^{-1} \hat{\Lambda}_{+} \mathbf{e}_{2}^{\top} \hat{\Gamma}_{+, 2}^{-1} \right\}^{\top}}
    \\
    X_{i, +} &\coloneq \pmat{\K{\frac{D_i - d^{*}}{h_n}} \1{D_i \geq d^{*}} \bmat{R_{i, 1} \\ Z_{i}} \left(Y_{i} - g_{+}(D_{i} - d^{*}) - W_{i}^{\top}\gamma_{+}\right)\\
                        \K{\frac{D_i - d^{*}}{h_n}} \1{D_i \geq d^{*}} \bmat{R_{i, 2} \\ Z_{i}}\left(Y_{i} - g_{+}(D_{i} - d^{*}) - W_{i}^{\top}\gamma_{+}\right) \\
                        \K{\frac{D_i - d^{*}}{h_n}} \1{D_i \geq d^{*}} R_{i, 1} \left\{W_{i, 1} - \mu_{+, w_1}(D_i - d^{*}) \right\} \\
                        \K{\frac{D_i - d^{*}}{h_n}} \1{D_i \geq d^{*}} R_{i, 2} \left\{W_{i, 1} - \mu_{+, w_1}(D_i - d^{*}) \right\}}\\
    \mathbf{X}_{n, +} &\coloneq  \sum_{i = 1}^{n} \frac{1}{\sqrt{nh_n}} X_{i, +}
    \end{aligned}\end{equation*}
    
    We proceed in three steps.
    \begin{itemize}
        \item We show $\mathbf{X}_{n, +} \indist \Norm{0, \Sigma_{x, +}}$ via the Cram\'er--Wold device, verifying the Lyapunov condition.
        \item We show $\hat{\mathbf{s}}_{+} \inprob \mathbf{s}_{+}$, so Slutsky gives $\hat{\mathbf{s}}_{+}^{\top} \mathbf{X}_{n, +} \indist \Norm{0, \mathbf{s}_{+}^{\top} \Sigma_{x, +} \mathbf{s}_{+}}$.
        \item By symmetry $\hat{\mathbf{s}}_{-}^{\top} \mathbf{X}_{n, -} \indist \Norm{0, \mathbf{s}_{-}^{\top} \Sigma_{x, -} \mathbf{s}_{-}}$, and by independence across the cutoff we conclude.
    \end{itemize}

    \begin{enumerate}[label = \roman*]
        \item Let $\mathbf{t} = (t_1, \dots, t_{12})^{\top}$ be a nonzero finite vector and $\zeta > 0$ the constant in Assumption~\ref{assum:reg}d. Group $\mathbf{t}^{\top} X_{i, +}$ into its four blocks, the second and fourth of which are localized at $b_{n}$:
        \begin{align*}
            \mathbf{t}^{\top} X_{i, +} &= \1{D_i \geq d^{*}}\K{\tfrac{D_i - d^{*}}{h_n}} A_{i}^{(1)} + \1{D_i \geq d^{*}}\K{\tfrac{D_i - d^{*}}{b_n}} A_{i}^{(2)} \\
            &\quad + \1{D_i \geq d^{*}}\K{\tfrac{D_i - d^{*}}{h_n}} B_{i}^{(1)} + \1{D_i \geq d^{*}}\K{\tfrac{D_i - d^{*}}{b_n}} B_{i}^{(2)},
        \end{align*}
        where, writing $u = \tfrac{D_i - d^{*}}{h_n}$ and $v = \tfrac{D_i - d^{*}}{b_n} = \xi u$,
        \begin{align*}
            A_{i}^{(1)} &= \left(t_1 + t_2 u + t_3 Z_i\right)\varepsilon_{i, y}, &
            A_{i}^{(2)} &= \left(t_4 + t_5 v + t_6 v^{2} + t_7 Z_i\right)\varepsilon_{i, y}, \\
            B_{i}^{(1)} &= \left(t_8 + t_9 u\right)\varepsilon_{i, w_1}, &
            B_{i}^{(2)} &= \left(t_{10} + t_{11} v + t_{12} v^{2}\right)\varepsilon_{i, w_1},
        \end{align*}
        with $\varepsilon_{i, y} = Y_i - g_{+}(D_i - d^{*}) - W_i^{\top}\gamma_{+}$ and $\varepsilon_{i, w_1} = W_{i, 1} - \mu_{+, w_1}(D_i - d^{*})$. By the $c_{r}$-inequality, which states that $\vert{a + b + c + d}^{r} \leq 4^{r - 1}\left(\vert{a}^{r} + \vert{b}^{r} + \vert{c}^{r} + \vert{d}^{r}\right)$ for $r \geq 1$,
        \begin{align*}
            \vert{\mathbf{t}^{\top}X_{i,+}}^{2+\zeta} \leq 4^{1+\zeta}\left(\vert{\K{u}A_{i}^{(1)}}^{2+\zeta} + \vert{\K{v}A_{i}^{(2)}}^{2+\zeta} + \vert{\K{u}B_{i}^{(1)}}^{2+\zeta} + \vert{\K{v}B_{i}^{(2)}}^{2+\zeta}\right)\1{D_i\geq d^*}.
        \end{align*}
        We bound the $\frac{1}{h_n}\E{\cdot}$ of each term, starting with $A_{i}^{(1)} = (t_1 + t_2 u + t_3 Z_i)\varepsilon_{i, y}$. By the law of iterated expectations, conditioning on $D_i$, and the substitution $u = \tfrac{D_i - d^{*}}{h_n}$,
        \begin{align*}
            &\frac{1}{h_n}\E{\K{u}^{2+\zeta}\1{D_i \geq d^{*}}\vert{A_{i}^{(1)}}^{2+\zeta}} \\
            &\quad= \frac{1}{h_n}\E{\K{\tfrac{D_i - d^{*}}{h_n}}^{2+\zeta}\1{D_i \geq d^{*}}\,\E{\vert{(t_1 + t_2 u + t_3 Z_i)\varepsilon_{i, y}}^{2+\zeta} \mid D_i}}\\
            &\quad= \int_{0}^{\infty}\K{u}^{2+\zeta}\,\E{\vert{(t_1 + t_2 u + t_3 Z_i)\varepsilon_{i, y}}^{2+\zeta} \mid D_i = uh_n + d^{*}}\,f_{d}(uh_n + d^{*})\dv u.
        \end{align*}
        By the $c_{r}$-inequality, $\vert{(t_1 + t_2 u + t_3 Z_i) \varepsilon_{i, y}}^{2+\zeta} \leq 3^{1+\zeta}(\vert{t_1}^{2+\zeta}\vert{\varepsilon_{i, y}}^{2+\zeta} + \vert{t_2}^{2+\zeta}\vert{u}^{2+\zeta}\vert{\varepsilon_{i, y}}^{2+\zeta} + \vert{t_3}^{2+\zeta}\vert{Z_i \varepsilon_{i, y}}^{2+\zeta})$, so, using Assumption~\ref{assum:reg}d to bound the conditional $(2+\zeta)$-moments of $Z_i$ and $\varepsilon_{i, y}$ by $M$, the conditional expectation is bounded by $M'(1 + \vert{u}^{2+\zeta})$ for a constant $M'$ depending only on $\mathbf{t}$ and $M$. Since the kernel constants $\gamma_{k, 2+\zeta} = \int_{0}^{\infty}u^{k}\K{u}^{2+\zeta}\dv u$ are finite under Assumption~\ref{assum:kernel}, the integrand is dominated by the integrable function $M' f_{d}(uh_n + d^{*})\K{u}^{2+\zeta}(1 + \vert{u}^{2+\zeta})$. By continuity of the conditional moments and of $f_{d}$ (Assumption~\ref{assum:reg}), dominated convergence then gives $\frac{1}{h_n}\E{\K{u}^{2+\zeta}\1{D_i \geq d^{*}}\vert{A_{i}^{(1)}}^{2+\zeta}} \to q_{A}^{h}(\mathbf{t})$, where
        \begin{align*}
            q_{A}^{h}(\mathbf{t}) \coloneq f_{d, +}(d^{*})\int_{0}^{\infty}\K{u}^{2+\zeta}\,\E{\vert{(t_1 + t_2 u + t_3 Z_i)\varepsilon_{i, y}}^{2+\zeta} \mid D_i = d^{*}}\dv u < \infty.
        \end{align*}
        The term in $B_{i}^{(1)} = (t_8 + t_9 u)\varepsilon_{i, w_1}$ is handled identically, giving $\frac{1}{h_n}\E{\K{u}^{2+\zeta}\1{D_i \geq d^{*}}\vert{B_{i}^{(1)}}^{2+\zeta}} \to q_{B}^{h}(\mathbf{t}) < \infty$. Both limits are finite linear combinations of the kernel constants $\gamma_{k, 2+\zeta}$ and the conditional moments $\E{\vert{Z_i}^{2+\zeta}\mid D_i = d^{*}}$, $\E{\vert{\varepsilon_{i, y}}^{2+\zeta}\mid D_i = d^{*}}$, $\E{\vert{\varepsilon_{i, w_1}}^{2+\zeta}\mid D_i = d^{*}}$, weighted by $f_{d, +}(d^{*})$ (all bounded by $M$ under Assumption~\ref{assum:reg}d).
 
        For the two $b_{n}$-localized terms, the same steps apply, but the $u$-substitution introduces the bandwidth ratio. With $v = \xi u$,
        \begin{align*}
            \frac{1}{h_n}\E{\K{v}^{2+\zeta}\1{D_i \geq d^{*}}\vert{A_{i}^{(2)}}^{2+\zeta}} &= \int_{0}^{\infty} \K{\xi u}^{2+\zeta}\, \E{\vert{A^{(2)}}^{2+\zeta}\mid D_i = uh_n + d^{*}}\, f_{d}(uh_n + d^{*}) \dv u\\
            &= \frac{1}{\xi}\int_{0}^{\infty} \K{v}^{2+\zeta}\, \E{\vert{A^{(2)}}^{2+\zeta}\mid D_i = \tfrac{v}{\xi}h_n + d^{*}}\, f_{d}(\tfrac{v}{\xi}h_n + d^{*}) \dv v\\
            &\to q_{A}^{b}(\mathbf{t}) < \infty,
        \end{align*}
        and similarly $\frac{1}{h_n}\E{\K{v}^{2+\zeta}\1{D_i \geq d^{*}}\vert{B_{i}^{(2)}}^{2+\zeta}} \to q_{B}^{b}(\mathbf{t}) < \infty$, where $q_{A}^{b}, q_{B}^{b}$ are finite. The swap of limit and integral uses the boundedness from Assumptions~\ref{assum:reg} and~\ref{assum:kernel}, and existence of the limits uses continuity in Assumption~\ref{assum:reg}.
 
        Combining, $\frac{1}{h_n}\E{\vert{\mathbf{t}^{\top}X_{i, +}}^{2+\zeta}} \to 4^{1+\zeta}\left(q_{A}^{h} + q_{A}^{b} + q_{B}^{h} + q_{B}^{b}\right)(\mathbf{t}) < \infty$. The Lyapunov ratio is then
        \begin{align*}
            \sum_{i = 1}^{n}\E{\vert{\tfrac{1}{\sqrt{nh_n}}\mathbf{t}^{\top}X_{i, +}}^{2+\zeta}} = \left(\tfrac{1}{nh_n}\right)^{\zeta/2}\frac{1}{h_n}\E{\vert{\mathbf{t}^{\top}X_{i, +}}^{2+\zeta}} \to 0,
        \end{align*}
        since $\left(\tfrac{1}{nh_n}\right)^{\zeta/2}\to 0$ and $\frac{1}{h_n}\E{\vert{\mathbf{t}^{\top}X_{i, +}}^{2+\zeta}}$ converges to a finite limit. This verifies the Lyapunov condition. Noting $\E{X_{i, +}} = 0$ (which follows from Lemma~\ref{lemma:factuals} for the IV blocks and from the definition of $\mu_{+, w_1}(\cdot)$ for the $w$ blocks), we conclude $\mathbf{X}_{n, +} \indist \Norm{0, \Sigma_{x, +}}$, with $\Sigma_{x, +}$ computed in Lemma~\ref{lem:pdd_variance}.

        \item We show $\hat{\mathbf{s}}_{+} \inprob \mathbf{s}_{+}$ termwise.
        \begin{itemize}
            \item $\hat{\mathbf{e}}_{+, w}^{\top} \hat{\Psi}_{+, 1}^{-1} \inprob \mathbf{e}_{+, w}^{\top} \Psi_{+, 1}^{-1}$: standard RDD results \citep[Section 4.2]{hahn2001identification} give $\hat{\mathbf{e}}_{+, w} \inprob \mathbf{e}_{+, w}$; Lemma~\ref{lem:psi1_converge} gives $\hat{\Psi}_{+, 1} \inprob \Psi_{+, 1}$, and writing $\hat{\Psi}_{+, 1}^{-1} - \Psi_{+, 1}^{-1} = \hat{\Psi}_{+, 1}^{-1}(\Psi_{+, 1} - \hat{\Psi}_{+, 1})\Psi_{+, 1}^{-1}$ gives $\hat{\Psi}_{+, 1}^{-1} \inprob \Psi_{+, 1}^{-1}$. Slutsky concludes.
            \item $\xi^{3}\hat{\mathbf{e}}_{+, w}^{\top} \hat{\Psi}_{+, 1}^{-1} \hat{\Omega}_{+} \mathbf{e}_{2}^{\top} \hat{\Psi}_{+, 2}^{-1} \inprob \xi^{3}\mathbf{e}_{+, w}^{\top} \Psi_{+, 1}^{-1} \Omega_{+} \mathbf{e}_{2}^{\top} \Psi_{+, 2}^{-1}$: $\xi = h_n/b_n$ is constant, $\hat{\Psi}_{+, 2} \inprob \Psi_{+, 2}$ (Lemma~\ref{lem:psi2_conv}, hence $\hat{\Psi}_{+, 2}^{-1}\inprob\Psi_{+, 2}^{-1}$), and $\hat{\Omega}_{+} \inprob \Omega_{+}$ (Lemma~\ref{lem:omega_conv}). Slutsky concludes.
            \item $(\gamma_{+, 1} - \gamma_{-, 1}) \mathbf{e}_{0}^{\top} \hat{\Gamma}_{+, 1}^{-1} \inprob (\gamma_{+, 1} - \gamma_{-, 1}) \mathbf{e}_{0}^{\top}\Gamma_{+, 1}^{-1}$ by Lemma~\ref{lem:Gamma1_converge} and the inverse expansion above.
            \item $\xi^{3} (\hat{\gamma}_{+, 1} - \hat{\gamma}_{-, 1}) \mathbf{e}_{0}^{\top} \hat{\Gamma}_{+, 1}^{-1} \hat{\Lambda}_{+} \mathbf{e}_{2}^{\top} \hat{\Gamma}_{+, 2}^{-1} \inprob \xi^{3} (\gamma_{+, 1} - \gamma_{-, 1}) \mathbf{e}_{0}^{\top} \Gamma_{+, 1}^{-1} \Lambda_{+} \mathbf{e}_{2}^{\top} \Gamma_{+, 2}^{-1}$: $\hat{\gamma}_{+} \inprob \gamma_{+}$ (Proposition~\ref{prop:estimand}, so $\hat{\gamma}_{+, 1} - \hat{\gamma}_{-, 1} \inprob \gamma_{+, 1} - \gamma_{-, 1}$), $\hat{\Lambda}_{+} \inprob \Lambda_{+}$ (Lemma~\ref{lem:lambda_conv}), and $\hat{\Gamma}_{+, 2} \inprob \Gamma_{+, 2}$ (Lemma~\ref{lem:Gamma2_conv}). Slutsky concludes.
        \end{itemize}
        Hence $\hat{\mathbf{s}}_{+} \inprob \mathbf{s}_{+}$, where
        \begin{align*}
            \mathbf{s}_{+} &= \left(\mathbf{e}_{+, w}^{\top} \Psi_{+, 1}^{-1}, \;- \xi^{3} \mathbf{e}_{+, w}^{\top} \Psi_{+, 1}^{-1} \Omega_{+} \mathbf{e}_{2}^{\top} \Psi_{+, 2}^{-1}, \; (\gamma_{+, 1} - \gamma_{-, 1}) \mathbf{e}_{0}^{\top}\Gamma_{+, 1}^{-1}, \; - \xi^{3} (\gamma_{+, 1} - \gamma_{-, 1}) \mathbf{e}_{0}^{\top} \Gamma_{+, 1}^{-1} \Lambda_{+} \mathbf{e}_{2}^{\top} \Gamma_{+, 2}^{-1} \right)^{\top}.
        \end{align*}
        By Slutsky, $\hat{\mathbf{s}}_{+}^{\top}\mathbf{X}_{n, +} \indist \Norm{0, \mathbf{s}_{+}^{\top}\Sigma_{x, +}\mathbf{s}_{+}}$.

        \item Repeating steps (i)--(ii) below the cutoff gives $\hat{\mathbf{s}}_{-} \inprob \mathbf{s}_{-}$ and $\mathbf{X}_{n, -} \indist \Norm{0, \Sigma_{x, -}}$, with $\mathbf{s}_{-} = \big(\mathbf{e}_{+, w}^{\top} \Psi_{-, 1}^{-1}, \; \xi^{3}\mathbf{e}_{+, w}^{\top} \Psi_{-, 1}^{-1} \Omega_{-} \mathbf{e}_{2}^{\top} \Psi_{-, 2}^{-1}\big)^{\top}$, hence $\hat{\mathbf{s}}_{-}^{\top}\mathbf{X}_{n, -} \indist \Norm{0, \mathbf{s}_{-}^{\top}\Sigma_{x, -}\mathbf{s}_{-}}$. By applying Cramer-Wold to the concatenated plus/minus triangular array with block-diagonal covariance, since the +/- terms are uncorrelated\footnote{Recall that we always have either $\omega_{+, i} = 0$ or $\omega_{-, i} = 0$.} we can conclude that,
        \begin{align*}
            \sqrt{nh_n}\left(\hat{\tau}_{\mathrm{pdd}}^{\mathrm{rbc}} - \tau\right) = \hat{\mathbf{s}}_{+}^{\top} \mathbf{X}_{n, +} + \hat{\mathbf{s}}_{-}^{\top} \mathbf{X}_{n, -} \indist \Norm{0, \mathbf{s}_{+}^{\top} \Sigma_{x, +} \mathbf{s}_{+} + \mathbf{s}_{-}^{\top} \Sigma_{x, -} \mathbf{s}_{-}}.
        \end{align*}
        We set $V = \mathbf{s}_{+}^{\top} \Sigma_{x, +} \mathbf{s}_{+} + \mathbf{s}_{-}^{\top} \Sigma_{x, -} \mathbf{s}_{-}$. We estimate $V$ with, 
        \begin{align*}
            \hat{V} = \hat{\mathbf{s}}_{+}^{\top} \left(\frac{1}{n h_n}\sum_{i} \tilde{X}_{i, +} \tilde{X}_{i, +}^{\top}\right) \hat{\mathbf{s}}_{+} + \hat{\mathbf{s}}_{-}^{\top} \left(\frac{1}{n h_n}\sum_{i} \tilde{X}_{i, -}\tilde{X}_{i, -}^{\top}\right) \hat{\mathbf{s}}_{-},
        \end{align*} 
        where $\tilde{X}_{i}$ replaces $\varepsilon_{i, y} = Y_{i} - g_{+}(D_i - d^{*}) - W_i^{\top}\gamma_{+}$ with the estimate $\hat{\varepsilon}_{i, y} = Y_{i} - R_{i, 2}^{\top}\hat{\pi}_{+} - W_{i}^{\top} \hat{\gamma}_{+}$, and $\varepsilon_{i, w_1} = W_{i, 1} - \mu_{+, w_1}(D_i - d^{*})$ with the estimate $\hat{\varepsilon}_{i, w_1} = W_{i, 1} - R_{i, 2}^{\top} \hat{\kappa}_{+}$.
    \end{enumerate}
\end{proof}

\begin{lemma}[PDD variance]\label{lem:pdd_variance} Under the conditions of Theorem~\ref{thm:limit_dist}, $ \frac{1}{h_n} \E{ X_{i, +} X_{i, +}^{\top}} \to \Sigma_{x, +}$ and $\frac{1}{h_n} \E{X_{i, -} X_{i, -}^{\top}} \to \Sigma_{x, -}$, where
    \begin{equation*}\begin{aligned}
        \Sigma_{x, +} &= \pmat{ \Sigma_{\mathrm{IV}, 11}^{+} & \Sigma_{\mathrm{IV}, 12}^{+} & \mathbf{P}_{\mathrm{IV},w_1,11}^{+} & \mathbf{P}_{\mathrm{IV},w_1,12}^{+} \\
        (\Sigma_{\mathrm{IV}, 12}^{+})^{\top} & \Sigma_{\mathrm{IV}, 22}^{+} & \mathbf{P}_{\mathrm{IV}, w_1, 21}^{+} & \mathbf{P}_{\mathrm{IV}, w_1, 22}^{+} \\
        (\mathbf{P}_{\mathrm{IV}, w_1, 11}^{+})^{\top} & (\mathbf{P}_{\mathrm{IV}, w_1, 21}^{+})^{\top} & \Sigma^{+}_{w_1, 11} & \Sigma^{+}_{w_1, 12} \\
        (\mathbf{P}_{\mathrm{IV}, w_1, 12}^{+})^{\top} & (\mathbf{P}_{\mathrm{IV}, w_1, 22}^{+})^{\top} & (\Sigma^{+}_{w_1, 12})^{\top} & \Sigma^{+}_{w_1, 22} }
        \end{aligned}\hspace{1em}\begin{aligned}
            \Sigma_{x, -} &= \pmat{\Sigma_{\mathrm{IV}, 11}^{-} & \Sigma_{\mathrm{IV}, 12}^{-}\\
                                  (\Sigma_{\mathrm{IV}, 12}^{-})^{\top} & \Sigma_{\mathrm{IV}, 22}^{-}}.
    \end{aligned}\end{equation*}
    Writing $\bar{\gamma}_{k, 2} = \bar{\gamma}_{k, 2}(\xi)$, the diagonal blocks are
    \begin{align*}
        \Sigma_{\mathrm{IV}, 11}^{+} &= f_{d, +}(d^{*}) \pmat{ \bmat{\gamma_{0, 2} & \gamma_{1, 2}\\ \gamma_{1, 2} & \gamma_{2, 2}} \sigma^{2}_{+, y}(0) & \bmat{\gamma_{0, 2} \\ \gamma_{1, 2}}\rho_{+, y^{2}z}(0)^{\top} \\
            \rho_{+, y^{2}z}(0) \bmat{\gamma_{0, 2} & \gamma_{1, 2}} & \gamma_{0, 2}\,\rho_{+, y^{2}zz^{\top}}(0)}\\
        \Sigma_{\mathrm{IV}, 22}^{+} &= \frac{1}{\xi}\,f_{d, +}(d^{*}) \pmat{ \bmat{\gamma_{0, 2} & \gamma_{1, 2} & \gamma_{2, 2} \\ \gamma_{1, 2} & \gamma_{2, 2} & \gamma_{3, 2} \\ \gamma_{2, 2} & \gamma_{3, 2} & \gamma_{4, 2}} \sigma^{2}_{+, y}(0) & \bmat{\gamma_{0, 2} \\ \gamma_{1, 2} \\ \gamma_{2, 2} }\rho_{+, y^{2}z}(0)^{\top} \\
            \rho_{+, y^{2}z}(0) \bmat{\gamma_{0, 2} & \gamma_{1, 2} & \gamma_{2, 2}} & \gamma_{0, 2}\,\rho_{+, y^{2}zz^{\top}}(0)}\\
        \Sigma_{w_1, 11}^{+} &= f_{d, +}(d^{*}) \sigma^{2}_{+, w_1}(0) \bmat{\gamma_{0, 2} & \gamma_{1, 2} \\ \gamma_{1, 2} & \gamma_{2, 2}}, \qquad
        \Sigma_{w_1, 22}^{+} = \frac{1}{\xi}\,f_{d, +}(d^{*}) \sigma^{2}_{+, w_1}(0) \bmat{\gamma_{0, 2} & \gamma_{1, 2} & \gamma_{2, 2} \\ \gamma_{1, 2} & \gamma_{2, 2} & \gamma_{3, 2} \\ \gamma_{2, 2} & \gamma_{3, 2} & \gamma_{4, 2}},
    \end{align*}
    and the cross-bandwidth blocks are
    \begin{align*}
        \Sigma_{\mathrm{IV}, 12}^{+} &= f_{d, +}(d^{*}) \pmat{ \bmat{\bar{\gamma}_{0, 2} & \xi\bar{\gamma}_{1, 2} & \xi^{2}\bar{\gamma}_{2, 2}\\ \bar{\gamma}_{1, 2} & \xi\bar{\gamma}_{2, 2} & \xi^{2}\bar{\gamma}_{3, 2}} \sigma^{2}_{+, y}(0) & \bmat{\bar{\gamma}_{0, 2} \\ \bar{\gamma}_{1, 2}}\rho_{+, y^{2}z}(0)^{\top} \\
            \rho_{+, y^{2}z}(0) \bmat{\bar{\gamma}_{0, 2} & \xi\bar{\gamma}_{1, 2} & \xi^{2}\bar{\gamma}_{2, 2}} & \bar{\gamma}_{0, 2}\,\rho_{+, y^{2}zz^{\top}}(0)}\\
        \Sigma_{w_1, 12}^{+} &= f_{d, +}(d^{*}) \sigma^{2}_{+, w_1}(0) \bmat{\bar{\gamma}_{0, 2} & \xi\bar{\gamma}_{1, 2} & \xi^{2}\bar{\gamma}_{2, 2}\\ \bar{\gamma}_{1, 2} & \xi\bar{\gamma}_{2, 2} & \xi^{2}\bar{\gamma}_{3, 2}}\\
        \mathbf{P}_{\mathrm{IV}, w_1, 11}^{+} &= f_{d, +}(d^{*}) \pmat{\rho_{+, yw_1}(0) \bmat{\gamma_{0, 2} & \gamma_{1, 2} \\ \gamma_{1, 2} & \gamma_{2, 2}}\\ \rho_{+, yw_1z}(0) \bmat{\gamma_{0, 2} & \gamma_{1, 2}}}, \qquad
        \mathbf{P}_{\mathrm{IV}, w_1, 12}^{+} = f_{d, +}(d^{*}) \pmat{\rho_{+, yw_1}(0) \bmat{\bar{\gamma}_{0, 2} & \xi\bar{\gamma}_{1, 2} & \xi^{2}\bar{\gamma}_{2, 2}\\ \bar{\gamma}_{1, 2} & \xi\bar{\gamma}_{2, 2} & \xi^{2}\bar{\gamma}_{3, 2}}\\ \rho_{+, yw_1z}(0) \bmat{\bar{\gamma}_{0, 2} & \xi\bar{\gamma}_{1, 2} & \xi^{2}\bar{\gamma}_{2, 2}}}\\
        \mathbf{P}_{\mathrm{IV}, w_1, 21}^{+} &= f_{d, +}(d^{*}) \pmat{\rho_{+, yw_1}(0) \bmat{\bar{\gamma}_{0, 2} & \bar{\gamma}_{1, 2} \\ \xi\bar{\gamma}_{1, 2} & \xi\bar{\gamma}_{2, 2} \\ \xi^{2}\bar{\gamma}_{2, 2} & \xi^{2}\bar{\gamma}_{3, 2}}\\ \rho_{+, yw_1z}(0) \bmat{\bar{\gamma}_{0, 2} & \bar{\gamma}_{1, 2}}}, \qquad
        \mathbf{P}_{\mathrm{IV}, w_1, 22}^{+} = \frac{1}{\xi}\,f_{d, +}(d^{*}) \pmat{\rho_{+, yw_1}(0) \bmat{\gamma_{0, 2} & \gamma_{1, 2} & \gamma_{2, 2} \\ \gamma_{1, 2} & \gamma_{2, 2} & \gamma_{3, 2} \\ \gamma_{2, 2} & \gamma_{3, 2} & \gamma_{4, 2}}\\ \rho_{+, yw_1z}(0) \bmat{\gamma_{0, 2} & \gamma_{1, 2} & \gamma_{2, 2}}}.
    \end{align*}
    The blocks of $\Sigma_{x, -}$ are $\Sigma_{\mathrm{IV}, 11}^{-}$, $\Sigma_{\mathrm{IV}, 22}^{-}$, $\Sigma_{\mathrm{IV}, 12}^{-}$, defined as their $``+"$ analogues with all $``+"$ subscripts replaced by $``-"$. \hfill$\triangleleft$
\end{lemma}

\subsection{Proofs of technical lemmas}

\begin{proof}[Proof of Lemma~\ref{lem:psi1_converge}]
    Plugging in for $\omega_{+, i}$, $\hat{\Psi}_{+, 1} = \frac{1}{n h_{n}} \sum_{i} \1{D_{i} \geq d^{*}} \K{\tfrac{D_{i} - d^{*}}{h_{n}}} \left(\begin{smallmatrix}R_{i, 1}R_{i, 1}^{\top} & R_{i, 1} W_{i}^{\top} \\  Z_{i} R_{i, 1}^{\top} & Z_{i} W_{i}^{\top}\end{smallmatrix}\right)$. We treat a representative entry; the rest are analogous.
    \begin{enumerate}[label = \roman*]
        \item The $R_{i, 1}R_{i, 1}^{\top}$ block converges to $f_{d, +}(d^{*})\Gamma_{1}$ by \citet[Lemma 1]{hahn1999evaluating}.
        \item For the $R_{i, 1} W_{i, j}$ block, set $X_{i} = \frac{1}{h_{n}}\1{D_{i} \geq d^{*}} \K{\tfrac{D_{i} - d^{*}}{h_{n}}} R_{i, 1} W_{i, j} \in \R^{2}$. For each $\ell \in \{1, 2\}$, we show $\frac{1}{n}\sum_{i} X_{i, \ell} = \E{X_{i, \ell}} + \op{1}$ by bounding $\Var{\frac{1}{n}\sum_{i}X_{i, \ell}} = \frac{1}{n}\Var{X_{1, \ell}} \leq \frac{1}{n}\norm{X_{1, \ell}}_{L_2}^{2}$ and applying Chebyshev's inequality. Computing the $L_2$ bound,
        \begin{align*}
            h_{n}\E{X_{i}^{2}} &= \frac{1}{h_{n}}\E{\1{D_{i} \geq d^{*}}\K{\tfrac{D_{i} - d^{*}}{h_{n}}}^{2} R_{i, 1}^{2}\, W_{i, j}^{2}}\\
            &= \int_{0}^{\infty} \K{u}^{2} \bmat{1 \\ u^{2}} \E{W_{i, j}^{2} \mid D_{i} = uh_{n} + d^{*}} f_{d}(uh_{n} + d^{*}) \dv u\\
            &\to \left(\sigma^{2}_{+, w_j}(0) + \mu_{+, w_j}(0)^{2}\right) f_{d, +}(d^{*}) \bmat{\gamma_{0, 2}\\ \gamma_{2, 2}} < \infty,
        \end{align*}
        where the squaring is element-wise and the limit/integral swap uses Assumptions~\ref{assum:reg} and~\ref{assum:kernel}. This shows that $\norm{X_{1, \ell}}_{L_2}^{2} = O(1/h_{n})$ and thus $\frac{1}{n}\norm{X_{1, \ell}}_{L_2}^{2} = O\left(\tfrac{1}{nh_{n}}\right) \to 0$. Then from Chebyshev's inequality,
        \begin{align*}
            \pr{\vert{\frac{1}{n} \sum_{i = 1}^{n} (X_{i, \ell} - \E{X_{i, \ell}})} > \epsilon} \leq \frac{1}{\epsilon^{2}} \frac{1}{n} \norm{X_{1, \ell}}_{L_{2}}^{2} \to 0,
        \end{align*}
        and the conclusion follows.
        \item By Step (ii) it suffices to compute the mean. Using the law of iterated expectations and $u$-substitution,
        \begin{align*}
            \frac{1}{h_{n}}\E{\1{D_{i} \geq d^{*}} \K{\tfrac{D_{i} - d^{*}}{h_{n}}} R_{i, 1} W_{i, j}} &= \int_{0}^{\infty} \K{u} \bmat{1 \\ u} \E{W_{i, j} \mid D_{i} = uh_{n} + d^{*}} f_{d}(uh_{n} + d^{*}) \dv u\\
            &\to \mu_{+, w_{j}}(0) f_{d, +}(d^{*}) \bmat{\gamma_{0, 1} \\ \gamma_{1, 1}}.
        \end{align*}
        \item Repeating (ii)--(iii) on the remaining entries gives $\frac{1}{n h_{n}}\sum_{i}\1{D_i\geq d^*}\K{\tfrac{D_i-d^*}{h_n}} Z_{i, j} R_{i, 1}^{\top} \inprob \mu_{+, z_{j}}(0) f_{d, +}(d^{*}) [\gamma_{0, 1}, \gamma_{1, 1}]$ and $\frac{1}{n h_{n}}\sum_{i}\1{D_i\geq d^*}\K{\tfrac{D_i-d^*}{h_n}}Z_{i, j} W_{i, j} \inprob \mu_{+, z_{j} w_{j}}(0) f_{d, +}(d^{*}) \gamma_{0, 1}$. Combining yields $\hat{\Psi}_{+, 1} \inprob \Psi_{+, 1}$.
    \end{enumerate}
\end{proof}

\begin{proof}[Proof of Lemma~\ref{lem:Gamma1_converge}]
    Under our conditions this is identical to \citet[Lemma 1]{hahn1999evaluating}.
\end{proof}

\begin{proof}[Proof of Lemma~\ref{lem:IV_bias_converge}]
    From Equation~\eqref{eq:nu_decompose}, the first term matches the lemma statement. For the second, since $\hat{\Psi}_{+, 1}^{-1} \inprob \Psi_{+, 1}^{-1}$ (Lemma~\ref{lem:psi1_converge}), it suffices to show
    \begin{align*}
        \sqrt{\frac{h_n}{n}}\sum_{i = 1}^{n}\omega_{i, +}\bmat{R_{i, 1}\\Z_i}\left(g_{+}(D_i - d^{*}) - R_{i, 1}^{\top}\alpha_{+}\right) \inprob \eta_{+, \mathrm{IV}}.
    \end{align*}
    Since $\alpha_{+} = \mathbf{H}_{1}(g_{+}(0), g_{+}^{(1)}(0))^{\top}$, we have $R_{i, 1}^{\top}\alpha_{+} = g_{+}(0) + g_{+}^{(1)}(0)(D_i - d^{*})$, the first-order Taylor polynomial of $g_{+}$ at $0$. By the smoothness of $g_{+}$ in Assumption~\ref{assum:reg}, a second-order Taylor expansion with Lagrange remainder gives, for some $\tilde{D}_i - d^{*} \in [0, D_i - d^{*}]$,
    \begin{align*}
        g_{+}(D_i - d^{*}) = g_{+}(0) + g_{+}^{(1)}(0)(D_i - d^{*}) + \tfrac{1}{2}g_{+}^{(2)}(\tilde{D}_i - d^{*})(D_i - d^{*})^{2},
    \end{align*}
    so that $g_{+}(D_i - d^{*}) - R_{i, 1}^{\top}\alpha_{+} = \tfrac{1}{2}g_{+}^{(2)}(\tilde{D}_i - d^{*})(D_i - d^{*})^{2}$. Substituting $\omega_{i, +} = \tfrac{1}{h_n}\1{D_i \geq d^{*}}\K{(D_i - d^{*})/h_n}$ and writing $(D_i - d^{*})^{2} = h_n^{2}\left(\tfrac{D_i - d^{*}}{h_n}\right)^{2}$, the left-hand side becomes
    \begin{align*}
        \sqrt{\frac{h_n}{n}}\sum_{i = 1}^{n}\omega_{i, +}\bmat{R_{i, 1}\\Z_i}\tfrac{1}{2}g_{+}^{(2)}(\tilde{D}_i - d^{*})(D_i - d^{*})^{2} = \sqrt{nh_n}\,h_n^{2}\cdot\frac{1}{nh_n}\sum_{i = 1}^{n}\1{D_i \geq d^{*}}\K{\tfrac{D_i - d^{*}}{h_n}}\bmat{R_{i, 1}\\Z_i}\tfrac{1}{2}g_{+}^{(2)}(\tilde{D}_i - d^{*})\left(\tfrac{D_i - d^{*}}{h_n}\right)^{2}.
    \end{align*}
    By a Chebyshev argument identical to Lemma~\ref{lem:psi1_converge}(ii), the average equals its expectation up to $\op{1}$. By the law of iterated expectations, conditioning on $D_i$ (so that $\E{Z_i \mid D_i} = \mu_{+, z}(D_i - d^{*})$ for $D_i \geq d^{*}$), and the substitution $u = \tfrac{D_i - d^{*}}{h_n}$, with $\tilde{u} = \tfrac{\tilde{D}_i - d^{*}}{h_n} \in [0, u]$,
    \begin{align*}
        &\frac{1}{h_{n}}\E{\1{D_{i} \geq d^{*}}\K{\tfrac{D_{i} - d^{*}}{h_{n}}} \bmat{R_{i, 1} \\ Z_{i}} \tfrac{1}{2}g_{+}^{(2)}(\tilde{D}_{i} - d^{*}) \left(\tfrac{D_{i} - d^{*}}{h_{n}}\right)^{2}}\\
        &\qquad = \frac{1}{h_n}\int_{d^{*}}^{\infty}\K{\tfrac{d - d^{*}}{h_n}}\bmat{1 \\ \tfrac{d - d^{*}}{h_n} \\ \mu_{+, z}(d - d^{*})}\tfrac{1}{2}g_{+}^{(2)}(\tilde{d} - d^{*})\left(\tfrac{d - d^{*}}{h_n}\right)^{2}f_{d}(d)\dv d\\
        &\qquad = \int_{0}^{\infty}\K{u}\bmat{1 \\ u \\ \mu_{+, z}(uh_n)}\tfrac{1}{2}g_{+}^{(2)}(\tilde{u}h_n)\,u^{2}\,f_{d}(uh_n + d^{*})\dv u\\
        &\qquad \to \tfrac{1}{2}g_{+}^{(2)}(0)\,f_{d, +}(d^{*})\int_{0}^{\infty}\K{u}\bmat{1 \\ u \\ \mu_{+, z}(0)}u^{2}\dv u \\
        &= \tfrac{1}{2}g_{+}^{(2)}(0)\,f_{d, +}(d^{*})\bmat{\gamma_{2, 1} \\ \gamma_{3, 1} \\ \gamma_{2, 1}\mu_{+, z}(0)} = \tfrac{1}{2}g_{+}^{(2)}(0)\Omega_{+},
    \end{align*}
    where the limit uses $\tilde{u}h_n \to 0$, so that $g_{+}^{(2)}(\tilde{u}h_n) \to g_{+}^{(2)}(0)$, together with $\mu_{+, z}(uh_n) \to \mu_{+, z}(0)$ and $f_{d}(uh_n + d^{*}) \to f_{d, +}(d^{*})$, and the limit/integral swap is justified by the boundedness in Assumptions~\ref{assum:reg} and~\ref{assum:kernel}. Hence the average converges in probability to $\tfrac{1}{2}g_{+}^{(2)}(0)\Omega_{+}$, and multiplying by $\sqrt{nh_{n}}h_{n}^{2} \to C = c_{h}^{5/2}$ and applying Slutsky's theorem gives the second term $\inprob \frac{C}{2}g_{+}^{(2)}(0)\Omega_{+} = \eta_{+, \mathrm{IV}}$.
\end{proof}

\begin{proof}[Proof of Lemma~\ref{lem:rdd_bias_converge}]
    The argument follows Lemma~\ref{lem:IV_bias_converge}, replacing the IV score by the local-linear $w$ score. A second-order Taylor expansion of $\mu_{+, w_{j}}$ gives the remainder $\frac{1}{2}\mu_{+, w_{j}}^{(2)}(\tilde{D}_i - d^{*})(D_i - d^{*})^{2}$, and the same Chebyshev and $u$-substitution steps give
    \begin{align*}
        \frac{1}{h_{n}}\E{\1{D_{i} \geq d^{*}}\K{\tfrac{D_{i} - d^{*}}{h_{n}}} R_{i, 1}\, \tfrac{1}{2}\mu_{+, w_{j}}^{(2)}(\tilde{D}_i - d^{*}) \left(\tfrac{D_{i} - d^{*}}{h_{n}}\right)^{2}} \to \tfrac{1}{2}\mu_{+, w_{j}}^{(2)}(0)\Lambda_{+}.
    \end{align*}
    Multiplying by $\sqrt{nh_n}h_n^{2}\to C$ and applying Slutsky's theorem gives that  $\left(\sqrt{\frac{h_n}{n}} \sum_{i = 1}^{n} \omega_{+, i} R_{i, 1} \left\{\mu_{+, w}(D_{i} - d^{*})^{\top} - R_{i, 1}^{\top} \beta_{+}^{w}\right\}\right) \inprob \eta_{+, w_j}$.
\end{proof}

\begin{proof}[Proof of Lemma~\ref{lem:lambda_conv}]
    We proceed in steps. Some steps of this proof are identical to Lemma~\ref{lem:Gamma1_converge} and~\ref{lem:IV_bias_converge}, so for brevity we only summarize those steps.
    \begin{enumerate}[label = \roman*]
        \item Expanding the sum, note that we can write,
        \begin{align*}
            \hat{\Lambda}_{+} &= \frac{1}{nh_{n}} \sum_{i = 1}^{n} \1{D_i \geq d^{*}} \K{\frac{D_i - d^{*}}{h_{n}}} R_{i, 1} \left( \frac{D_i - d^{*}}{h_{n}}\right)^{2}.
        \end{align*}
        \item Following the same argument as in Lemma~\ref{lem:Gamma1_converge}(ii) and applying Chebyshev's inequality, we can show that,
        \begin{align*}
            &\frac{1}{nh_{n}} \sum_{i = 1}^{n} \1{D_i \geq d^{*}} \K{\frac{D_i - d^{*}}{h_{n}}} R_{i, 1} \left( \frac{D_i - d^{*}}{h_{n}}\right)^2 \\
            &\quad = \frac{1}{h_n}\E{\1{D_i \geq d^{*}} \K{\frac{D_i - d^{*}}{h_{n}}} R_{i, 1}  \left( \frac{D_i - d^{*}}{h_{n}}\right)^2} + \op{1}.
        \end{align*}
        \item Evaluating the integral in the step above, we find that,
        \begin{align*}
            &\lim_{n \to \infty} \frac{1}{h_n} \E{\1{D_i \geq d^{*}} \K{\frac{D_i - d^{*}}{h_{n}}} R_{i, 1} \left( \frac{D_i - d^{*}}{h_{n}}\right)^2}\\
            &= \lim_{n \to \infty} \frac{1}{h_n} \int_{0}^{\infty} \K{\frac{d - d^*}{h_n}} \pmat{1 \\ \left(\frac{d - d^{*}}{h_n}\right)} \left(\frac{d - d^{*}}{h_{n}}\right)^2 f_{d}(d) \dv d\\
            &= \lim_{n \to \infty} \int_{0}^{\infty} \K{u} \pmat{1 \\ u} u^{2} f_{d}(uh_{n} + d^{*}) \dv u \qquad \text{(I)}\\
            &= f_{d, +}(d^{*}) \left(\gamma_{2, 1}, \gamma_{3, 1}\right)^{\top},
        \end{align*}
        where (I) uses $u$-substitution. The swapping of the limits uses the boundedness from Assumptions~\ref{assum:reg} and~\ref{assum:kernel}, and the existence of the limits uses the continuity conditions in Assumption~\ref{assum:reg}.
        \item Combining steps (ii) and (iii), it follows that $\hat{\Lambda}_{+} \inprob \Lambda_{+}$.
    \end{enumerate}
\end{proof}

\begin{proof}[Proof of Lemma~\ref{lem:omega_conv}]
    The argument follows Lemma~\ref{lem:lambda_conv}, with the extra coordinate $Z_i$. Expanding, $\hat{\Omega}_{+} = \hfill\break \frac{1}{nh_n}\sum_{i}\1{D_i \geq d^{*}}\K{\tfrac{D_i - d^{*}}{h_n}}\left(\begin{smallmatrix}R_{i, 1}\\Z_i\end{smallmatrix}\right)\left(\tfrac{D_i - d^{*}}{h_n}\right)^{2}$. By the Chebyshev and $u$-substitution argument,
    \begin{align*}
        \frac{1}{h_n}\E{\1{D_i \geq d^{*}}\K{\tfrac{D_i - d^{*}}{h_n}}\pmat{R_{i, 1}\\ Z_i}\left(\tfrac{D_i - d^{*}}{h_n}\right)^{2}} \to \int_{0}^{\infty}\K{u}\pmat{1\\u\\\mu_{+, z}(uh_n)}u^{2}f_{d}(uh_n + d^{*})\dv u = \Omega_{+},
    \end{align*}
    so $\hat{\Omega}_{+}\inprob\Omega_{+}$.
\end{proof}

\begin{proof}[Proof of Lemma~\ref{lem:Gamma2_conv}]
    Plugging in $\delta_{+, i}$, $\hat{\Gamma}_{+, 2} = \frac{1}{nb_n}\sum_{i}\1{D_i \geq d^{*}}\K{\tfrac{D_i - d^{*}}{b_n}}R_{i, 2}R_{i, 2}^{\top}$. By the same Chebyshev and $u$-substitution argument, now substituting $v = \tfrac{D_i - d^{*}}{b_n}$ at the bandwidth $b_n$,
    \begin{align*}
        \frac{1}{b_n}\E{\1{D_i \geq d^{*}}\K{\tfrac{D_i - d^{*}}{b_n}}R_{i, 2}R_{i, 2}^{\top}} \to f_{d, +}(d^{*})\int_{0}^{\infty}\K{v}\pmat{1 & v & v^2 \\ v & v^2 & v^3 \\ v^2 & v^3 & v^4}\dv v = \Gamma_{+, 2},
    \end{align*}
    so $\hat{\Gamma}_{+, 2}\inprob\Gamma_{+, 2}$.
\end{proof}

\begin{proof}[Proof of Lemma~\ref{lem:psi2_conv}]
    The argument follows Lemma~\ref{lem:psi1_converge}, now localized at $b_{n}$ via $\delta_{+, i}$ and with the additional quadratic coordinate of $R_{i, 2}$. Each entry converges by the Chebyshev and $u$-substitution argument with $v = \tfrac{D_i - d^{*}}{b_n}$, giving $\hat{\Psi}_{+, 2}\inprob\Psi_{+, 2}$.
\end{proof}

\begin{proof}[Proof of Lemma~\ref{lem:eta_IV_decomp}]
    We proceed in steps. Recall that we defined,
    \begin{align*}
        \sqrt{nh_{n}}\hat{\eta}_{+, \mathrm{IV}} &= \sqrt{nh_{n}}h_{n}^{2} \hat{\Omega}_{+} \frac{1}{b_{n}^{2}} \mathbf{e}_{2}^{\top}\hat{\pi}_{+}.
    \end{align*}
    \begin{enumerate}[label = \roman*]
        \item First, we rewrite the estimator $\frac{1}{b_{n}^{2}}\mathbf{e}_{2}^{\top}\hat{\pi}_{+}$. We have that,
        \begin{align*}
            &\frac{1}{b_{n}^{2}} \mathbf{e}_{2}^{\top}\hat{\pi}_{+} \\
            &\quad = \mathbf{e}_{2}^{\top} \hat{\Psi}_{+, 2}^{-1}\left(\frac{1}{nb_{n}^{2}} \sum_{i = 1}^{n} \delta_{+, i} \bmat{R_{i, 2} \\ Z_{i}} Y_{i} \right)\\
            &\quad= \mathbf{e}_{2}^{\top} \hat{\Psi}_{+, 2}^{-1}\left(\frac{1}{n b_{n}^{2}} \sum_{i = 1}^{n} \delta_{+, i} \bmat{R_{i, 2} \\ Z_{i}} \left(Y_{i} - g_{+}(D_i - d^{*}) - W_i^{\top} \gamma_{+}\right) \right) \\
            &\qquad + \mathbf{e}_{2}^{\top} \hat{\Psi}_{+, 2}^{-1}\left(\frac{1}{nb_{n}^{2}} \sum_{i = 1}^{n} \delta_{+, i} \bmat{R_{i, 2} \\ Z_{i}} \left(g_{+}(D_i - d^{*}) - g_{+}(0) - g_{+}^{(1)}(0)(D_i - d^{*}) - \frac{1}{2} g_{+}^{(2)}(0) (D_i - d^{*})^{2} \right) \right)\\
            &\qquad + \mathbf{e}_{2}^{\top} \hat{\Psi}_{+, 2}^{-1}\left(\frac{1}{nb_{n}^{2}} \sum_{i = 1}^{n} \delta_{+, i} \bmat{R_{i, 2} \\ Z_{i}} \left(g_{+}(0) + g_{+}^{(1)}(0)(D_i - d^{*}) + \frac{1}{2} g_{+}^{(2)}(0) (D_i - d^{*})^{2} + W_{i}^{\top} \gamma_{+} \right) \right).
        \end{align*}
        The first term matches the term in the lemma statement. We treat the remaining two terms individually.

        \item We show that the second term converges to zero in probability. Using the continuous differentiability of $g_{+}(\cdot)$, for some $\tilde{D}_i - d^* \in [0, D_{i} - d^{*}]$, a Taylor expansion gives,
        \begin{align*}
            g_{+}(D_i - d^{*}) &= g_{+}(0) + g_{+}^{(1)}(0) (D_{i} - d^{*}) + \frac{1}{2} g_{+}^{(2)}(0) (D_{i} - d^{*})^{2} + \frac{1}{6} g_{+}^{(3)}(\tilde{D}_i - d^{*}) (D_i - d^{*})^{3}.
        \end{align*}
        Substituting into the second term,
        \begin{align*}
            &\mathbf{e}_{2}^{\top} \hat{\Psi}_{+, 2}^{-1}\left(\frac{1}{nb_{n}^{2}} \sum_{i = 1}^{n} \delta_{+, i} \bmat{R_{i, 2} \\ Z_{i}} \left(g_{+}(D_i - d^{*}) - g_{+}(0) - g_{+}^{(1)}(0)(D_i - d^{*}) - \frac{1}{2} g_{+}^{(2)}(0) (D_i - d^{*})^{2} \right) \right)\\
            &\qquad = \mathbf{e}_{2}^{\top} \hat{\Psi}_{+, 2}^{-1}\left(\frac{1}{n} \sum_{i = 1}^{n} \1{D_i \geq d^{*}}\K{\frac{D_i - d^{*}}{b_n}} \bmat{R_{i, 2} \\ Z_{i}} \frac{1}{6} g_{+}^{(3)}(\tilde{D}_i - d^{*}) \left( \frac{D_i - d^{*}}{b_n} \right)^{3} \right).
        \end{align*}
        By Lemma~\ref{lem:psi2_conv}, $\hat{\Psi}_{+, 2} \inprob \Psi_{+, 2}$, and writing $\hat{\Psi}_{+, 2}^{-1} - \Psi_{+, 2}^{-1} = \hat{\Psi}_{+, 2}^{-1}(\Psi_{+, 2} - \hat{\Psi}_{+, 2})\Psi_{+, 2}^{-1}$ gives $\hat{\Psi}_{+, 2}^{-1} \inprob \Psi_{+, 2}^{-1}$. Hence it suffices to show that
        \begin{align*}
            \frac{1}{n} \sum_{i = 1}^{n} \1{D_i \geq d^{*}}\K{\frac{D_i - d^{*}}{b_n}} \bmat{R_{i, 2} \\ Z_{i}} \frac{1}{6} g_{+}^{(3)}(\tilde{D}_i - d^{*}) \left( \frac{D_i - d^{*}}{b_n} \right)^{3} \inprob 0.
        \end{align*}
        Using a Chebyshev argument as in Lemma~\ref{lem:IV_bias_converge}(ii), the average over the bandwidth $b_n$ equals its expectation up to $\op{1}$,
        \begin{align*}
            &\frac{1}{nb_n} \sum_{i = 1}^{n} \1{D_i \geq d^{*}}\K{\frac{D_i - d^{*}}{b_n}} \bmat{R_{i, 2} \\ Z_{i}} \frac{1}{6} g_{+}^{(3)}(\tilde{D}_i - d^{*}) \left( \frac{D_i - d^{*}}{b_n} \right)^{3}\\
            &\qquad = \frac{1}{b_{n}} \E{\1{D_i \geq d^{*}}\K{\frac{D_i - d^{*}}{b_n}} \bmat{R_{i, 2} \\ Z_{i}} \frac{1}{6} g_{+}^{(3)}(\tilde{D}_i - d^{*}) \left( \frac{D_i - d^{*}}{b_n} \right)^{3}} + \op{1}.
        \end{align*}
        Computing the expectation, for some $\tilde{u} \in [0, u]$ and substituting $u = \tfrac{d - d^{*}}{b_n}$,
        \begin{align*}
            &\lim_{n \to \infty} \frac{1}{b_{n}} \E{\1{D_i \geq d^{*}}\K{\frac{D_i - d^{*}}{b_n}} \bmat{R_{i, 2} \\ Z_{i}} \frac{1}{6} g_{+}^{(3)}(\tilde{D}_i - d^{*}) \left( \frac{D_i - d^{*}}{b_n} \right)^{3}}\\
            &\qquad = \lim_{n \to \infty} \frac{1}{6} \int_{0}^{\infty} \K{u} \bmat{1 \\ u \\ u^{2} \\ \mu_{+, z}(ub_n)} g_{+}^{(3)}(\tilde{u}b_n) u^{3} f_{d}(ub_n + d^{*}) \dv u\\
            &\qquad = f_{d, +}(d^{*}) \frac{g_{+}^{(3)}(0)}{6} \left(\gamma_{3, 1}, \gamma_{4, 1}, \gamma_{5, 1}, \mu_{+, z}(0)^{\top}\gamma_{3, 1} \right)^{\top} < \infty,
        \end{align*}
        where the limit/integral swap uses Assumptions~\ref{assum:reg} and~\ref{assum:kernel}. Hence the average converges in probability to a finite constant $\tilde{C}$, and since $b_{n} \to 0$, Slutsky's theorem gives
        \begin{align*}
            \frac{1}{n} \sum_{i = 1}^{n} \1{D_i \geq d^{*}}\K{\frac{D_i - d^{*}}{b_n}} \bmat{R_{i, 2} \\ Z_{i}} \frac{1}{6} g_{+}^{(3)}(\tilde{D}_i - d^{*}) \left( \frac{D_i - d^{*}}{b_n} \right)^{3} = b_n \cdot \tilde{C} + \op{1} \inprob 0.
        \end{align*}

        \item We show the third term equals $\frac{1}{2}g_{+}^{(2)}(0)$. Collecting the polynomial into $R_{i, 2}$ and $W_{i}$ and using $\mathbf{B}_2 = \diag(1, b_n, b_n^{2})$,
        \begin{align*}
            &\frac{1}{b_{n}^{2}} \mathbf{e}_{2}^{\top} \hat{\Psi}_{+, 2}^{-1}\left(\frac{1}{n} \sum_{i = 1}^{n} \delta_{+, i} \bmat{R_{i, 2} \\ Z_{i}} \left(g_{+}(0) + g_{+}^{(1)}(0)(D_i - d^{*}) + \frac{1}{2} g_{+}^{(2)}(0) (D_i - d^{*})^{2} + W_{i}^{\top} \gamma_{+} \right) \right)\\
            &\qquad = \frac{1}{b_{n}^{2}} \mathbf{e}_{2}^{\top} \hat{\Psi}_{+, 2}^{-1}\left(\frac{1}{n} \sum_{i = 1}^{n} \delta_{+, i} \bmat{R_{i, 2} \\ Z_{i}} \bmat{R_{i, 2}^{\top} & W_{i}^{\top}} \right) \pmat{g_{+}(0) \\ b_{n} g_{+}^{(1)}(0) \\ b_{n}^{2} \frac{1}{2}g_{+}^{(2)}(0) \\ \gamma_{+} }\\
            &\qquad = \frac{1}{b_{n}^{2}} \mathbf{e}_{2}^{\top} \hat{\Psi}_{+, 2}^{-1} \hat{\Psi}_{+, 2} \pmat{g_{+}(0) \\ b_{n} g_{+}^{(1)}(0) \\ b_{n}^{2} \frac{1}{2}g_{+}^{(2)}(0) \\ \gamma_{+} }
            = \frac{1}{b_{n}^{2}} \mathbf{e}_{2}^{\top} \pmat{g_{+}(0) \\ b_{n} g_{+}^{(1)}(0) \\ b_{n}^{2} \frac{1}{2}g_{+}^{(2)}(0) \\ \gamma_{+} } = \frac{1}{2}g_{+}^{(2)}(0).
        \end{align*}

        \item Combining steps (i)--(iii),
        \begin{align*}
            \sqrt{nh_{n}}\hat{\eta}_{+, \mathrm{IV}} &= \sqrt{nh_{n}}h_{n}^{2}\hat{\Omega}_{+}\mathbf{e}_{2}^{\top} \hat{\Psi}_{+, 2}^{-1}\left(\frac{1}{n b_{n}^{2}} \sum_{i = 1}^{n} \delta_{+, i} \bmat{R_{i, 2} \\ Z_{i}} \left(Y_{i} - g_{+}(D_i - d^{*}) - W_i^{\top} \gamma_{+}\right) \right) \\
            &\quad + \sqrt{nh_{n}}h_{n}^{2} \hat{\Omega}_{+} \cdot O\left(b_{n}\right) + \sqrt{nh_{n}}h_{n}^{2} \hat{\Omega}_{+} \frac{1}{2} g_{+}^{(2)}(0).
        \end{align*}
        By Lemma~\ref{lem:omega_conv}, $\hat{\Omega}_{+} \inprob \Omega_{+}$, and since $h_{n} \asymp n^{-1/5}$, $\sqrt{nh_{n}} h_{n}^{2} \to C = c_{h}^{5/2}$. Hence by Slutsky's theorem the middle term is $\op{1}$ and the last term converges to $\frac{C}{2} \Omega_{+} g_{+}^{(2)}(0) = \eta_{+, \mathrm{IV}}$. Pre-multiplying by $\hat{\Psi}_{+, 1}^{-1} \inprob \Psi_{+, 1}^{-1}$ gives the statement.
    \end{enumerate}
\end{proof}

\begin{proof}[Proof of Lemma~\ref{lem:eta_w_decomp}]
    We proceed in steps, following the structure of Lemma~\ref{lem:eta_IV_decomp}. Recall that we defined,
    \begin{align*}
        \sqrt{nh_{n}}\hat{\eta}_{+, w_{j}} &= \sqrt{nh_{n}}h_{n}^{2} \hat{\Lambda}_{+} \frac{1}{b_{n}^{2}} \mathbf{e}_{2}^{\top}\hat{\kappa}_{+}^{w_{j}}.
    \end{align*}
    \begin{enumerate}[label = \roman*]
        \item Rewriting the estimator $\frac{1}{b_{n}^{2}}\mathbf{e}_{2}^{\top}\hat{\kappa}_{+}^{w_{j}}$,
        \small
        \begin{align*}
            &\frac{1}{b_{n}^{2}} \mathbf{e}_{2}^{\top}\hat{\kappa}_{+}^{w_{j}} \\
            &\quad= \mathbf{e}_{2}^{\top} \hat{\Gamma}_{+, 2}^{-1}\left(\frac{1}{nb_{n}^{2}} \sum_{i = 1}^{n} \delta_{+, i} R_{i, 2} W_{i, j} \right)\\
            &\quad= \mathbf{e}_{2}^{\top} \hat{\Gamma}_{+, 2}^{-1}\left(\frac{1}{n b_{n}^{2}} \sum_{i = 1}^{n} \delta_{+, i} R_{i, 2} \left(W_{i, j} - \mu_{+, w_{j}}(D_i - d^{*})\right) \right) \\
            &\qquad + \mathbf{e}_{2}^{\top} \hat{\Gamma}_{+, 2}^{-1}\left(\frac{1}{nb_{n}^{2}} \sum_{i = 1}^{n} \delta_{+, i} R_{i, 2} \left(\mu_{+, w_{j}}(D_i - d^{*}) - \mu_{+, w_{j}}(0) - \mu_{+, w_{j}}^{(1)}(0)(D_i - d^{*}) - \frac{1}{2} \mu_{+, w_{j}}^{(2)}(0) (D_i - d^{*})^{2} \right) \right)\\
            &\qquad + \mathbf{e}_{2}^{\top} \hat{\Gamma}_{+, 2}^{-1}\left(\frac{1}{nb_{n}^{2}} \sum_{i = 1}^{n} \delta_{+, i} R_{i, 2} \left(\mu_{+, w_{j}}(0) + \mu_{+, w_{j}}^{(1)}(0)(D_i - d^{*}) + \frac{1}{2} \mu_{+, w_{j}}^{(2)}(0) (D_i - d^{*})^{2} \right) \right).
        \end{align*}\normalsize
        The first term matches the term in the lemma statement. We treat the remaining two terms individually.

        \item We show the second term converges to zero in probability. Using the continuous differentiability of $\mu_{+, w_{j}}(\cdot)$, for some $\tilde{D}_i - d^{*} \in [0, D_i - d^{*}]$,
        \begin{align*}
            \mu_{+, w_{j}}(D_i - d^{*}) &= \mu_{+, w_{j}}(0) + \mu_{+, w_{j}}^{(1)}(0)(D_i - d^{*}) + \frac{1}{2}\mu_{+, w_{j}}^{(2)}(0)(D_i - d^{*})^{2} + \frac{1}{6}\mu_{+, w_{j}}^{(3)}(\tilde{D}_i - d^{*})(D_i - d^{*})^{3}.
        \end{align*}
        Substituting, the second term equals
        \begin{align*}
            \mathbf{e}_{2}^{\top} \hat{\Gamma}_{+, 2}^{-1}\left(\frac{1}{n} \sum_{i = 1}^{n} \1{D_i \geq d^{*}}\K{\frac{D_i - d^{*}}{b_n}} R_{i, 2} \frac{1}{6} \mu_{+, w_{j}}^{(3)}(\tilde{D}_i - d^{*}) \left( \frac{D_i - d^{*}}{b_n} \right)^{3} \right).
        \end{align*}
        By Lemma~\ref{lem:Gamma2_conv}, $\hat{\Gamma}_{+, 2}^{-1} \inprob \Gamma_{+, 2}^{-1}$, so it suffices to show the average converges to zero. By a Chebyshev argument as in Lemma~\ref{lem:IV_bias_converge}(ii) and $u$-substitution at the bandwidth $b_n$, for some $\tilde{u} \in [0, u]$,
        \begin{align*}
            &\frac{1}{b_n}\E{\1{D_i \geq d^{*}}\K{\frac{D_i - d^{*}}{b_n}} R_{i, 2} \frac{1}{6} \mu_{+, w_{j}}^{(3)}(\tilde{D}_i - d^{*}) \left( \frac{D_i - d^{*}}{b_n} \right)^{3}} \\
            &\qquad \to f_{d, +}(d^{*})\frac{\mu_{+, w_{j}}^{(3)}(0)}{6}\left(\gamma_{3, 1}, \gamma_{4, 1}, \gamma_{5, 1}\right)^{\top} < \infty,
        \end{align*}
        so the average converges to a finite constant $\tilde{C}$, and since $b_n \to 0$, Slutsky's theorem gives
        \begin{align*}
            \frac{1}{n} \sum_{i = 1}^{n} \1{D_i \geq d^{*}}\K{\frac{D_i - d^{*}}{b_n}} R_{i, 2} \frac{1}{6} \mu_{+, w_{j}}^{(3)}(\tilde{D}_i - d^{*}) \left( \frac{D_i - d^{*}}{b_n} \right)^{3} = b_n \cdot \tilde{C} + \op{1} \inprob 0.
        \end{align*}

        \item We show the third term equals $\frac{1}{2}\mu_{+, w_{j}}^{(2)}(0)$. Collecting the polynomial into $R_{i, 2}$,
        \begin{align*}
            &\frac{1}{b_{n}^{2}} \mathbf{e}_{2}^{\top} \hat{\Gamma}_{+, 2}^{-1}\left(\frac{1}{n} \sum_{i = 1}^{n} \delta_{+, i} R_{i, 2} \left(\mu_{+, w_{j}}(0) + \mu_{+, w_{j}}^{(1)}(0)(D_i - d^{*}) + \frac{1}{2} \mu_{+, w_{j}}^{(2)}(0) (D_i - d^{*})^{2} \right) \right)\\
            &\qquad = \frac{1}{b_{n}^{2}} \mathbf{e}_{2}^{\top} \hat{\Gamma}_{+, 2}^{-1}\left(\frac{1}{n} \sum_{i = 1}^{n} \delta_{+, i} R_{i, 2} R_{i, 2}^{\top} \right) \pmat{\mu_{+, w_{j}}(0) \\ b_{n} \mu_{+, w_{j}}^{(1)}(0) \\ b_{n}^{2} \frac{1}{2}\mu_{+, w_{j}}^{(2)}(0)}
            = \frac{1}{b_{n}^{2}} \mathbf{e}_{2}^{\top} \pmat{\mu_{+, w_{j}}(0) \\ b_{n} \mu_{+, w_{j}}^{(1)}(0) \\ b_{n}^{2} \frac{1}{2}\mu_{+, w_{j}}^{(2)}(0)} = \frac{1}{2}\mu_{+, w_{j}}^{(2)}(0).
        \end{align*}

        \item Combining steps (i)--(iii),
        \begin{align*}
            \sqrt{nh_{n}}\hat{\eta}_{+, w_{j}} &= \sqrt{nh_{n}}h_{n}^{2}\hat{\Lambda}_{+}\mathbf{e}_{2}^{\top} \hat{\Gamma}_{+, 2}^{-1}\left(\frac{1}{n b_{n}^{2}} \sum_{i = 1}^{n} \delta_{+, i} R_{i, 2} \left(W_{i, j} - \mu_{+, w_{j}}(D_i - d^{*})\right) \right)\\
            &\quad + \sqrt{nh_{n}}h_{n}^{2} \hat{\Lambda}_{+} \cdot O(b_n) + \sqrt{nh_{n}}h_{n}^{2} \hat{\Lambda}_{+} \frac{1}{2}\mu_{+, w_{j}}^{(2)}(0).
        \end{align*}
        By Lemma~\ref{lem:lambda_conv}, $\hat{\Lambda}_{+} \inprob \Lambda_{+}$, and $\sqrt{nh_{n}}h_{n}^{2} \to C$. Hence by Slutsky's theorem the middle term is $\op{1}$ and the last term converges to $\frac{C}{2}\Lambda_{+}\mu_{+, w_{j}}^{(2)}(0) = \eta_{+, w_{j}}$. Pre-multiplying by $\hat{\Gamma}_{+, 1}^{-1} \inprob \Gamma_{+, 1}^{-1}$ gives the statement.
    \end{enumerate}
\end{proof}

\begin{proof}[Proof of Corollary~\ref{cor:clt_decomp}]
    Starting from the definition of $\hat{\tau}^{\mathrm{rbc}}_{\mathrm{pdd}}$ in Equation~\eqref{eq:pdd_rbc_def_app} and the bias decomposition of $\hat{\tau}_{\mathrm{pdd}} - \tau_{\mathrm{pdd}}$,
    \begin{align*}
        \sqrt{nh_n}\left(\hat{\tau}_{\mathrm{pdd}}^{\mathrm{rbc}} - \tau \right) &= \hat{\mathbf{e}}_{+, w}^{\top} \sqrt{nh_{n}} \left\{\hat{\nu}_{+} - \nu_{+}\right\} - \hat{\mathbf{e}}_{+, w}^{\top} \sqrt{nh_{n}}\left\{\hat{\nu}_{-} - \nu_{-} \right\} + \sum_{j = 1}^{q}\left(\gamma_{+, j} - \gamma_{-, j}\right) \sqrt{nh_{n}}\left(\hat{\beta}_{+, 0}^{w_j} - \beta_{+, 0}^{w_j} \right) \\
        &\quad - \hat{\mathbf{e}}_{+, w}^{\top}\hat{\Psi}_{+, 1}^{-1}\sqrt{nh_{n}}\hat{\eta}_{+, \mathrm{IV}} +  \hat{\mathbf{e}}_{+, w}^{\top}\hat{\Psi}_{-, 1}^{-1}\sqrt{nh_{n}}\hat{\eta}_{-, \mathrm{IV}} -  \sum_{j = 1}^{q} \left(\hat{\gamma}_{+, j} - \hat{\gamma}_{-, j}\right)\mathbf{e}_{0}^{\top}\hat{\Gamma}_{+, 1}^{-1} \sqrt{nh_{n}}\hat{\eta}_{+, w_j}.
    \end{align*}
    We substitute the decompositions from Lemmas~\ref{lem:IV_bias_converge},~\ref{lem:rdd_bias_converge},~\ref{lem:eta_IV_decomp}, and~\ref{lem:eta_w_decomp}. For the $``+"$ IV terms, Lemmas~\ref{lem:IV_bias_converge} and~\ref{lem:eta_IV_decomp} give
    \begin{align*}
        &\hat{\mathbf{e}}_{+, w}^{\top}\sqrt{nh_{n}}\left\{\hat{\nu}_{+} - \nu_{+}\right\} - \hat{\mathbf{e}}_{+, w}^{\top}\hat{\Psi}_{+, 1}^{-1}\sqrt{nh_{n}}\hat{\eta}_{+, \mathrm{IV}}\\
        &\quad = \hat{\mathbf{e}}_{+, w}^{\top}\hat{\Psi}_{+, 1}^{-1}\left(\sqrt{\frac{h_{n}}{n}}\sum_{i = 1}^{n}\omega_{i, +}\bmat{R_{i, 1}\\Z_i}\varepsilon_{i, y}^{+}\right) + \hat{\mathbf{e}}_{+, w}^{\top}\Psi_{+, 1}^{-1}\eta_{+, \mathrm{IV}}\\
        &\qquad - \hat{\mathbf{e}}_{+, w}^{\top}\sqrt{nh_{n}}h_{n}^{2}\hat{\Psi}_{+, 1}^{-1}\hat{\Omega}_{+}\mathbf{e}_{2}^{\top}\hat{\Psi}_{+, 2}^{-1}\left(\frac{1}{nb_{n}^{2}}\sum_{i = 1}^{n}\delta_{+, i}\bmat{R_{i, 2}\\Z_i}\varepsilon_{i, y}^{+}\right) - \hat{\mathbf{e}}_{+, w}^{\top}\Psi_{+, 1}^{-1}\eta_{+, \mathrm{IV}} + \op{1},
    \end{align*}
    where $\varepsilon_{i, y}^{+} = Y_{i} - g_{+}(D_i - d^{*}) - W_{i}^{\top}\gamma_{+}$. The two $\hat{\mathbf{e}}_{+, w}^{\top}\Psi_{+, 1}^{-1}\eta_{+, \mathrm{IV}}$ terms cancel. The $``-"$ IV terms are handled identically, and the bias terms there cancel in the same way.

    For the $w$ terms, Lemmas~\ref{lem:rdd_bias_converge} and~\ref{lem:eta_w_decomp} give, for each $j$,
    \begin{align*}
        &\left(\gamma_{+, j} - \gamma_{-, j}\right)\sqrt{nh_{n}}\left(\hat{\beta}_{+, 0}^{w_j} - \beta_{+, 0}^{w_j}\right) - \left(\hat{\gamma}_{+, j} - \hat{\gamma}_{-, j}\right)\mathbf{e}_{0}^{\top}\hat{\Gamma}_{+, 1}^{-1}\sqrt{nh_{n}}\hat{\eta}_{+, w_j}\\
        &\quad = \left(\gamma_{+, j} - \gamma_{-, j}\right)\mathbf{e}_{0}^{\top}\hat{\Gamma}_{+, 1}^{-1}\left(\sqrt{\frac{h_n}{n}}\sum_{i = 1}^{n}\omega_{+, i}R_{i, 1}\varepsilon_{i, w_j}^{+}\right) + \left(\gamma_{+, j} - \gamma_{-, j}\right)\mathbf{e}_{0}^{\top}\Gamma_{+, 1}^{-1}\eta_{+, w_j}\\
        &\qquad - \left(\hat{\gamma}_{+, j} - \hat{\gamma}_{-, j}\right)\mathbf{e}_{0}^{\top}\hat{\Gamma}_{+, 1}^{-1}\sqrt{nh_{n}}h_{n}^{2}\hat{\Lambda}_{+}\mathbf{e}_{2}^{\top}\hat{\Gamma}_{+, 2}^{-1}\left(\frac{1}{nb_{n}^{2}}\sum_{i = 1}^{n}\delta_{+, i}R_{i, 2}\varepsilon_{i, w_j}^{+}\right) - \left(\hat{\gamma}_{+, j} - \hat{\gamma}_{-, j}\right)\mathbf{e}_{0}^{\top}\Gamma_{+, 1}^{-1}\eta_{+, w_j} + \op{1},
    \end{align*}
    where $\varepsilon_{i, w_j}^{+} = W_{i, j} - \mu_{+, w_j}(D_i - d^{*})$. Since $\hat{\gamma}_{\pm} \inprob \gamma_{\pm}$ by Proposition~\ref{prop:estimand} and $\mathbf{e}_{0}^{\top}\Gamma_{+, 1}^{-1}\eta_{+, w_j}$ is a constant, the difference of the two bias terms is $\left[\left(\gamma_{+, j} - \gamma_{-, j}\right) - \left(\hat{\gamma}_{+, j} - \hat{\gamma}_{-, j}\right)\right]\mathbf{e}_{0}^{\top}\Gamma_{+, 1}^{-1}\eta_{+, w_j} = \op{1}$.

    What remains are the main scores at $h_{n}$ and the bias-correction scores at $b_{n}$. For each bias-correction score we apply the identity
    \begin{align*}
        &\sqrt{nh_{n}}\,h_{n}^{2}\cdot\frac{1}{nb_{n}^{2}}\sum_{i = 1}^{n}\delta_{+, i}[\,\cdots] \\
        &\quad = \frac{h_{n}^{3}}{b_{n}^{3}}\cdot\frac{1}{\sqrt{nh_{n}}}\sum_{i = 1}^{n}\1{D_i \geq d^{*}}\K{\frac{D_i - d^{*}}{b_n}}[\,\cdots] = \xi^{3}\cdot\frac{1}{\sqrt{nh_{n}}}\sum_{i = 1}^{n}\1{D_i \geq d^{*}}\K{\frac{D_i - d^{*}}{b_n}}[\,\cdots],
    \end{align*}
    which holds because $\delta_{+, i} = \frac{1}{b_n}\1{D_i \geq d^{*}}\K{(D_i - d^{*})/b_n}$ and $\sqrt{nh_n}h_n^{2}\cdot\frac{\sqrt{nh_n}}{nb_n^{3}} = h_n^{3}/b_n^{3} = \xi^{3}$. Collecting the surviving terms,
    \begin{align*}
        \sqrt{nh_n}\left(\hat{\tau}_{\mathrm{pdd}}^{\mathrm{rbc}} - \tau \right) &= \hat{\mathbf{e}}_{+, w}^{\top}\hat{\Psi}_{+, 1}^{-1}\cdot\frac{1}{\sqrt{nh_n}}\sum_{i = 1}^{n}\1{D_i \geq d^{*}}\K{\tfrac{D_i - d^{*}}{h_n}}\bmat{R_{i, 1}\\Z_i}\varepsilon_{i, y}^{+}\\
        &\quad - \xi^{3}\hat{\mathbf{e}}_{+, w}^{\top}\hat{\Psi}_{+, 1}^{-1}\hat{\Omega}_{+}\mathbf{e}_{2}^{\top}\hat{\Psi}_{+, 2}^{-1}\cdot\frac{1}{\sqrt{nh_n}}\sum_{i = 1}^{n}\1{D_i \geq d^{*}}\K{\tfrac{D_i - d^{*}}{b_n}}\bmat{R_{i, 2}\\Z_i}\varepsilon_{i, y}^{+}\\
        &\quad + \sum_{j = 1}^{q}\left(\gamma_{+, j} - \gamma_{-, j}\right)\mathbf{e}_{0}^{\top}\hat{\Gamma}_{+, 1}^{-1}\cdot\frac{1}{\sqrt{nh_n}}\sum_{i = 1}^{n}\1{D_i \geq d^{*}}\K{\tfrac{D_i - d^{*}}{h_n}}R_{i, 1}\varepsilon_{i, w_j}^{+}\\
        &\quad - \sum_{j = 1}^{q}\xi^{3}\left(\hat{\gamma}_{+, j} - \hat{\gamma}_{-, j}\right)\mathbf{e}_{0}^{\top}\hat{\Gamma}_{+, 1}^{-1}\hat{\Lambda}_{+}\mathbf{e}_{2}^{\top}\hat{\Gamma}_{+, 2}^{-1}\cdot\frac{1}{\sqrt{nh_n}}\sum_{i = 1}^{n}\1{D_i \geq d^{*}}\K{\tfrac{D_i - d^{*}}{b_n}}R_{i, 2}\varepsilon_{i, w_j}^{+}\\
        &\quad - \hat{\mathbf{e}}_{+, w}^{\top}\hat{\Psi}_{-, 1}^{-1}\cdot\frac{1}{\sqrt{nh_n}}\sum_{i = 1}^{n}\1{D_i < d^{*}}\K{\tfrac{D_i - d^{*}}{h_n}}\bmat{R_{i, 1}\\Z_i}\varepsilon_{i, y}^{-}\\
        &\quad + \xi^{3}\hat{\mathbf{e}}_{+, w}^{\top}\hat{\Psi}_{-, 1}^{-1}\hat{\Omega}_{-}\mathbf{e}_{2}^{\top}\hat{\Psi}_{-, 2}^{-1}\cdot\frac{1}{\sqrt{nh_n}}\sum_{i = 1}^{n}\1{D_i < d^{*}}\K{\tfrac{D_i - d^{*}}{b_n}}\bmat{R_{i, 2}\\Z_i}\varepsilon_{i, y}^{-} + \op{1}.
    \end{align*}
    Each coefficient is the corresponding block of $\hat{\mathbf{s}}$ and each average the corresponding block of $\mathbf{X}_{n}$, so $\sqrt{nh_n}(\hat{\tau}_{\mathrm{pdd}}^{\mathrm{rbc}} - \tau) = \hat{\mathbf{s}}^{\top}\mathbf{X}_{n} + \op{1}$.
\end{proof}

\begin{proof}[Proof of proposition~\ref{prop:bias}]
    Using the decomposition in Equation~\eqref{eq:nu_decompose}, we can express,
    \begin{align*}
        \sqrt{nh_{n}}\left(\hat{\tau}_{\text{pdd}} - \tau_{\text{pdd}}\right) &=  \hat{\mathbf{e}}_{+, w}^{\top} \sqrt{nh_{n}} \left\{\hat{\nu}_{+} - \nu_{+}\right\} - \hat{\mathbf{e}}_{+, w}^{\top} \sqrt{nh_{n}} \left\{\hat{\nu}_{-} - \nu_{-} \right\} + (\gamma_{+} - \gamma_{-})^{\top} \sqrt{nh_{n}} \left(\hat{\beta}_{+, 0}^{w} - \beta_{+, 0}^{w} \right),
    \end{align*}
    where $\mathbf{e}_{+, w} = (1, 0, \lim_{\epsilon \to 0^{+}} \E{W_{i}^{\top} \mid D = d^{*} + \epsilon})^{\top} \in \mathbb{R}^{2 + q}$ and $\hat{\mathbf{e}}_{+, w} = \left(1, 0, \left(\hat{\beta}_{+, 0}^{w}\right)^{\top}\right)^{\top} \in \R^{2 + q}$.
    
    We focus on each of these three terms individually.
    \begin{enumerate}[label = \roman*]
        \item First, we show that,
        \begin{align*}
            \E{\hat{\Psi}_{+, 1}^{-1} \left( \frac{1}{n} \sum_{i = 1}^{n} \omega_{i, +} \bmat{R_{i, 1} \\ Z_{i}} (Y_{i} - g_{+}(D_{i} - d^{*}) - W_{i}^{\top} \gamma_{+})\right)} &\to 0.
        \end{align*}
        Using Lemma~\ref{lemma:factuals} (factuals), we have by the law of iterated expectations,
        \begin{align*}
            &\E{\left( \sqrt{\frac{h_{n}}{n}} \sum_{i = 1}^{n} \omega_{i, +} \bmat{R_{i, 1} \\ Z_{i}} (Y_{i} - g_{+}(D_{i} - d^{*}) - W_{i}^{\top} \gamma_{+})\right)}\\
            &\quad = \E{\left( \sqrt{\frac{h_{n}}{n}} \sum_{i = 1}^{n} \omega_{i, +} \bmat{R_{i, 1} \\ Z_{i}} \E{Y_{i} - g_{+}(D_{i} - d^{*}) - W_{i}^{\top} \gamma_{+} \mid D, Z}\right)}\\
            &\quad = 0.
        \end{align*}
        Therefore, we can conclude that,
        \begin{align*}
            &\E{ \hat{\Psi}_{+, 1}^{-1} \left( \sqrt{\frac{h_{n}}{n}} \sum_{i = 1}^{n} \omega_{i, +} \bmat{R_{i, 1} \\ Z_{i}} (Y_{i} - g_{+}(D_{i} - d^{*}) - W_{i}^{\top} \gamma_{+})\right)} \\
            &\quad = \E{ \Psi_{+, 1}^{-1} \left( \sqrt{\frac{h_{n}}{n}} \sum_{i = 1}^{n} \omega_{i, +} \bmat{R_{i, 1} \\ Z_{i}} (Y_{i} - g_{+}(D_{i} - d^{*}) - W_{i}^{\top} \gamma_{+})\right)} + \op{1}\\
            &\quad = \op{1},
        \end{align*}
        where formally, the second line uses that $\left( \sqrt{\frac{h_{n}}{n}} \sum_{i = 1}^{n} \omega_{i, +} \bmat{R_{i, 1} \\ Z_{i}} (Y_{i} - g_{+}(D_{i} - d^{*}) - W_{i}^{\top} \gamma_{+})\right) = O(1)$, which is shown in the proof of Theorem~\ref{thm:limit_dist}, and that $\hat{\Psi}_{+, 1}^{-1} \inprob \Psi_{+, 1}^{-1}$, as argued in the proof of Lemma~\ref{lem:IV_bias_converge}. Therefore, this term is asymptotically unbiased. 

        \item Next, we demonstrate that,  
        \begin{align*}
            \E{\mathbf{e}_{0}^{\top} \hat{\Gamma}_{+, 1}^{-1} \left(\sqrt{\frac{h_{n}}{n}} \sum_{i = 1}^{n} \omega_{+, i} R_{i, 1} \left\{W_{i, j} - \mu_{+, w_j}(D_{i} - d^{*})^{\top} \right\}\right)} \to 0.
        \end{align*}
    
        Using the law of iterated expectations, it follows that,
        \begin{align*}
            &\E{\left(\sqrt{\frac{h_{n}}{n}} \sum_{i = 1}^{n} \omega_{+, i} R_{i, 1} \left\{W_{i, j} - \mu_{+, w_j}(D_{i} - d^{*})^{\top} \right\}\right)}\\
            &\quad = \E{\left(\sqrt{\frac{h_{n}}{n}} \sum_{i = 1}^{n} \omega_{+, i} R_{i, 1} \E{W_{i, j} - \mu_{+, w_j}(D_{i} - d^{*})^{\top} \mid D_{i} = d}\right)}\\
            &\quad = 0,
        \end{align*}
        Therefore, we can conclude that,
        \begin{align*}
            &\E{ \mathbf{e}_{0}^{\top} \hat{\Gamma}_{+, 1}^{-1} \left(\sqrt{\frac{h_{n}}{n}} \sum_{i = 1}^{n} \omega_{+, i} R_{i, 1} \left\{W_{i}^{\top} - \mu_{+, w}(D_{i} - d^{*})^{\top} \right\}\right)} \\
            &\quad = \E{ \mathbf{e}_{0}^{\top} \Gamma_{+, 1}^{-1} \left(\sqrt{\frac{h_{n}}{n}} \sum_{i = 1}^{n} \omega_{+, i} R_{i, 1} \left\{W_{i}^{\top} - \mu_{+, w}(D_{i} - d^{*})^{\top} \right\}\right)} + \op{1}\\
            &\quad = \op{1},
        \end{align*}
         where formally, the second line uses that $\left(\sqrt{\frac{h_{n}}{n}} \sum_{i = 1}^{n} \omega_{+, i} R_{i, 1} \left\{W_{i}^{\top} - \mu_{+, w}(D_{i} - d^{*})^{\top} \right\}\right) = O(1)$, which is shown in the proof of Theorem~\ref{thm:limit_dist}, and that $\hat{\Psi}_{+, 1}^{-1} \inprob \Psi_{+, 1}^{-1}$, as argued in the proof of Lemma~\ref{lem:rdd_bias_converge}. Therefore, this term is asymptotically unbiased.  
        
        \item Here, we argue that $\E{\hat{\mathbf{e}}_{+, w}^{\top} \sqrt{nh_{n}} \left\{\hat{\nu}_{+} - \nu_{+}\right\}} \inprob \mathbf{e}_{+, w}^{\top} \Psi_{+, 1}^{-1} \eta_{+, \mathrm{IV}}$. Standard results from \citet[Section 4.1]{hahn2001identification} imply that $\hat{\mathbf{e}}_{+, w} \inprob \mathbf{e}_{+, w}$. Therefore, we can write that,
        \begin{align*}
            \E{\hat{\mathbf{e}}_{+, w}^{\top} \sqrt{nh_{n}} \left\{\hat{\nu}_{+} - \nu_{+}\right\}} &= \mathbf{e}_{+, w}^{\top} \E{\sqrt{nh_{n}} \left\{\hat{\nu}_{+} - \nu_{+}\right\}} + \op{1}\\
            &= \mathbf{e}_{+, w}^{\top} \Psi_{+, 1}^{-1}\eta_{+, \mathrm{IV}},
        \end{align*}
        where the last line uses Lemma~\ref{lem:IV_bias_converge} and our result from step (i).

        \item By symmetry, following the argument from the step above we can conclude that $\E{\hat{\mathbf{e}}_{+, w}^{\top} \sqrt{nh_{n}} \left\{\hat{\nu}_{-} - \nu_{-} \right\}} = \mathbf{e}_{w}^{\top} \Psi_{-, 1}^{-1}\eta_{-, \mathrm{IV}}$.

        \item Finally, we argue that $(\gamma_{+, j} - \gamma_{-, j}) \E{\sqrt{nh_{n}} \left(\hat{\beta}_{+, 0}^{w_j} - \beta_{+, 0}^{w_j} \right)} \inprob (\gamma_{+, j} - \gamma_{-, j}) \mathbf{e}_{0}^{\top} \Gamma_{+, 1}^{-1} \eta_{+, w_j}$. This follows directly from Lemma~\ref{lem:rdd_bias_converge} and our result from step (ii). 

        \item Writing,
        \begin{align*}
            \sqrt{nh_{n}}\left(\hat{\beta}_{+, 0}^{w} - \beta_{+, 0}^{w} \right) &= \pmat{\sqrt{nh_n}\left(\hat{\beta}_{+, 0}^{w_1} - \beta_{+, 0}^{w_1} \right) \\ \vdots \\ \sqrt{nh_n}\left(\hat{\beta}_{+, 0}^{w_q} - \beta_{+, 0}^{w_q} \right) },
        \end{align*}
        the result follows.
    \end{enumerate}
\end{proof}

\begin{proof}[Proof of Lemma~\ref{lem:pdd_variance}]
    We prove $\frac{1}{h_n}\E{X_{i, +}X_{i, +}^{\top}}\to\Sigma_{x, +}$; the result for $\Sigma_{x, -}$ follows by symmetry. Each entry is the limit of $\frac{1}{h_n}\E{\K{a}\K{b}\1{D_i\geq d^*}(\text{poly}_a)(\text{poly}_b)(\text{resid}_a)(\text{resid}_b)}$, where the kernels are localized at $h_n$ or $b_n$ according to the block. We compute three representative blocks; the remaining blocks follow by the same law-of-iterated-expectations and $u$-substitution argument, with the residual moments $\sigma^{2}_{+, y}(\cdot)$, $\rho_{+, y^{2}z}(\cdot)$, $\rho_{+, y^{2}zz^{\top}}(\cdot)$, $\sigma^{2}_{+, w_1}(\cdot)$, $\rho_{+, yw_1}(\cdot)$, $\rho_{+, yw_1z}(\cdot)$ in place of the generic residual term.

    For $\Sigma_{w_1, 11}^{+}$, both factors are localized at $h_n$. With $u = \tfrac{D_i - d^{*}}{h_n}$,
    \begin{align*}
        &\frac{1}{h_n}\E{\K{\tfrac{D_i - d^{*}}{h_n}}^{2}\1{D_i \geq d^{*}}R_{i, 1}R_{i, 1}^{\top}\varepsilon_{i, w_1}^{2}} \\
        &\quad = \int_{0}^{\infty}\K{u}^{2}\bmat{1 & u \\ u & u^{2}}\sigma^{2}_{+, w_1}(uh_n)f_{d}(uh_n + d^{*})\dv u \\
        &\quad \to f_{d, +}(d^{*})\sigma^{2}_{+, w_1}(0)\bmat{\gamma_{0, 2} & \gamma_{1, 2}\\ \gamma_{1, 2} & \gamma_{2, 2}}.
    \end{align*}

    For $\Sigma_{w_1, 22}^{+}$, both factors are localized at $b_n$. Substituting $u = \tfrac{D_i - d^{*}}{h_n}$ so that $\tfrac{D_i - d^{*}}{b_n} = \xi u$, and then $v = \xi u$,
    \begin{align*}
        \frac{1}{h_n}\E{\K{\tfrac{D_i - d^{*}}{b_n}}^{2}\1{D_i \geq d^{*}}R_{i, 2}R_{i, 2}^{\top}\varepsilon_{i, w_1}^{2}} &= \int_{0}^{\infty}\K{\xi u}^{2}(1, \xi u, (\xi u)^2)^{\otimes 2}\sigma^{2}_{+, w_1}(uh_n)f_{d}(uh_n + d^{*})\dv u\\
        &= \frac{1}{\xi}\int_{0}^{\infty}\K{v}^{2}(1, v, v^{2})^{\otimes 2}\sigma^{2}_{+, w_1}(\tfrac{v}{\xi}h_n)f_{d}(\tfrac{v}{\xi}h_n + d^{*})\dv v\\
        &\to \frac{1}{\xi}f_{d, +}(d^{*})\sigma^{2}_{+, w_1}(0)\bmat{\gamma_{0, 2} & \gamma_{1, 2} & \gamma_{2, 2} \\ \gamma_{1, 2} & \gamma_{2, 2} & \gamma_{3, 2} \\ \gamma_{2, 2} & \gamma_{3, 2} & \gamma_{4, 2}}.
    \end{align*}

    For $\Sigma_{w_1, 12}^{+}$, one factor is localized at $h_n$ and the other at $b_n$. With $u = \tfrac{D_i - d^{*}}{h_n}$ and $\tfrac{D_i - d^{*}}{b_n} = \xi u$, the $(a, b)$ entry, with $a \in \{0, 1\}$ from $R_{i, 1}$ and $b \in \{0, 1, 2\}$ from $R_{i, 2}$, is
    \begin{align*}
        \frac{1}{h_n}\E{\K{\tfrac{D_i - d^{*}}{h_n}}\K{\tfrac{D_i - d^{*}}{b_n}}\1{D_i \geq d^{*}}\left(\tfrac{D_i - d^{*}}{h_n}\right)^{a}\left(\tfrac{D_i - d^{*}}{b_n}\right)^{b}\varepsilon_{i, w_1}^{2}} &= \int_{0}^{\infty}\K{u}\K{\xi u}u^{a}(\xi u)^{b}\sigma^{2}_{+, w_1}(uh_n)f_{d}(uh_n + d^{*})\dv u\\
        &\to \xi^{b}\,\bar{\gamma}_{a + b, 2}(\xi)\,\sigma^{2}_{+, w_1}(0)f_{d, +}(d^{*}),
    \end{align*}
    which assembles into $\Sigma_{w_1, 12}^{+} = f_{d, +}(d^{*})\sigma^{2}_{+, w_1}(0)\left(\begin{smallmatrix}\bar{\gamma}_{0, 2} & \xi\bar{\gamma}_{1, 2} & \xi^{2}\bar{\gamma}_{2, 2}\\ \bar{\gamma}_{1, 2} & \xi\bar{\gamma}_{2, 2} & \xi^{2}\bar{\gamma}_{3, 2}\end{smallmatrix}\right)$. In every case the limit/integral swap uses the boundedness in Assumptions~\ref{assum:reg} and~\ref{assum:kernel}, and the existence of the limits uses the continuity in Assumption~\ref{assum:reg}. The IV and IV-$w$ blocks follow identically with the corresponding residual moments, giving the matrix in the lemma statement.
\end{proof}
\section{Details on partial linearity assumption}\label{sec:sim_description}

\subsection{Partial linearity}

In order for $\hat{\tau}_{\text{pdd}}^{\text{rbc}}$ to be consistent for the RDD treatment effect $\tau_{0}$, the confounding bridges $h_{+}(\cdot, \cdot)$ and $h_{-}(\cdot, \cdot)$ must be partially linear functions of the placebo outcome $W$. In other words, for some continuous, nonlinear functions $g_{+}(\cdot)$ and $g_{-}(\cdot)$, the confounding bridges must satisfy
$$
    h_{+}(d - d^{*}, w) = g_{+}(d - d^{*}) + w^{\top}\gamma_{+}, \qquad h_{-}(d - d^{*}, w) = g_{-}(d - d^{*}) + w^{\top}\gamma_{-}.
$$
We give sufficient conditions under which this restriction is satisfied.

For simplicity, suppose $W$ is a scalar; the result easily extends for multidimensional $W$.

\begin{lemma}[Sufficient conditions]\label{lemma:pl_sufficiency}
    Suppose that for some bounded, nonlinear function $f(\cdot)$ that is continuous in some neighborhood around $d^{*}$ (excluding $d^{*})$, the outcome satisfies
    \begin{align*}
        Y &= f(D) + \beta_{u} U + \eta_{y}, \qquad \E{\eta_{y} \mid D, Z, U} = 0. 
    \end{align*}
    Furthermore, suppose that the placebo outcome satisfies
    \begin{align*}
        W &= \alpha_{0} + \alpha_{u} U + \eta_{w}, \qquad \E{\eta_{w} \mid D, Z, U} = 0, ,\qquad \alpha_{u} \neq 0.
    \end{align*}
    Then, the confounding bridges $h_{+}(\cdot, \cdot)$ and $h_{-}(\cdot, \cdot)$ satisfy equation~\eqref{eq:partial}.
\end{lemma}

\begin{proof}[Proof of Lemma~\ref{lemma:pl_sufficiency}]\hfill

    By taking a conditional expectation and rearranging the equation for $W$,
    \begin{align*}
        U &= \frac{1}{\alpha_{u}} \left( \E{W \mid D, U} - \alpha_{0}\right).
    \end{align*}
    Then, taking a conditional expectation in the equation for $Y$,
    \begin{align*}
        \E{Y \mid D, U} &= f(D) + \beta_{u} U\\
        &= f(D) + \frac{\beta_{u}}{\alpha_{u}} \left(\E{W \mid D, U} - \alpha_{0}\right)\\
        &=f(D)  - \frac{\beta_{u}}{\alpha_{u}}\alpha_{0}+ \frac{\beta_{u}}{\alpha_{u}}\E{W \mid D, U}.
    \end{align*}
Consider the confounding bridge 
$$h_{0}(D - d^{*}, W)=g_{0}(D - d^{*}) + W \gamma_{0},\qquad g_{0}(D - d^{*})=f(D)  - \frac{\beta_{u}}{\alpha_{u}}\alpha_{0},\qquad  \gamma_0=\frac{\beta_{u}}{\alpha_{u}}.$$
    It is immediate that $\E{Y \mid D, U} = \E{h_{0}(D - d^{*}, W) \mid D, U}$, satisfying Assumption~\ref{assum:bridge}. 
\end{proof}



\end{document}